\documentclass[11pt]{article}
\usepackage{amsmath,amsfonts,amsthm,amssymb,color}
\usepackage{fancybox}
\usepackage{graphics}
\usepackage{graphicx}
\usepackage{tikz}
\usepackage{hyperref,url}

\usepackage{algorithm,algorithmic}
\usepackage{mathtools}

\usepackage[margin=1in]{geometry}

\newtheorem{theorem}{Theorem}[section]

\newtheorem{lemma}[theorem]{Lemma}

\newtheorem{claim}[theorem]{Claim}
\newtheorem{fact}[theorem]{Fact}

\theoremstyle{definition}
\newtheorem{definition}[theorem]{Definition}

\newenvironment{fminipage}%
  {\begin{Sbox}\begin{minipage}}%
  {\end{minipage}\end{Sbox}\fbox{\TheSbox}}

\def\pleq{\preccurlyeq}
\def\pgeq{\succcurlyeq}

\def\Span#1{\textbf{Span}\left(#1  \right)}

\def\defeq{\stackrel{\mathrm{def}}{=}}
\def\setof#1{\left\{#1  \right\}}
\def\sizeof#1{\left|#1  \right|}

\def\union{\cup}

\def\abs#1{\left|#1  \right|}

\def\norm#1{\left\| #1 \right\|}

\def\calA{\mathcal{A}}
\def\calC{\mathcal{C}}

\def\calG{\mathcal{G}}

\def\calM{\mathcal{M}}
\def\calS{\mathcal{S}}
\def\calT{\mathcal{T}}
\def\calV{\mathcal{V}}

\newcommand\bz{\mathbb Z}

\newcommand\rea{\mathbb R}

\newcommand\PPi{\boldsymbol{\Pi}}

\def\aa{\pmb{\mathit{a}}}
\newcommand\bb{\boldsymbol{\mathit{b}}}
\newcommand\cc{\boldsymbol{\mathit{c}}}
\newcommand\dd{\boldsymbol{\mathit{d}}}
\newcommand\ee{\boldsymbol{\mathit{e}}}
\newcommand\ff{\boldsymbol{\mathit{f}}}
\renewcommand\gg{\boldsymbol{\mathit{g}}}

\newcommand\pp{\boldsymbol{\mathit{p}}}
\newcommand\qq{\boldsymbol{\mathit{q}}}

\newcommand\rr{\boldsymbol{\mathit{r}}}
\renewcommand\ss{\boldsymbol{\mathit{s}}}
\def\tt{\boldsymbol{\mathit{t}}}
\newcommand\uu{\boldsymbol{\mathit{u}}}
\newcommand\vv{\boldsymbol{\mathit{v}}}
\newcommand\ww{\boldsymbol{\mathit{w}}}
\newcommand\yy{\boldsymbol{\mathit{y}}}
\newcommand\zz{\boldsymbol{\mathit{z}}}
\newcommand\xx{\boldsymbol{\mathit{x}}}

\renewcommand\AA{\boldsymbol{\mathit{A}}}
\newcommand\BB{\boldsymbol{\mathit{B}}}
\newcommand\CC{\boldsymbol{\mathit{C}}}
\newcommand\DD{\boldsymbol{\mathit{D}}}

\newcommand\GG{\boldsymbol{\mathit{G}}}
\newcommand\HH{\boldsymbol{\mathit{H}}}
\newcommand\II{\boldsymbol{\mathit{I}}}

\newcommand\NN{\boldsymbol{\mathit{N}}}
\newcommand\MM{\boldsymbol{\mathit{M}}}
\newcommand\LL{\boldsymbol{\mathit{L}}}
\newcommand\PP{\boldsymbol{\mathit{P}}}
\newcommand\QQ{\boldsymbol{\mathit{Q}}}

\renewcommand\SS{\boldsymbol{\mathit{S}}}
\newcommand\TT{\boldsymbol{\mathit{T}}}

\newcommand\WW{\boldsymbol{\mathit{W}}}

\newcommand\AAhat{\boldsymbol{\widehat{\mathit{A}}}}

\newcommand\BBhat{\boldsymbol{\widehat{\mathit{B}}}}

\newcommand\BBtil{\boldsymbol{\widetilde{\mathit{B}}}}

\newcommand\MMhat{\boldsymbol{\widehat{\mathit{M}}}}

\newcommand\WWhat{\boldsymbol{\widehat{\mathit{W}}}}

\newcommand\WWtil{\boldsymbol{\widetilde{\mathit{W}}}}

\newcommand\cctil{\boldsymbol{\tilde{\mathit{c}}}}

\newcommand\Otil{\widetilde{O}}

\newcommand{\one}{\mathbf{1}}
\newcommand{\zero}{\mathbf{0}}
\newcommand{\diag}{\textsc{diag}}

\DeclareMathOperator*{\argmin}{arg\,min}

\DeclareMathOperator*{\im}{im}

\newcommand{\eps}{\epsilon}

\newcommand{\LSD}{\textsc{lsd}}
\newcommand{\LSA}{\textsc{lsa}}

\newcommand{\YES}{\textsc{yes}}
\newcommand{\NO}{\textsc{no}}

\newcommand{\asoln}{approximate solution}

\newcommand{\proj}{\boldsymbol{\mathit{\Pi}}}

\newcommand{\trp}{\top}
\newcommand{\pinv}{\dagger}

\newcommand{\matclass}{\calM}
\newcommand{\algo}{\calA}

\newcommand{\CN}{\kappa} 

\newcommand{\orthog}{\bot}

\DeclareMathOperator*{\nnz}{nnz}

\DeclareMathOperator*{\nulls}{null}

\DeclareMathOperator*{\polylog}{polylog}

\DeclareMathOperator*{\poly}{poly}

\DeclareMathOperator*{\anzmin}{min_{+}}

\newcommand{\assign}{\leftarrow}

\newcommand{\genCl}{\calG}
\newcommand{\genZCl}{\calG_{\text{z}}}
\newcommand{\genZtwoCl}{\calG_{\text{z},2}}
\newcommand{\mctwoCl}{\calM\calC_{2}}
\newcommand{\mctwostrictCl}{{\calM\calC_{2}^{>0}}}
\newcommand{\mctwostrictintCl}{\calM\calC_{{2,\bz}}^{>0}}
\newcommand{\trusstwoCl}{\calT_{2}}
\newcommand{\tvtwoCl}{\calV_{2}}

\usepackage{placeins}

\title{
Hardness Results for Structured Linear Systems
}

\author{
  Rasmus Kyng\thanks{Supported by ONR Award N00014-16-1-2374.} \\
Yale University\\
rasmus.kyng@yale.edu
  \and
  Peng Zhang\thanks{Partially supported by the NSF under Grant No. 1637566.}\\
Georgia Institute of Technology\\
pzhang60@gatech.edu
}

\begin{document}

\begin{titlepage}
\maketitle

\thispagestyle{empty}

\begin{abstract}
We show that if the nearly-linear time solvers for Laplacian matrices and their generalizations can be extended to solve just slightly larger families of linear systems, then they can be used to quickly solve all systems of linear equations over the reals.  This result can be viewed either positively or negatively: either we will develop nearly-linear time algorithms for solving all systems of linear equations over the reals, or progress on the families we can solve in nearly-linear time will soon halt.
\end{abstract}
  
\end{titlepage}

\tableofcontents



\newcommand{\AAZtwo}{\AA^{\text{Z},2}}
\newcommand{\AAZ}{\AA^{\text{Z}}}
\newcommand{\BBtwostrict}{\BB^{>0}}
\newcommand{\BBtwostrictint}{\BB^{>0, \mathbb{Z}}}
\newcommand{\ccz}{\cc^{\text{Z}}}
\newcommand{\xxz}{\xx^{\text{Z}}}
\newcommand{\ccztwo}{\cc^{\text{Z},2}}
\newcommand{\xxztwo}{\xx^{\text{Z},2}}
\newcommand{\ddtwostrict}{\dd^{>0}}
\newcommand{\ddtwostrictint}{\dd^{>0, \mathbb{Z}}}

\section{Introduction}

We establish a dichotomy result for the families of linear equations that can be solved in nearly-linear time.
If nearly-linear time solvers exist for a slight generalization of the families for which they are currently known, 
then nearly-linear time solvers exist for all linear systems over the
reals.

This type of reduction is related to the successful research program of
fine-grained complexity, such as the result~\cite{WilW10} which
showed that the existence of a ``truly subcubic'' time algorithm for
All-Pairs Shortest Paths Problem
is equivalent to the existence of ``truly subcubic'' time
algorithm for a wide range of other problems.
For any
constant $a \geq 1$, our result establishes for 2-commodity matrices, and several other
classes of graph structured linear systems, that we can solve a
linear system in a matrix of this type with $s$ nonzeros in time $\Otil(s^a)$ if
and only if we can solve linear systems in all matrices with
polynomially bounded integer entries in time $\Otil(s^a)$.

In the RealRAM model, given a matrix $\AA \in \rea^{n \times n}$ and a
vector $\cc \in \rea^{n}$, we can solve the linear system $\AA \xx = \cc$
in $O(n^{\omega})$ time, where $\omega$ is the matrix multiplication
constant, for which the best currently known bound is $\omega <
2.3727$ \cite{Str69, Wil12}.
Such a running time bound is cost prohibitive for the large sparse
matrices often encountered in practice.
Iterative methods~\cite{Saad03:book}, first order methods~\cite{BoydV04:book},
and matrix sketches~\cite{Woodruff14} can all be viewed as ways of obtaining
significantly better performance in cases where the matrices have additional
structure.

\sloppy In contrast, when $\AA$ is an $n \times n$ Laplacian matrix with $m$ non-zeros,
and polynomially bounded entries, the linear system $\AA \xx = \cc$ can be
solved approximately to $\eps$-accuracy\ in $O((m+n) \log^{1/2+o(1)} n\log(1/\eps))$
time~\cite{SpielmanTengSolver:journal, CKMPPRX14}.
This result spurred a series of major developments in fast graph
algorithms, sometimes referred to as ``the Laplacian Paradigm''
of designing graph algorithms~\cite{Teng10:survey}.
The asymptotically fastest known algorithms for
Maximum Flow in directed unweighted graphs \cite{Mad13,Mad16},
Negative Weight Shortest Paths and Maximum Weight Matchings~\cite{CMSA17},
Minimum Cost Flows and Lossy Generalized Flows~\cite{LeeS14,DaitchS08}
all rely on fast Laplacian linear system solvers.

The core idea of the Laplacian paradigm can be viewed as showing that the
linear systems that arise from interior point algorithms, or second-order
optimization methods, have graph structure, and can be preconditioned and
solved using graph theoretic techniques.
These techniques could potentially be extended to a range of other problems,
provided fast solvers can be found for the corresponding linear systems.
Here a natural generalization is in terms of the number of labels per vertex:
graph Laplacians correspond to graph labeling problems where each vertex
has one label, and these labels interact pairwise via edges.
Multi-label variants of these exist in Markov random fields~\cite{SzeliskiZSVKATR08},
image processing~\cite{WangYYZ08}, Euclidean embedding of
graphs~\cite{CohenLMPS16},
 data processing for cryo-electron microscopy
  \cite{singer2011three,shkolnisky2012viewing,zhao2014rotationally}, phase retrieval \cite{alexeev2014phase,marchesini2014alternating},
  and many image processing problems
  (e.g. \cite{stable2015,arie2012global}).
Furthermore, linear systems with multiple labels per vertex arise when solving multi-commodity flow
problems using primal-dual methods.
Linear systems related to multi-variate labelings of graphs have been formulated
as the quadratically-coupled flow problem~\cite{KelnerMP12} and
Graph-Structured Block Matrices~\cite{Spi16}.
They also occur naturally in linear elasticity problems for simulating
the effect of forces on truss systems~\cite{DS07}.

Due to these connections, a central question in the Laplacian paradigm of designing
graph algorithms is whether all Graph-Structured Block Matrices can be solved in
(approximately) nearly-linear time. Even obtaining  subquadratic
running time would be constitute significant progress.
There has been some optimism in this direction due to the existence of faster
solvers for special cases:
nearly-linear time solvers for Connection Laplacians~\cite{KPSS16},
1-Laplacians of collapsible 3-D simplicial complexes~\cite{CohenFMNPW14},
and algorithms with runtime about $n^{5/4}$ for 2D planar truss
stiffness matrices~\cite{DS07}.
Furthermore, there exists a variety of faster algorithms for approximating
multi-commodity flows to $(1 + \epsilon)$
accuracy in time that scales as
$\poly(\eps^{-1})$~\cite{GargK96,Madry10a,Fleischer00,LeightonMPSTT91},
even obtaining nearly-linear running times when the graph is
undirected~\cite{Sherman13,KelnerLOS14,Peng16}.

The subquadratic variants of these routines also interact naturally
with tools that in turn utilize Laplacian solvers~\cite{KelnerMP12,LeightonMPSTT91}.
These existing tight algorithmic connections, as well as the solver for Connection
Laplacians, and the fact that combinatorial preconditioners partly originated
from speeding up interior point methods through preconditiong Hessians~\cite{Vaidya89},
together provide ample reason to hope that one could develop nearly-linear time solvers for linear systems
related to multicommodity flows.
Any algorithm that solves such systems to high accuracy in
$m^{1 + \alpha}$ time would in turn imply multicommodity flow algorithms that
run in about $n^{1/2} m^{1 + \alpha}$ time~\cite{LeeS14}, while the current
best running times are about $n^{2.5}$~\cite{LeeS15}.

Unfortunately, we show that if linear systems in general 2D truss stiffness matrices or
2-commodity Laplacians can be solved approximately in nearly-linear time, then all
linear systems in matrices with polynomially bounded integer entries
can be solved in nearly-linear time.
In fact, we show in a strong sense that any progress made in developing
solvers for these classes of matrices will translate directly into
similarly fast solvers for all matrices with polynomially bounded
integer entries. 
Thus developing faster algorithms for these systems will be
as difficult as solving \emph{all} linear systems faster.

Since linear system solvers used inside Interior Point Methods play a
central role in the Laplacian paradigm for designing high-accuracy algorithms, this
may suggest that in the high-accuracy regime the paradigm will not extend to most problems that
require multiple labels/variables per edge or vertex.
Alternatively, an algorithmic optimist might view our result as a
road-map for solving all linear systems via reductions to fast linear
system solvers for Graph-Structured Block Matrices.






\subsection{Our Results }

Fast linear system solvers for Laplacians, Connection Laplacians,
Directed Laplacians, and 2D Planar Truss Stiffness matrices are all based on
iterative methods and only produce approximate solutions.
The running time for these solvers scales logarithmically with
the error parameter $\eps$, i.e. as $\log(1/\eps)$.
Similarly, the running time for iterative methods usually depends on the condition number of
the matrix, but for state-of-the-art solvers for Laplacians, Connection Laplacians,
and Directed Laplacians, the dependence is logarithmic.
Consequently, a central open question is whether fast approximate
solvers exist for other structured linear systems, with running times
that depend logarithmically on the condition number and the accuracy parameter.

{\bf Integral linear systems are reducible to Graph-Structured
  Block Matrices.} 
Our reductions show that if fast approximate linear system solvers
exist for multi-commodity Laplacians, 2D Truss Stiffness, or Total
Variation (TV) Matrices, then 
fast approximate linear system solvers exist for any matrix, in the very
general sense of minimizing $\min_{\xx} \norm{\AA \xx
  - \cc}^2_2$. Thus our result also applies to singular matrices,
where we solve the pseudo-inverse problem to high
accuracy.
Theorem~\ref{thm:informalHardness} gives an informal statement of our
main result. The result is stated formally in
Section~\ref{sec:results} as Theorem~\ref{thm:runningTimeHardness}.

\begin{theorem}[Hardness for Graph-Structured Linear Systems
  (Informal)]
\label{thm:informalHardness}
  We consider three types of Graph-Structured Block Matrices: Multi-commodity Laplacians,
Truss Stiffness Matrices, and Total Variation Matrices.
Suppose that for one or more of these classes, the linear system $\AA \xx = \cc$ in a
  matrix $\AA$ with $s$
  non-zeros can be solved in time $\Otil(s^a)$, for some constant $a
  \geq 1$,
  with the running time
  having logarithmic dependence on condition number and
  accuracy\footnote{This is the kind of running time guarantee
    established for Laplacians, Directed Laplacians, Connection
    Laplacians, and bounded-weight planar 2D Truss Stiffness matrices.}.
Then linear systems in \emph{all matrices} with polynomially bounded integer entries and
condition number can be solved to high accuracy in time
$\Otil(s^a)$, where again $s$ is the number of non-zero entries of the
matrix.
\end{theorem}
Our results can easily be adapted to show that if
fast exact linear system solvers exist for multi-commodity Laplacians,
then exact solvers exist for all non-singular integer matrices.
However, this is of less interest since there is less
evidence that would suggest we should expect 
fast exact solvers to
exist.

The notion of approximation used throughout this paper is the same as
that used in the Laplacian solver literature (see
Section~\ref{sec:apxprelim}). 
To further justify the notion of approximate solutions to linear
systems, we show that it lets us solve a natural decision problem for
linear systems:

{\bf We show that deciding if a vector is approximately in the image
  of a matrix can be reduced to approximately solving linear systems.} 
We show this in Section~\ref{sec:lsd}.
We also show that the exact image decision problem
    requires working with exponentially small numbers, even when the
    input has polynomially bounded integral entries and 
  condition number. This means that in fixed-point arithmetic, we can
  only hope to solve an approximate version of the problem.
%
  The problem of approximately solving general linear systems
  can be reduced the problem of approximately solving Graph-Structured
  Block Matrix linear systems.
  Together, these statements imply that we can also reduce the problem of deciding
  whether a vector is approximately in the image of a general matrix to the problem
  of approximately solving Graph-Structured
  Block Matrix linear systems.


{ \bf We establish surprising separations between many
  problems known to have fast solutions and problems that are
  as hard solving general linear systems.} 
Our results trace out several interesting dichotomies: restricted
cases of 2D truss stiffness matrices have fast solvers, but
fast solvers for all 2D truss stiffness matrices would imply equally fast
solvers for all linear systems.
TV matrices can be solved quickly in the anisotropic case, but in the
isotropic case imply solvers for all linear systems.
Fast algorithms exist for multi-commodity problems in the low accuracy
regime, but existing approaches for the high accuracy regime seem to
require fast solvers for multi-commodity linear systems, which again
would imply fast solvers for all linear systems.

Our reductions only require the simplest cases of the classes
we consider: 2-Commodity Laplacians are sufficient, as are (non-planar) 2D Truss
Stiffness matrices, and Total Variation Matrices with 2-by-2
interactions.
Linear systems of these three classes have many applications, and faster solvers
for these would be useful in all applications. Trusses have been studied as the canonical multi-variate problem,
involving definitions such as Fretsaw extensions~\cite{ShklarskiT08}
and factor widths~\cite{BomanCPT05}, and fast linear system solvers
exist for the planar 2D case with bounded weights~\cite{DS07}.
Total Variation Matrices are widely used in image denoising~\cite{ChanS05:book}.
The anisotropic version can be solved using nearly-linear time linear system solvers~\cite{KoutisMT11},
while the isotropic version has often been studied using linear
systems for which fast solvers are not
known~\cite{GoldfarbY04,WohlbergRodriguez07,ChinMMP13}.
Multi-commodity flow problems have been the subject of extensive
study, with significant progress on algorithms with low
accuracy~\cite{GargK96,Madry10a,Fleischer00,LeightonMPSTT91,
  Sherman13,KelnerLOS14,Peng16}, while high accuracy approaches use
slow general linear system solvers.

\subsection{Approximately Solving Linear Systems and Normal Equations}

The simplest notion of solving a linear system $\AA \xx =\cc$, is to
seek an $\xx$ s.t. the equations are exactly satisfied.
More generally, if the system is not guaranteed to have a solution, we can
ask for an $\xx$ which minimizes $\norm{\AA \xx - \cc}^2_2$.
An $\xx$ which minimizes this always exists. In general, it
may not be unique. 
Finding an $\xx$ which minimizes $\norm{\AA \xx - \cc}^2_2$ is
equivalent to solving the linear system ${\AA^{\trp} \AA \xx
=\AA^{\trp} \cc}$, which is referred to as the normal equation for
the linear system $\AA \xx = \cc$ (see \cite{trefethen1997numerical}).
The problem of solving the normal equations (or equivalently, minimizing
$\norm{\AA \xx - \cc}^2_2$), is a generalization of the problem of
linear system solving, since the approach works when $\AA$ is
non-singular, while also giving meaningful results when $\AA$ is
singular.
The normal equation problem can also be understood in terms of the
Moore-Penrose pseudo-inverse of a matrix $\MM$, which is denoted
$\MM^{\pinv}$ as ${\xx =(\AA^{\trp} \AA)^{\pinv} \AA^{\trp} \cc}$ is a
solution to the normal equations.
Taking the view of linear system solving as minimizing $\norm{\AA \xx
  - \cc}^2_2$ also gives sensible ways to define an approximate
solution to a linear system: It is an $\xx$ that ensures $\norm{\AA \xx
  - \cc}^2_2$ is close to $\min_{\xx} \norm{\AA \xx
  - \cc}^2_2$.
In Section~\ref{sec:prelim}, we formally define several notions of
approximate solutions to linear systems that we will use throughout
the paper.

An important special case of linear systems is when the matrix of
coefficients of the system is positive semi-definite. 
Since $\AA^{\trp} \AA$ is always positive semi-definite, solving the
normal equations for a linear system falls into this case.
Linear systems over positive semi-definite matrices can be solved
(approximately) by approaches known as iterative methods, which often
lead to much faster algoritms than the approaches used for general
linear systems.
Iterative methods inherently produce approximate
solutions\footnote{A seeming counterexample to this is Conjugate
  Gradient which is an iterative method
  that produces exact solutions in the RealRAM model. But it requires
  extremely high precision calculations to exhibit this behaviour in
  finite precision arithmetic, and so Conjugate Gradient is also best
  understood as an approximate method.}.



\subsection{Graph-Structured Block Matrices}

Graph-Structured Block Matrices are a type of linear system that
arise in many applications.
Laplacian matrices and Connection Laplacians both fall in this
category.

Suppose we have a collection of $n$ disjoint sets of variables $X_{1}, \ldots, X_{n}$, with each set
having the same size, $\sizeof{X_{i}} = d$.  
Let $\xx^{i}$ denote the vector\footnote{We use superscripts to index a sequence of vectors or matrices, and we use subscripts to denote entries of a vector or matrix, see Section~\ref{sec:prelim}.}
 of variables in $X_{i}$,
and consider an equation of the form $\SS \xx^{i} - \TT \xx^{j} = {\bf 0}$, 
where $\SS$ and $\TT$ are both $r \times d$ matrices. 
Now we form a linear system $\BB \xx = \zero$ by stacking $m$ equations of
the form given above as the rows of the system.
Note that, very importantly, we allow a different choice of $\SS$ and
$\TT$ for every pair of $i$ and $j$.
This matrix $\BB \in \rea^{mr \times nd}$ we refer to as a
Incidence-Structured Block Matrix (ISBM), while we refer to $\BB^{\trp}
\BB$ as a Graph-Structured Block Matrix (GSBM).
Note that $\BB$ is not usually PSD, but $\BB^{\trp} \BB$ is.
The number of non-zeros in $\BB^{\trp} \BB$ is $O(m d^2)$.
GSBMs come up in many applications, where we typically want to solve a
linear system in the normal equations of $\BB$.

Laplacian matrices are GSBMs where $d = 1$ and $\SS = \TT =w$,
where $w$ is a real number, and we allow different $w$ for each pair of
$i$ and $j$.  The corresponding ISBM for Laplacians is called an
edge-vertex incidence matrix.
Connection Laplacians are GSBMs where $d = O(1)$ and $\SS =
\TT^{\trp} = w\QQ$, for some rotation matrix $\QQ$ and
a real number $w$.
Again, we allow a different rotation matrix and scaling for every edge.
For both Laplacians and Connection Laplacians, there exist linear
system solvers that run in time $O( m \polylog(n,\eps^{-1}))$ and
produce $\eps$ approximate solutions to the corresponding normal
equations.

We now introduce several classes of ISBMs and their associated
GSBMs. Our Theorem~\ref{thm:runningTimeHardness} shows that fast
linear system solvers for any of these classes would imply fast linear
system solvers for all matrices with polynomially bounded entries and
condition number.
\begin{definition}[2-commodity incidence matrix]
A \emph{2-commodity incidence matrix} is an ISBM where $d = 2$ and $r = 1$,
and $\SS = \TT $, and we allow three types of $\SS$:
$\SS=
w
\begin{pmatrix}
  1 & 0
\end{pmatrix}
$,
$\SS=
w
\begin{pmatrix}
  0 &1
\end{pmatrix}
$
and
$\SS=
w \begin{pmatrix}
  1 &- 1
\end{pmatrix}
$, where in each case $w$ is a real number which may depend on the pair
$i$ and $j$. 
We denote the set of all 2-commodity incidence matrices by $\mctwoCl$. 
The corresponding GSBM is called a 2-commodity Laplacian.
The ISBM definition is equivalent to requiring the GSBM to
have the form
\[
\LL^1 \otimes \left( \begin{array}{cc}
1 & 0 \\
0 & 0
\end{array} \right)
+ \LL^2 \otimes \left( \begin{array}{cc}
0 & 0 \\
0 & 1
\end{array} \right)
+ \LL^{1+2} \otimes \left( \begin{array}{cc}
1 & -1 \\
-1 & 1
\end{array} \right)
\]
where $\otimes$ is the tensor product and  $\LL^1$, $\LL^2$, and
$\LL^{1+2}$ are all Laplacian matrices.

We adopt a convention that  the first variable in a set $X_i$ is
labelled $\uu_i$ and the second is labelled $\vv_i$.
Using this convention, given a 2-commodity incidence matrix $\BB$,
the equation $\BB \xx = \zero$ must consist of 
scalings of the following three types of equations: 
$\uu_i - \uu_j= 0$, $\vv_i - \vv_j = 0$, and $\uu_i - \vv_i - (\uu_j - \vv_j) = 0$.
\label{def:MC2}
\end{definition}

\begin{definition}[Strict 2-commodity incidence matrix]
\label{def:MC2Strict}
A \emph{strict 2-commodity incidence matrix} is a 2-commodity
incidence matrix where the corresponding 2-commodity Laplacian has the
property that $\LL^1$, $\LL^2$, and
$\LL^{1+2}$ all have the same non-zero pattern.
We denote the set of all strict 2-commodity incidence matrices by
$\mctwostrictCl$. 
We denote the set of all strict 2-commodity incidence matrices with
\emph{integer entries} by $\mctwostrictintCl$.
\end{definition}
Linear systems in $\mctwostrictCl$ are exactly the systems that one
has to solve to when solving 2-commodity problems using Interior Point
Methods (IPMs).
For readers unfamiliar with 2-commodity problems or IPMs, we provide a brief explanation of
why this is the case in Section~\ref{sec:ipm}.
The $\mctwostrictCl$ is more restrictive than $\mctwoCl$, 
and $\mctwostrictintCl$ in turn is even more restrictive. 
One could hope that fast linear system solvers exist for
$\mctwostrictCl$ or $\mctwostrictintCl$, even
if they do not exist for $\mctwoCl$.
However, our reductions show that even getting a fast approximate solver for
$\mctwostrictintCl$ with polynomially bounded entries and
condition number will lead to a fast solver for all matrices with polynomially bounded entries and
condition number.

The next class we consider is 2D Truss Stiffness Matrices. They have
been studied extensively in the numerical linear algebra
community~\cite{ShklarskiT08,BomanCPT05}. 
For \emph{Planar} 2D Trusses with some bounds on ranges of edges, Daitch
and Spielman obtained linear system solvers that run in time
$\Otil(n^{5/4} \log(1/\eps))$.

\begin{definition}[2D Truss Incidence Matrices]
\label{def:trusses}
Let $G = (V, E)$ be a graph whose vertices are $n$ points in
2-dimension: $\ss^{1}, \ldots, \ss^{n} \in \mathbb{R}^2$.
Consider $X_{1}, \ldots, X_{n}$ where $d = 2$.
A \emph{2D Truss Incidence Matrix} is an ISBM where $d = 2$ and $r = 1$,
and for each $i$ and $j$,
 we have $\SS = \TT$ and
$\SS=
w
(\ss^{i} - \ss^{j})^{\trp}
$,
and $w$ is a real number that may depend on the pair $i$ and $j$, but
$\ss^{i}$ depends only on $i$ and vice versa for $\ss^{j}$.
We denote the class of all 2D Truss Incidence Matrices by $\trusstwoCl$. 
\end{definition}
Another important class of matrices is Total Variation Matrices (TV matrices).
TV matrices come from Interior Point Methods for solving total
variation minimization problem in image, see for example~\cite{GY05} and~\cite{CMMP13}.
Not all TV matrices are GSBMs, but many GSBMs can be
expressed as TV matrices.

\begin{definition}[TV matrix and 2-TV Incidence Matrices]
\label{def:TV}
Let $E_1 \cup \ldots \cup E_s$ be a partition of the edge set of a graph. 
For each $1 \le i \le s$, 
let $\BB^i$ be the edge-vertex incidence matrix of $E_i$,
$\WW^i$ be a diagonal matrix of edge weights, and $\rr^{i}$ be a vector satisfying $\WW^i\pgeq \rr^{i} (\rr^{i})^{\trp}$.
Given  these objects, the associated \emph{total variation matrix} (TV matrix) is a matrix $\MM$ defined as
\begin{align*}
\MM = \sum_{1 \le i \le s} (\BB^i)^{\trp} \left( \WW^i - \rr^{i} (\rr^{i})^{\trp}
  \right) \BB^i
.
\end{align*}
A \emph{2-TV Incidence Matrix} is defined as any ISBM whose
corresponding GSBM is a
TV matrix with $\WW^i\in \rea^{2 \times 2}$ and $\rr^{i} \in \rea^{2}$.
We denote the class of all 2-TV incidence matrices by $\tvtwoCl$.
\end{definition}

\subsection{Our Reduction: Discussion and an Example}
%
%
In this section we give a brief sketch of the ideas behind our
reduction from general linear systems, over matrices in $\genCl$, to
multi-commodity linear systems, over matrices in $\mctwoCl$, and we demonstrate the most important
transformation 
 through an example.

The starting point for our approach is the folklore idea that any
linear system can be written as a factor-width 3 system by introducing
a small number of extra variables.
Using a set of multi-commodity constraints, we are able to express one
particular factor-width 3 equation, namely $ 2x'' =  x + x' $.
After a sequence of preprocessing steps, we are then able to
efficiently express arbitrary linear systems over integer matrices
using constraints of this form.
A number of further issues arise when the initial matrix does not have
full column rank, requiring careful weighting of the constraints we introduce.

Given a matrix $\AA$ with polynomially bounded integer entries and condition number,
we reduce the linear system
$\AA \xx = \cc$ to a linear system $\BB \yy = \dd$,
where $\BB$ is a
strict multi-commodity edge-vertex incidence matrix 
with integer entries (i.e. in $\mctwostrictintCl$),
with polynomially bounded entries and condition number.
More precisely, we reduce $\AA^{\trp}\AA \xx = \AA^{\trp}\cc$ to
$\BB^{\trp}\BB \yy = \BB^{\trp}\dd$. These systems always have a solution.
We show that we can find an $\eps$-approximate solution to the linear
system $\AA^{\trp}\AA \xx = \AA^{\trp}\cc$ by a simple mapping on any 
$\yy$ that $\eps'$-approximately solves the linear system
$\BB^{\trp}\BB \yy = \BB^{\trp}\dd$, where $\eps'$ is only
polynomially smaller than $\eps$.
If $\AA$ has $s$ non-zero entries and the maximum absolute value of an
entry in $\AA$ is $U$, then $\BB$ will have $O(s\log(sU))$ non-zero
entries and our algorithm computes the reduction in time $O(s\log(sU))$.
Note that $\BB^{\trp}\BB $ has $r = O(s\log(sU))$ non-zeros, because
every row of $\BB$ has $O(1)$ entries. 
All together, this means that getting a solver for $\BB^{\trp}\BB \xx
= \BB^{\trp}\dd$ with running time $\Otil(r^{a} \log(1/\eps))$ will give
a solver for $\AA$ with $\Otil(s^{a} \log(1/\eps))$ running time.

We achieve this through a chain of reductions. Each reduction produces a new
matrix and vector, as well as a new error parameter giving the
accuracy required in the new system to achieve the accuracy desired in
the original system.
\begin{enumerate}
\item We get a new linear system $\AAZtwo \xxztwo = \ccztwo$ where 
$\AAZtwo$ has integer entries, and the entries of each row of $\AAZtwo$ sum to zero,
i.e. $\AAZtwo \one = \zero$, and finally in every row the sum of the
positive coefficients is a power of two.
\item
\label{enm:GZTtoMCT}
$\AAZtwo \xxztwo = \ccztwo$ is then transformed to $\BB \yy = \dd$, 
where $\BB$ is a 2-commodity edge-vertex incidence
matrix.
\item $\BB \yy = \dd$ is then transformed to $\BBtwostrict \yy = \ddtwostrict$, 
where $\BBtwostrict$ is a strict 2-commodity edge-vertex incidence
matrix.
\item $\BBtwostrict \yy = \ddtwostrict$ is then transformed to $\BBtwostrictint \yy = \ddtwostrictint$, 
where $\BBtwostrictint$ is a 2-commodity edge-vertex incidence
matrix with integer entries.
\end{enumerate}
We will demonstrate step~\ref{enm:GZTtoMCT}, the main
transformation,  by example. 
When the context is clear, we drop the superscripts of matrices for simplicity.
The reduction handles each row (i.e. equation) of the linear system
independently, so we focus on the reduction for a single row.

Consider a linear system $\AA \xx =
\cc$, and let us pick a single row (i.e. equation) $\AA_{i}  \xx =
\cc_{i}$ \footnote{We use $\AA_i$ to denote the $i$th row of $\AA$, and $\cc_i$ to denote the $i$th entry of $\cc$, see Section~\ref{sec:prelim}.}.
We will repeatedly pick pairs of existing variables of $\xx$, say $x$ and $x'$, based on their
current coefficients in $\AA_{i}  \xx = \cc_{i}$, and modify the row
by adding $C (2x'' - (x + x') )$ to the left hand side where $x''$ is a new variable and
$C$ is a real number we pick.
As we will see in a moment, we can use this pair-and-replace operation
to simplify the row until it eventually becomes a 2-commodity
equation.
At the same time as we modify $\AA_{i}$, we also store an auxiliary equation $C (x + x' -
2x'' ) = 0$.
Suppose initially that $\AA_{i}  \xx = \cc_{i}$ is satisfied.
After this modification of $\AA_{i}  \xx = \cc_{i}$, if the
auxiliary equation is satisfied, $\AA_{i}  \xx = \cc_{i}$ is 
still satisfied by the same values of $x$ and $x'$.
Crucially, we can express the auxiliary equation $C (x + x' -
2x'' ) = 0$ by a set of ten 2-commodity equations, i.e. a
``2-commodity gadget'' for this equation.
Our final output matrix \emph{will not} contain the equation $C (x + x' -
2x'' ) = 0$ as a row, but will instead contain 10 rows of 2-commodity
equations from our gadget construction.
Eventually, our pair-and-replace scheme will also transform the row
$\AA_{i}  \xx = \cc_{i}$ into a 2-commodity equation on just two variables.

Next, we need to understand how the pair-and-replace scheme makes
progress.
The pairing handles the positive and the negative coefficients of
$\AA_{i}$ separately, and eventually ensures that $\AA_{i}  \xx =
\cc_{i}$ has only a single positive and a single negative coefficient
in the modified row $\AA_{i}  \xx = \cc_{i}$, in particular it is of the form $a x
- a x' = \cc_{i}$ for two variables $x$ and $x'$ that appear in the
modified vector of variables $\xx$, i.e. it is a 2-commodity equation.

To understand the pairing scheme, it is helpful to think about the
entries of $\AA$ written using binary (ignoring the sign of the entry).
The pairing scheme proceeds in a sequence of rounds: In the first
round we pair variables whose 1st (smallest) bit is 1. There must be
an even number of variables with smallest bit 1, as the sum of the positive
(and respectively negative) coefficients is a power of 2.
We then replace the terms
corresponding to the  1st bit of the pair with a new single variable
with a coefficient of 2. 
After the first round, every coefficient has zero in the 1st bit.
In the next round, we pair variables whose 2nd bit is 1, and replace the terms
corresponding to the the 2nd bit of the pair with a new single variable
with a coefficient of 4, and so on.
Because the positive coefficients sum to a power of two, we are able
to guarantee that pairing is always possible.
It is not too hard to show that we do not create a large number of new
variables or equations using this scheme.


For example, let us consider an equation
\begin{alignat*}{3}
  && 3 \xx_{1} + 5 \xx_{2} + 4 \xx_{3} + 4 \xx_{4} - 16 \xx_{5} &= 1 &
\\
\ArrowBetweenLines[\downarrow]
&&&&\text{ Replace } \xx_{1} + \xx_{2} \text{ by } 2 \xx_{6}.  
\\
&&&&
\text{ Add auxiliary equation } \xx_{1} + \xx_{2} - 2 \xx_{6} = 0.
\\
&& 2 ( \xx_{1} + 2 \xx_{2} + \xx_{6} + 2 \xx_{3} + 2 \xx_{4} ) - 16 \xx_{5} &= 1 
\\
\ArrowBetweenLines[\downarrow]
&&&&\text{ Replace } 2(\xx_{1} + \xx_{6}) \text{ by } 4 \xx_{7}.  
\\
&&&&
\text{ Add auxiliary equation } 2 (\xx_{1} + \xx_{6} - 2 \xx_{7}) = 0.
\\
&& 4 ( \xx_{2} + \xx_{7} +  \xx_{3} + \xx_{4} ) - 16 \xx_{5} &= 1 &
\\
\ArrowBetweenLines[\downarrow]
&&&&\text{ Replace } 4(\xx_{2} + \xx_{7}) \text{ by } 8 \xx_{8},  
\\
&&&&\text{ and } 4(\xx_{3} + \xx_{4}) \text{ by } 8 \xx_{9}.  
\\
&&&&
\text{ Add auxiliary equations } 4 (\xx_{2} + \xx_{7} - 2 \xx_{8}) = 0,
\\
&&&&
\text{                                  and } 4 (\xx_{3} + \xx_{4} - 2 \xx_{9}) = 0.
\\
&& 8 ( \xx_{8} + \xx_{9} ) - 16 \xx_{5} &= 1 &
\\
\ArrowBetweenLines[\downarrow]
&&&&\text{ Replace } 8 (\xx_{8} + \xx_{9}) \text{ by } 16 \xx_{10}.  
\\
&&&&
\text{ Add auxiliary equation } 8 ( \xx_{8} + \xx_{9} - 2 \xx_{10} ) = 0.
\\
&& 16 \xx_{10} - 16 \xx_{5} &= 1 &
\end{alignat*}

In this way, we process $\AA \xx = \cc$ to produce a new
set of equations $\BB \yy = \dd$ where $\BB$ is a 2-commodity matrix.
If $\AA \xx = \cc$ has an exact solution, this solution can be obtained
directly from an exact solution to $\BB \yy = \dd$.
We also show that an approximate solution to $\BB \yy = \dd$ leads to
an approximate solution for $\AA \xx = \cc$, and we show that $\BB$
does not have much larger entries or condition number than $\AA$.

The situation is more difficult when $\AA \xx = \cc$ does not have a
solution and we want to obtain an approximate minimizer
$\argmin_{\xx \in \rea^{n}} \norm{ \AA \xx - \cc}_2^{2}$ from an
approximate solution to $\argmin_{\yy \in \rea^{n'}} \norm{ \BB \yy - \dd}^{2}_2$.
This corresponds to approximately applying the Moore-Penrose pseudo-inverse of $\AA$
to $\cc$.
We deal with the issues that arise here using a carefully chosen
scaling of each auxiliary constraint to ensure a strong relationship
between different solutions.

In order to switch from a linear system in a general 2-commodity
matrix to a linear system in a \emph{strict} 2-commodity matrix, we
need to reason very carefully about the changes to the null space that
this transformation inherently produces.
By choosing sufficiently small weights, we are nonetheless able to establish a strong relationship between
the normal equation solutions despite the change to the null space.




\section{Preliminaries}
\label{sec:prelim}

%

We use subscripts to denote entries of a matrix or a vector:
let $\AA_i$ denote the $i$th row of matrix $\AA$ and $\AA_{ij}$ denote the $(i,j)$th entry of $\AA$;
let $\xx_i$ denote the $i$th entry of vector $\xx$ and $\xx_{i:j}$ ($i < j$) denote the vector of entries 
$\xx_i, \xx_{i+1}, \ldots, \xx_j$.
We use superscripts to index a sequence of matrices or vectors, e.g., $\AA^1, \AA^2, \ldots,$ and $\xx^1, \xx^2, \ldots$, except when some other meaning is clearly stated.

We use $\AA^\pinv$ to denote the Moore-Penrose pseudo-inverse of a
matrix $\AA$.
We use $\im(\AA)$ to denote the image of a matrix $\AA$.
We use $\norm{\cdot}_2$ to denote the Euclidean norm on vectors and the
spectral norm on matrices. 
When $\MM$ is an $n \times n$ positive semidefinite matrix, 
we define a norm on vectors $\xx \in \rea^{n}$ by
 $\norm{\xx}_{\MM} \defeq \sqrt{\xx^{\trp} \MM \xx }$.
We let $\nnz(\AA)$ denote the number of non-zero entries in a matrix $\AA$.
We define $\norm{\AA}_{\max} = \max_{i,j} \abs{\AA_{ij}} $, $\norm{\AA}_1 = \max_{j} \sum_i \abs{\AA_{ij}}$ and $\norm{\AA}_{\infty} = \max_i \sum_j \abs{\AA_{ij}}$.
We let
$
\anzmin(\AA) = \min_{i,j \text{ s.t. } \AA_{ij} \neq 0} \abs{\AA_{ij}} 
$.
Given a matrix $\AA \in \rea^{m \times n}$ and a vector $\cc \in
\rea^{m}$ for some $m,n$, we call the tuple $(\AA,\cc)$ a linear
system.
Given matrix $\AA \in \rea^{m \times n}$, 
let $\proj_{\AA} \defeq \AA (\AA \AA^{\trp})^{\pinv} \AA^{\trp} $,
i.e. the orthogonal projection onto $\im(\AA)$. 
Note that $\proj_{\AA} = \proj_{\AA}^{\trp}$ and $\proj_{\AA} = \proj_{\AA}^{2}$.

\subsection{Approximately Solving A Linear System}
\label{sec:apxprelim}
In this section we formally define the notions of approximate
solutions to linear systems that we work with throughout this paper.

\begin{definition}[Linear System Approximation Problem, \LSA]
\label{def:lsapx}
Given linear system $(\AA,\cc)$, 
where $\AA \in \rea^{m \times n}$,
and $\cc \in \rea^{m}$, and given a scalar $0 \leq \eps \leq 1$, 
we refer to the $\LSA$ problem for the triple
$(\AA,\cc,\eps)$ as the problem of finding 
$\xx \in \rea^{n}$ s.t.
\[
  \norm{ \AA \xx - \proj_{\AA} \cc }_2 \leq \eps \norm{ \proj_{\AA} \cc
  }_2
,
\]
and we say that such an $\xx$ is a solution to the $\LSA$ instance $(\AA,\cc,\eps)$.
\end{definition}
This definition of a {\LSA} instance and solution has several advantages: when
$\im(\AA) =\rea^{m}$, we get $\proj_{\AA} = \II$, and it reduces to the
natural condition ${\norm{ \AA \xx - \cc }_2 \leq \eps \norm{ \cc }}_2$,
which because $\im(\AA) =\rea^{m}$, can be satisfied for any
$\epsilon$, and for $\eps = 0$ tells us that $ \AA \xx = \cc $.

When $\im(\AA)$ does not include all of $\rea^{m}$, the vector
$\proj_{\AA} \cc$ is exactly the projection of $\cc$ onto $\im(\AA)$,
and so a solution can still be obtained for any $\epsilon$.
Further, as $(\II - \proj_{\AA}) \cc$ is orthogonal to $\proj_{\AA}
\cc$ and $\AA \xx$, it follows that 
\[
\norm{ \AA \xx - \cc}^{2}_2
=
\norm{(\II - \proj_{\AA}) \cc}^{2}_2
+
\norm{ \AA \xx - \proj_{\AA} \cc }^{2}_2
.
\]
Thus, when $\xx$ is a solution to the
{\LSA} instance $(\AA,\cc,\eps)$, then $\xx$ also gives an $\eps^{2} \norm{
  \proj_{\AA} \cc }^{2}_2$ additive approximation to 
\begin{equation}
\min_{\xx \in \rea^{n}} \norm{ \AA \xx - \cc}^{2}_2  = \norm{(\II - \proj_{\AA}) \cc}^{2}_2
.\label{eqn:minquad}
\end{equation}
Similarly, an $\xx$ which gives an additive $\eps^{2} \norm{\proj_{\AA}
  \cc }^{2}_2$ approximation to Problem~\eqref{eqn:minquad} is always a
solution to the {\LSA} instance
$(\AA,\cc,\eps)$.
These observations prove the following (well-known) fact:

\begin{fact}
\label{fac:projIsQFmin} Let
$
\xx^{*} \in \argmin_{\xx \in \rea^{m}} \norm{ \AA \xx - \cc}^{2}_2
$,
then for every $\xx$,
\[
\norm{ \AA \xx - \cc}^{2}_2 \leq
\norm{ \AA\xx^{*} - \cc}^{2}_2 + \eps^{2} \norm{\proj_{\AA} \cc }^{2}_2
\]
if and only if $\xx$ is a solution to
the {\LSA} instance $(\AA,\cc,\eps)$.
\end{fact}

When the linear system $\AA \xx = \cc$ does not have a solution, a
natural notion of solution is any minimizer of
Problem~\eqref{eqn:minquad}.
A simple calculation shows that this is equivalent to requiring that
$\xx$ is a solution to the linear system $\AA^{\trp} \AA \xx =
\AA^{\trp} \cc$, which always has a solution even when $\AA \xx = \cc$
does not.
The system $\AA^{\trp} \AA \xx = \AA^{\trp} \cc$ is referred to as the
\emph{normal equation} associated with $\AA \xx = \cc$ (see \cite{trefethen1997numerical}).
\begin{fact}
  $
\xx^{*} \in \argmin_{\xx \in \rea^{n}} \norm{ \AA \xx - \cc}^{2}_2
$, if and only if $\AA^{\trp} \AA \xx^{*} = \AA^{\trp} \cc$, and this
linear system always has a solution.
\end{fact}


This leads to a natural question: Suppose we want to approximately
solve the linear system $\AA^{\trp} \AA \xx = \AA^{\trp} \cc$.
Can we choose our notion of approximation to be equivalent to that of
a solution to the {\LSA} instance $(\AA,\cc,\eps)$?

A second natural question is whether we can choose a notion of
distance between a proposed solution $\xx$ and an optimal solution
$\xx^* \in \argmin_{\xx \in \rea^{n}} \norm{ \AA \xx - \cc}^{2}_2$ s.t. this distance being small is equivalent to $\xx$ being a solution to the {\LSA} instance $(\AA,\cc,\eps)$?
The answer to both questions is yes, as demonstrated by the following facts:
\begin{fact}
\label{fac:voltageErrorEquiv}
\noindent
Suppose $\xx^* \in \argmin_{\xx \in \rea^{n}} \norm{ \AA \xx - \cc}^{2}_2$ then
  \begin{enumerate}
  \item 
    $\norm{\AA^{\trp} \AA \xx - \AA^{\trp}
      \cc}_{(\AA^{\trp}\AA)^{\pinv}} 
    = \norm{ \AA \xx - \proj_{\AA} \cc }_2 
    = \norm{ \xx - \xx^* }_{\AA^{\trp}\AA}
    $.
   \item The following statements are each equivalent to $\xx$ being a solution
    to the {\LSA} instance $(\AA,\cc,\eps)$:
    \begin{enumerate}
    \item $\norm{\AA^{\trp} \AA \xx - \AA^{\trp}
      \cc}_{(\AA^{\trp}\AA)^{\pinv}} \leq \epsilon \norm{\AA^{\trp}
      \cc}_{(\AA^{\trp} \AA)^{\pinv}}$ if and only if $\xx$ is a solution
    to the {\LSA} instance $(\AA,\cc,\eps)$.
  \item $\norm{ \xx - \xx^* }_{\AA^{\trp}\AA} \leq \epsilon \norm{ \xx^* }_{\AA^{\trp}\AA}$ if and only if $\xx$ is a solution
    to the {\LSA} instance $(\AA,\cc,\eps)$.
    \end{enumerate}
  \end{enumerate}
\end{fact}
For completeness, we prove Fact~\ref{fac:voltageErrorEquiv} in
Appendix~\ref{sec:linalg}.
Fact~\ref{fac:voltageErrorEquiv} explains connection between our
Definition~\ref{def:lsapx}, and the usual convention for measuring
error in the Laplacian solver literature~\cite{SpielmanTengSolver:journal}.
In this setting, we consider a Laplacian matrix $\LL$, which can be written
as $\LL = \AA^{\trp} \AA \in \rea^{n \times n}$, and a vector $\bb$ s.t. $\PPi_{ \AA^{\trp} \AA }
\bb = \bb$.
This condition on $\bb$ is easy to verify in the case of Laplacians,
since for the Laplacian of a connected graph, 
$\PPi_{ \AA^{\trp} \AA } = \II - \frac{1}{n} \one \one^{\trp}$.
Additionally, it is also equivalent to the condition that there exists $\cc$
s.t. $\bb = \AA^{\trp} \cc$.
For Laplacians it is possible to
compute both $\AA$ and a vector $\cc$ s.t. $\bb = \AA^{\trp} \cc$ in
time linear in $\nnz(\LL)$. 
For Laplacian solvers, the approximation error of an approximate solution $\xx$
is measured by the $\eps$ s.t. 
$\norm{\AA^{\trp} \AA \xx - \bb}_{(\AA^{\trp}\AA)^{\pinv}} \leq
\epsilon \norm{\bb}_{(\AA^{\trp} \AA)^{\pinv}}$.
By Fact~\ref{fac:voltageErrorEquiv}, we see that this is exactly
equivalent to $\xx$ being a solution to the {\LSA} instance $(\AA,\cc,\eps)$.


\subsection{Measuring the Difficulty of Solving a Linear System}

Running times for iterative linear system solvers generally depend on
the number of non-zeros in the input matrix, the condition number of
the input matrix, the accuracy, and the bit complexity.

In this section, we formally define several measures of complexity of
the linear systems we use.
This is crucial, because we want to make sure that our reductions do
not rely on mapping into extremely ill-conditioned matrices, and so we
use these measures to show that this is in fact not the case.
\begin{definition}
\label{def:complexity_parameters}
  \noindent
  \begin{enumerate}
  \item Given a matrix $\AA \in \rea^{m \times n}$, we define the
    maximum singular value $\sigma_{\max}(\AA) $ in the usual way as
    $
      \sigma_{\max}(\AA) = \max_{\xx \in \rea^{n} , \xx \neq \zero }
      \sqrt{ \frac{ \xx^{\trp} \AA^{\trp} \AA \xx } {\xx^{\trp} \xx }
      }
      .
    $
  \item Given a matrix $\AA \in \rea^{m \times n}$ which is not all
      zeros, we define the
    minimum non-zero singular value $ \sigma_{\min}(\AA)$ as
    $
      \sigma_{\min}(\AA) = \min_{\xx \in \rea^{n} , \xx \orthog
        \nulls(\AA) } \sqrt{ \frac{ \xx^{\trp} \AA^{\trp} \AA \xx }
        {\xx^{\trp} \xx } }
      .
    $
    \item Given a matrix $\AA \in \rea^{m \times n}$ which is not all
      zeros, we define the non-zero condition number of $\AA$
      as
    $
      \CN(\AA) =
      \frac{ \sigma_{\max}(\AA) }
      { \sigma_{\min}(\AA)} 
      .
    $
  \end{enumerate}
\end{definition}

\begin{definition}
The sparse parameter complexity of an $\LSA$ instance $(\AA,\cc,\eps)$ where ${\AA \in
\bz^{m \times n}}$ and $\nnz(\AA) \geq \max(m,n)$, and $\eps > 0$, is 
\[
\calS(\AA,\cc,\eps) 
\defeq 
\left(\nnz(\AA),
\max\left(
\norm{\AA}_{\max},\norm{\cc}_{\max}, \frac{1}{\anzmin(\AA)},
  \frac{1}{\anzmin(\cc)}
\right),
\CN(\AA),
\eps^{-1}
\right)
.
\]
\end{definition}
Note in the definition above that when $\AA \neq \zero$ and $\cc
\neq \zero$ have only integer
entries, we trivially have $\anzmin(\AA) \geq 1$
and $\anzmin(\cc) \geq 1$. However, including
$\frac{1}{\anzmin(\AA)},$ and
  $\frac{1}{\anzmin(\cc)}$ in the definition stated above is useful
when working with intermediate matrices whose entries are not integer valued.

\subsection{Matrix Classes and Reductions Between Them}
We use the term \emph{matrix class} to refer to an infinite set of
matrices $\matclass$. 
In this section, we formally define a notion of efficient reduction between
linear systems in different classes of matrices.



\begin{definition}[Efficient $f$-reducibility]
Suppose we have two matrix classes $\matclass_{1}$ and
$\matclass_{2}$, and there exist two algorithms $\algo_{1\to2}$
and $\algo_{1\leftarrow2}$
s.t. given an $\LSA$ instance $(\MM^{1},\cc^{1},\eps)$,
 where $\MM^{1} \in \matclass_{1}$, the call
$\algo_{1\to2}(\MM^{1},\cc^{1},\eps_{1})$ returns an $\LSA$ instance
$(\MM^{2},\cc^{2},\eps_{2})$
s.t. if $\xx^{2}$ is a solution to $\LSA$ instance
$(\MM^{2},\cc^{2},\eps_{2})$ then
 $\xx^{1} =
\algo_{1\leftarrow2}(\MM^{1},\MM^{2},\xx^{2})$ is a
solution to $\LSA$ instance $(\MM^{1},\cc^{1},\eps_{1})$.

Consider a function of $f : \rea^{4}_+ \to \rea^{4}_+$ s.t. every
output coordinate is an increasing function of every input coordinate. 
Suppose that we always have
\[
\calS(\MM^{2},\cc^{2},\eps_{2}) \leq f(\calS(\MM^{1},\cc^{1},\eps_{1}))
,
\]
and the running times of $\algo_{1\to2}(\MM^{1},\cc^{1},\eps_{1})$ and
$\algo_{1\leftarrow2}(\MM^{1},\MM^{2},\xx^{2})$ are both bounded
by $O(\nnz(\MM^{1}))$.

Then we say that $\matclass_{1}$ is efficiently $f$-reducible to
$\matclass_{2}$, which we also write as
\[
\matclass_{1} \leq_{f}
\matclass_{2}
.
\]
\end{definition}

\begin{lemma}
\label{lem:composeReductions}
  Suppose $\matclass_{1} \leq_{f}
\matclass_{2}$ and $\matclass_{2} \leq_{g}
\matclass_{3}$.
Then $\matclass_{1} \leq_{g \circ f} \matclass_{3}$.
\end{lemma}
\begin{proof}
  The proof is simply by the trivial composition of the two reductions.
\end{proof}

\begin{definition}
We let $\genCl$ denote the class of all matrices with integer
    valued entries s.t. there is at least one non-zero entry in every
    row and column\footnote{If there is a row or column with only
      zeros, then it can always be handled trivially in the context of
    solving linear systems}.
\end{definition}




\section{Main Results}
\label{sec:results}
In this section, we use the notions of sparse parameter complexity and
matrix class reductions to prove our main technical result,
Theorem~\ref{thm:genToMc:param}, which shows that linear systems in
general matrices with integer entries can be efficiently reduced to
linear systems in several different classes of Incidence Structured
Block Matrices.
From this result, we derive as corollary our main result,
Theorem~\ref{thm:runningTimeHardness}, which states that fast high
accuracy solvers for several types of ISBMs imply fast high accuracy
solvers for all linear systems in general matrices with integer
entries.

\begin{theorem}
\label{thm:genToMc:param}
 Let $f(s,U,K,\eps) = (O(s\log(sU)), \poly(U K \eps^{-1} s)  ,  \poly(U K \eps^{-1} s)  , \poly(U K \eps^{-1} s)  )$,
then
\begin{enumerate}
\item 
$
\genCl \leq_{f} \mctwostrictintCl
$.
\item 
$
\genCl \leq_{f} \trusstwoCl
$.
\item 
$
\genCl \leq_{f} \tvtwoCl
$.
\end{enumerate}

\end{theorem}

\begin{theorem}
\label{thm:runningTimeHardness}
Suppose we have an algorithm which solves every Linear System
Approximation Problem $(\AA,\cc,\eps)$ with sparse parameter
complexity
$\calS(\AA,\cc,\eps) \leq (s,U,K,\eps^{-1})$ in time $O( s^{a}
 \polylog(s,U,K,\eps^{-1}))$ for some $a \geq 1$, 
whenever  $\AA \in \mathcal{R}$ for at least one of
 $\mathcal{R} \in
\setof{  \mctwostrictintCl, \trusstwoCl, \tvtwoCl }$.
I.e. we have a ``fast'' solver\footnote{The reduction requires only a single linear system
   solve, and uses the solution in a black-box way.
   So the reduction also applies if the solver for the class $\mathcal{R}$
   only works with high probability or only has running time
   guarantees in expectation.} for one of the matrix classes
$\mctwostrictintCl, \trusstwoCl,$
  or $\tvtwoCl$.
 Then every Linear System Approximation Problem $(\AA,\cc,\eps)$ where $\AA \in \genCl$  
 with sparse parameter complexity $\calS(\AA,\cc,\eps) \leq (s,U,K,\eps^{-1})$ can be
 solved
in time $O( s^{a}
 \polylog(s,U,K,\eps^{-1}))$.
\end{theorem}
\begin{proof}
  The theorem is a immediate corollary of Theorem~\ref{thm:genToMc:param}.
\end{proof}
\begin{definition}
We let $\genCl_{\text{z},2}$ denote the class of all matrices with integer
    valued entries s.t. there is at least one non-zero entry in every
    row and column, and every row has zero row sum, and for each row, 
    the sum of the positive coefficients is a power of 2.
\end{definition}
%
\begin{lemma}
\label{lem:genToZeroSumTwo}
 Let $f(s,U,K,\eps) = \left( O(s), O\left( \epsilon^{-1} s^{9/2}U^3 \right),  O\left( \epsilon^{-1} s^8 U^3 K \right), O \left( s^{5/2} U^2 \epsilon^{-1} \right)  \right)$,
then 
\[
\genCl \leq_{f} \genZtwoCl
.
\]
\end{lemma}

\begin{lemma}
\label{lem:zeroSumTwoToMcTwo}
 Let $f(s,U,K,\eps) = \left( O(s \log (sU)), O \left( s^{3/2}U \log^{1/2}(sU) \right),  O\left( K s^{4} U^2 \log^2 (sU)  \right), O \left(sU^2 \eps^{-1} \right)  \right)$,
then 
\[
 \genZtwoCl \leq_{f} \mctwoCl
.
\]
\end{lemma}

\begin{lemma}
\label{lem:McTwoToMcTwoStrict}
 Let $f(s,U,K,\eps) = \left( O(s), O \left( \eps^{-1}U^2 K \right), O\left( \eps^{-1} s^2 U^2 K \right), O(\epsilon^{-1})  \right)$,
then 
\[
\mctwoCl \leq_{f} \mctwostrictCl
.
\]
\end{lemma}

\begin{lemma}
\label{lem:McTwoStrictToMcTwoStrictInt}
 Let $f(s,U,K,\eps) = (s, \eps^{-1}sU,  2K  , O(\epsilon^{-1}) )$,
then 
\[
\mctwostrictCl \leq_{f} \mctwostrictintCl
.
\]
\end{lemma}

\begin{lemma}
\label{lem:McTwoStrictToTruss}
 Let $f(s,U,K, \eps)$ be as defined in Lemma~\ref{lem:zeroSumTwoToMcTwo}
then 
\[
 \genZtwoCl \leq_{f} \trusstwoCl
.
\]
\end{lemma}

\begin{lemma}
\label{lem:McTwoStrictToTV}
 Let $f(s,U,K, \eps) = (s,U,K,\eps^{-1})$,
then 
\[
 \mctwoCl \leq_{f} \tvtwoCl
.
\]
\end{lemma}

\begin{proof}[Proof of Theorem~\ref{thm:genToMc:param}]
Follows by appropriate composition (Lemma~\ref{lem:composeReductions})
applied to the the Lemmas above,
i.e. \ref{lem:genToZeroSumTwo}, 
\ref{lem:zeroSumTwoToMcTwo}, \ref{lem:McTwoToMcTwoStrict}, 
\ref{lem:McTwoStrictToMcTwoStrictInt}, 
\ref{lem:McTwoStrictToTruss}
and
\ref{lem:McTwoStrictToTV}.



\end{proof}

\subsection{Outline of Remaining Sections}

In Section~\ref{sec:GZ2toMC2} is presents the proof of
Lemma~\ref{lem:zeroSumTwoToMcTwo}, i.e. $\genZtwoCl \leq_{f} \mctwoCl$. This
statement is our most important reduction.
In Section~\ref{sec:McTwoToMcTwoStrict}, we prove
Lemma~\ref{lem:zeroSumTwoToMcTwo}, the surprising reduction 
$\mctwoCl \leq_{h} \mctwostrictCl$, which shows that we can solve
\emph{normal equations} even while changing null-spaces involved
substantially.
In Section~\ref{sec:McTwoToMcTwoStrict}, we show how to round weights
in a 2-commodity problem to integers, proving Lemma~\ref{lem:McTwoStrictToMcTwoStrictInt}.
In Section~\ref{sec:genToZeroSumTwo} is presents the proof of
Lemma~\ref{lem:genToZeroSumTwo}. This is a
simpler reduction that we use to establish the properties that our
Lemma~\ref{lem:zeroSumTwoToMcTwo} relies on.
In Section~\ref{sec:trusses}, we present the proof of
Lemma~\ref{lem:McTwoStrictToTruss}.
Section~\ref{sec:ipm} describes how Interior
Point Methods give rise to multi-commodity and Total
Variation linear systems.
Section~\ref{sec:ipm} also contains a proof of Lemma~\ref{lem:McTwoStrictToTV}.


\newcommand{\algGZTtoMCT}{\ensuremath{\textsc{Reduce\,}\genZtwoCl\textsc{to}\mctwoCl}}
\newcommand{\algMCTtoGZTsoln}{\ensuremath{\textsc{MapSoln\,}\mctwoCl\textsc{to}\genZtwoCl}}
\newcommand{\algMCToGZGadget}{\ensuremath{\mctwoCl\textsc{Gadget}}}
\newcommand{\xxa}{\xx^{\text{A}}}
\newcommand{\xxb}{\xx^{\text{B}}}
\newcommand{\xxaux}{\xx^{\text{aux}}}
\newcommand{\cca}{\cc^{\text{A}}}
\newcommand{\ccb}{\cc^{\text{B}}}
\newcommand{\epsa}{\eps^{\text{A}}}
\newcommand{\epsb}{\eps^{\text{B}}}
\newcommand{\xxbopt}{\xx^{\text{B}*}}
\newcommand{\xxaopt}{\xx^{\text{A}*}}
\newcommand{\xxatil}{\widetilde{\xx}^{\text{A}}}
\newcommand{\xxauxopt}{\xx^{\text{aux}*}}
\newcommand{\nb}{n_B}
\newcommand{\na}{n_A}
\newcommand{\vecppi}{\pp^i}
\newcommand{\qqj}{\qq^j}
\newcommand{\ppxa}{\pp(\xxa)}
\newcommand{\yyb}{\yy^{\text{B}}}
\newcommand{\yya}{\yy^{\text{A}}}
\newcommand{\yyaux}{\yy^{\text{aux}}}
\newcommand{\zzb}{\zz^{\text{B}}}
\newcommand{\zza}{\zz^{\text{A}}}
\newcommand{\zzaux}{\zz^{\text{aux}}}

\section{Reducing Zero-Sum Power Two Linear Systems to Two-Commodity
  Linear Systems}
\label{sec:GZ2toMC2}

To prove Lemma~\ref{lem:zeroSumTwoToMcTwo},
we need to provide mapping algorithms $\algo_{\genZtwoCl \to \mctwoCl}$
for mapping linear system approximation (\LSA) instances over matrices
in $\genZtwoCl$ to $\LSA$ problem instances over matrices in $\mctwoCl$,
as well as $\algo_{\mctwoCl \leftarrow \genZtwoCl}$ for mapping
the resulting solutions back.
These leads to the following main components:
\begin{itemize}
\item 
Algorithm~\ref{alg:GZ2toMC2} states the pseudo-code for the
  algorithm~\algGZTtoMCT, which implements the desired mapping of
  problem instances. 
Given $\LSA$ problem instance $(\AA,\cca,\epsa)$ where
$\AA \in \genZtwoCl$,
the call  \algGZTtoMCT$(\AA,\cca,\epsa)$ returns an $\LSA$
problem instance $(\BB,\cc^{B},\epsb)$
where $\BB \in \mctwoCl$.
(Strictly speaking, $\algGZTtoMCT$ also has a parameter $\alpha$ which
we will set before using the algorithm.)
\item
 Algorithm~\ref{alg:MC2gadget} provides the pseudo-code for
 $\mctwoCl$Gadget, a short subroutine used in \algGZTtoMCT to
 represent the equation $a + b = 2 c$ using two-commodity constraints.
 
\item
  Algorithm~\ref{alg:GZ2toMC2solnback} provides the (trivial)
  pseudo-code for {\algMCTtoGZTsoln} used to map solutions to $\LSA$
  problems over $\mctwoCl$ back to solutions over $\genZtwoCl$
  by restricting onto the original variables.
\end{itemize}

Pseducode of the key reduction routine that creates the new linear system,
$\algGZTtoMCT$, is shown in Algorithm~\ref{alg:GZ2toMC2}.
Note that given two singleton multi-sets each containing a single equation, e.g.
$\setof{\aa_{1}^{\trp} \xx = c_1}$ and $\setof{\aa_{2}^{\trp} \xx  =
  c_2}$ where $\aa_1, \aa_2$ are vectors, we define
$\setof{\aa_{1}^{\trp} \xx = c_1} +\setof{\aa_{2}^{\trp} \xx  = c_2} 
=
\setof{\aa_{1}^{\trp} \xx +  \aa_{2}^{\trp} \xx  =c_1+ c_2}$
and we define 
$\setof{\aa_{1}^{\trp} \xx = c_1} \union \setof{\aa_{2}^{\trp} \xx  = c_2} 
=
\setof{\aa_{1}^{\trp} \xx = c_1, \aa_{2}^{\trp} \xx  = c_2}$.

\begin{algorithm}[H]
  \renewcommand{\algorithmicrequire}{\textbf{Input:}}
  \renewcommand\algorithmicensure {\textbf{Output:}}
  \caption{\label{alg:GZ2toMC2}\algGZTtoMCT}
  \begin{algorithmic}[1]
    \REQUIRE{$(\AA,\cca,\eps_{1})$ where $\AA \in \genZtwoCl$ is an $m \times
      n$ matrix, $\cca \in \rea^n$, $0 <
      \eps_{1} < 1$, and 
      $\alpha > 0$.
    }
    \ENSURE{$(\BB,\ccb,\eps_{2})$ where $\BB \in \mctwoCl$ is an $m \times
      n$ matrix, $\ccb \in \rea^n$, and $0 <
      \eps_{2} < 1$.}
    \STATE $
    \epsilon_{2} \assign 
    \eps_1  \left(1 + \frac{1}{\alpha} \right)^{-1/2}
      \left( 1 + \frac{\norm{\cca}_2^2 \sigma^2_{\max}(\AA)}{\alpha+1}
      \right)^{-1/2}
    $
    \STATE $\mathcal{X} \assign \{\uu_1, \ldots, \uu_n, \vv_1,\ldots,
    \vv_n\}$\hfill \COMMENT{$\mctwoCl$ variables and index of new variables}
    \STATE Let $\xx$ be the vector of variables corresponding to the
    set of variables $\mathcal{X}$
    \STATE
    $t \assign n+1$
    
    \STATE Initialize $\mathcal{A} \assign \emptyset$,
    $\mathcal{B} \assign \emptyset$
    \hfill \COMMENT{Multisets of main and $\mctwoCl$ auxiliary equations respectively}
    \FOR {each equation $1 \leq i \leq m$ in $\AA$} 
    \STATE Let $\mathcal{I}^{+1} \assign \{j: \AA_{ij} >
    0\}, \mathcal{I}^{-1} \assign \{j: \AA_{ij} < 0 \}$
    \label{lin:mcAlready}
    \IF {$\abs{\mathcal{I}^{+1}} = 1$ and $\abs{\mathcal{I}^{-1}} = 1$}
    \STATE Let the only elements in $\mathcal{I}^{+1}$ and $\mathcal{I}^{-1}$
    be $j_+$ and $j_-$ respectively.
    \hfill    \COMMENT{Note $\AA_{i j_+} = -\AA_{i j_-}$}
    \STATE 
    $\mathcal{A} \assign \mathcal{A} \cup
    \{\AA_{i j_+} \uu_t - \AA_{i j_+} \uu_{j_-} =
    \cc_i\}$
    \STATE
    $w_{i} \assign \alpha $
    \STATE
    $\mathcal{B} \assign \mathcal{B} \cup
    w_{i}^{1/2} \cdot \{\AA_{i j_+} \uu_{j_+} - \AA_{i j_+}\uu_t = 0 \}$
    \STATE $\mathcal{X} \assign \mathcal{X} \cup \{\uu_t, \vv_t\}$, update $\xx$ accordingly
    \STATE$t \assign t+1$
    \ELSE
    \STATE $\mathcal{A}_{i} \assign \setof{\AA_{i} \uu = \cc_{i}}$ \label{lin:mainConstraint}
    \STATE $\mathcal{B}_{i} \assign \emptyset $
    \FOR {$s = -1,+1$}
    \STATE $r \assign 0$
    \WHILE{$\mathcal{A}_i$ has strictly more than $1$ coefficient with sign $s$} \label{lin:varReplace}
    \STATE For each $j$, let $\AAhat_{ij}$ be the coefficient of
    $\uu_j$ in $\mathcal{A}_{i}$.
    \STATE $\mathcal{I}_{odd}^{s} \assign \{j \in
    \mathcal{I}^{s}: \lfloor |\AAhat_{ij}|  / 2^{r} \rfloor \mbox{ is
      odd} \}$
    \STATE Pair the indices of $\mathcal{I}_{odd}^{s}$ into $k$
    disjoint pairs $(j_{k},l_{k})$
    \FOR {each pair of indices $(j_{k},l_{k})$}
    \STATE $\mathcal{A}_{i} \leftarrow \mathcal{A}_{i}+s\cdot 2^{r}
    \setof{ \left( 2 \uu_{t} - \left( \uu_{j_{k}} + \uu_{l_{k}}
      \right)\right) = 0}$  
    \STATE $\mathcal{B}_{i} \assign
    \mathcal{B}_{i} \cup -s\cdot 2^{r} \cdot
    \mctwoCl\textsc{Gadget}(\uu_{j_k}, \uu_{l_k}, t)$ 
    \label{lin:callGadget}
    \STATE $\mathcal{X} \assign \mathcal{X}
    \cup \{\uu_t, \ldots, \uu_{t+6}, \vv_{t}, \ldots,
    \vv_{t+6}\}$, update $\xx$ accordingly
    \STATE $t \assign t+7$
    \STATE $r \assign r+1$
    \ENDFOR
    \ENDWHILE
    \ENDFOR
    \STATE
    $w_{i} \assign \alpha \sizeof{\mathcal{B}_{i}}$
    \STATE $\mathcal{B} \assign
    \mathcal{B} \cup w_{i}^{1/2}  \cdot \mathcal{B}_{i}$  
    \label{lin:weight}
    \STATE $\mathcal{A} \assign \mathcal{A}
    \cup \mathcal{A}_{i}$.
    \ENDIF
    \ENDFOR
    \RETURN $\eps_{2}$ and $\BB, \cc$ s.t. $\BB\xx = \cc$ corresponds to the equations in $\mathcal{A} \cup \mathcal{B}$,
    on the variable set $\mathcal{X}$.
  \end{algorithmic}
\end{algorithm}

The central object created by Algorithm~\ref{alg:GZ2toMC2} is the matrix
$\BB$, which contains both new equations and new variables.
We will superscript the variables with $^{A}$ to
distinguish variables appear in th original equation $\AA \xxa =
\cca$ from new variables.
We will term the new variables as $\xxaux$, and write a vector
solution for the new problem, $\xxb$, as:
\begin{equation}
\xxb = \left( \begin{array}{c}
\xxa \\ 
\xxaux
\end{array} \right).
\label{eq:MC2BDecomposeX}
\end{equation}

Let $\na$ be the dimension of $\xxa$, and $\nb$ be the dimension of $\xxb$, respectively.
We order the variables so that for an appropriately chosen index $h$
we have
\begin{enumerate}
\item $\xxaux_{1:h}$ corresponds to the $\uu$-coordinates of the
  auxiliary variables created in $\mctwoCl$-gadgets. 
\item $\xxaux_{h+1:\nb-\na}$ corresponds to the $\vv$-coordinates of the auxiliary variables created in $\mctwoCl$-gadgets.
\end{enumerate}
With this ordering $\xxb_{t} = \uu_{t}$ for $t \leq \na + h$.

Furthermore, we will distinguish the equations in $\BB$ into ones
formed from manipulating $\AA$, i.e. the equations added to
the set  $\mathcal{A}$, from the auxiliary equations, i.e. the
equations added to the set $\mathcal{B}$.
We use $\WW^{1/2} = \diag_{i}(w_{i}^{1/2}) $ to refer to the diagonal matrix of
weights $w_{i}$ applied to the auxiliary equations $\mathcal{B}$ in
Algorithm~\ref{alg:GZ2toMC2}.
In Algorithm~\ref{alg:GZ2toMC2}, a real value $\alpha > 0$ is set
initially and used when computing the weights $w_{i}^{1/2}$.
For convenience, thoughout most of this section, we will treat
$\alpha$ as an arbitrary constant, and only eventually substitute in
its value to complete our main proof.

This leads to the following representation of $\BB$
and $\ccb$ which we will use throughout our analysis of the algorithm:
\begin{equation}
\BB
= \left( \begin{array}{c}
\AAhat \\
\WW^{1/2} \BBhat
\end{array} \right).
\label{eq:MC2BDecomposeB}
\end{equation}
Here the equations of $\AAhat$
corresponds to $\mathcal{A}$ in \algGZTtoMCT, and
$\BBhat$ corresponds to the auxiliary constraints,
i.e. equations of $\mathcal{B}$ in \algGZTtoMCT.
Also, the vector $\ccb$ created is simply an extension of $\cca$:
\begin{equation}
\ccb =
\left( \begin{array}{c}
	 \cca\\ 
	\zero
\end{array} \right).
\label{eq:MC2BDecomposeC}
\end{equation}

Finally, as Algorithm~\ref{alg:GZ2toMC2} creates new equations for
each row of $\AA$ independently, we will use 
$S_i$ to denote the subset of indices of the rows of $\BBhat$
that's created from $\AA_{i, :}$ aka. the the auxiliary constraints generated
from the call to $\algMCToGZGadget$ upon processing $\AA_{i, :}$.
We will also denote the size of these using
\[
m_i \defeq \abs{S_i},
\]
and use $\BBhat_{i}$
 to denote these part of $\BBhat$ that corresponds
to these rows.
The Gadget routine used in the reduction, and the (trivial)
solution mapper are stated below.


\begin{algorithm}[H]
\renewcommand{\algorithmicrequire}{\textbf{Input:}}
\renewcommand\algorithmicensure {\textbf{Output:}}
    \caption{\label{alg:MC2gadget}\algMCToGZGadget} 
    \begin{algorithmic}[1]
    \REQUIRE Scalar variables $\uu_{j_k}, \uu_{l_k}$, integer $t$.
    \ENSURE Set of equations for a multi-commodity gadget
	    that computes the average of $\uu_{j_k}$ and $\uu_{l_k}$.
    \RETURN  
    \begin{align*}
\{ \uu_{t+3} - \vv_{t+3}  - ( \uu_{t+4} - \vv_{t+4} ) &= 0,\\
\uu_{t} - \uu_{t+3} &= 0 ,\\
\uu_{t+4}  - \uu_{j_k} &=0 , \\
\vv_{t+3} - \vv_{t+1}   &=0,\\
\vv_{t+2} - \vv_{t+4} &=0,\\
\uu_{t+5} - \vv_{t+5}  - (\uu_{t+6} - \vv_{t+6}) &=0,\\
\uu_{t} - \uu_{t+5}   &=0,\\
\uu_{t+6}  - \uu_{l_k} &=0,\\
\vv_{t+5} -  \vv_{t+2}   &=0,\\
\vv_{t+1} - \vv_{t+6} &=0 \}
    \end{align*}
    \end{algorithmic}
\end{algorithm}

\begin{algorithm}[H]
\renewcommand{\algorithmicrequire}{\textbf{Input:}}
\renewcommand\algorithmicensure {\textbf{Output:}}
    \caption{\label{alg:GZ2toMC2solnback}\algMCTtoGZTsoln} 
    \begin{algorithmic}[1]
    \REQUIRE $m \times n$ matrix $\AA \in \genZtwoCl$, 
    $m' \times n'$ matrix $\BB \in \mctwoCl$,
    vector $\cca \in \rea^{m}$,
    vector $\xxb \in \rea^{n'}$.
    \ENSURE Vector $\xxa \in \rea^{n}$.
    \IF {$\AA^{\trp} \cca = {\bf 0}$}
    	\RETURN $\xxa \assign {\bf 0}$
    \ELSE
    	\RETURN $\xxa \assign \xxb_{1:n}$	
    \ENDIF
    \end{algorithmic}
\end{algorithm}


We upper bound the number of nonzero entries of $\BB$ by the number of nonzero entries of $\AA$.
\begin{lemma}
\noindent
\label{lem:GZToMC2NNZ}
\begin{itemize}
\item We have $\nnz(\BB) = O \left( \nnz(\AA) \log \norm{\AA}_{\infty}
  \right)$.
\item  Both dimensions of $\BB$ are $O \left( \nnz(\AA) \log
  \norm{\AA}_{\infty} \right)$.
\item The number of iterations required by Algorithm~\ref{alg:GZ2toMC2} 
to process row $\AA_i$ is at most $\log \norm{\AA_i}_1$. 
\item 
  When processing row $\AA_i$, every new variable that appears in the $\mathcal{A}_{i}$
  equation is paired in the following iteration of the while-loop in Line~\ref{lin:varReplace} (unless it is the final iteration), and then
  disappears from $\mathcal{A}_{i}$.
\end{itemize}

\end{lemma}

\begin{proof} 
Since Algorithm~\ref{alg:GZ2toMC2} constructs new equations for each row of $\AA$ independently, we bound the number of new variables and new equations (that is, the size of the submatrix $\left(\begin{array}{c}
\AAhat_i \\ \BBhat_i 
\end{array} \right)$) for each row $i$ of $\AA$ separately.

Let $n_i$ be the number of nonzero entries of the $i$th row of $\AA$.
We count the number of variables in the equation $\mathcal{A}_i$, at each iteration of the while-loop in line~\ref{lin:varReplace} of Algorithm~\ref{alg:GZ2toMC2}.
Let $X^{(r)}$ be the subset of variables with nonzero coefficients in $\mathcal{A}_i$, at the end of iteration $r$.
Let $X_{aux}^{(r)}$ be the subset of $X^{(r)}$ containing all auxiliary variables created in iteration $r$.

Note in each iteration, Algorithm~\ref{alg:GZ2toMC2} replaces two variables by a new auxiliary variable. It gives that
\[
\abs{X_{aux}^{(1)}} \le \frac{n_i}{2},
\mbox{ and }
\abs{X^{(1)} \setminus X_{aux}^{(1)}} \le n_i.
\]
Since each auxiliary variable in $X_{aux}^{(1)}$ has coefficient $2$, in the 2nd iteration, together with another variable of coefficient $2$, 
it must be replaced by a new auxiliary variable.
Thus, at the end of the 2nd iteration, all auxiliary variables $X_{aux}^{(1)}$ will not appear in the equation $\mathcal{A}_i$. 
This implies that
\[
\abs{X_{aux}^{(2)}} \le \frac{n_i}{2} + \frac{n_i}{4},
\mbox{ and }
\abs{X^{(2)} \setminus X_{aux}^{(2)}} \le n_i.
\]
Similarly, at the end of the $t^{\text{th}}$ iteration, we have
\[
\abs{X_{aux}^{(r)} } \le \sum_{1\le s\le r} \frac{n_i}{2^{s}}
\le n_i,
\]
and 
\[
\abs{X^{(r)} \setminus X_{aux}^{(r)}}  \le n_i.
\]

Since Algorithm~\ref{alg:GZ2toMC2} pulls out a factor 2 in each iteration, the total number of iterations is at most $\log \norm{\AA_i}_1$. 
Since each auxiliary variable in $\mathcal{A}_i$ during the construction corresponds to $O(1)$ auxiliary variables and $O(1)$ equations in $\BBhat_i$,
the total number of auxiliary variables and equations for row $i$ of $\AA$ is 
\[
O\left( n_i \log \norm{\AA_i}_1 \right).
\]
Therefore, the number of variables and the number of equations in $\BB$ (that is, the both dimensions of $\BB$) are
\[
O \left( \nnz(\AA) \log \norm{\AA}_{\infty} \right).
\]
Since each row of $\BB$ has $O(1)$ nonzero coefficients, we have
\[
\nnz(\BB) = O \left( \nnz(\AA) \log \norm{\AA}_{\infty} \right).
\]
This completes the proof.
\end{proof}

\subsection{Reduction Between Exact Solvers}
\label{subsec:GZToMCExactSolver}
The most important relation between $\AA$ and $\BB$ is given by the
following claim.

\begin{claim}[Reduction between exact solvers]
\label{clm:exactReduction}
Fix any $\xxa \in \rea^{n}$. Then
\[
\norm{\AA\xxa - \cca}_2^2
=
\frac{\alpha+1}{\alpha} \min_{\xxaux} 
\norm{\BB \left( \begin{array}{c}
\xxa\\ 
\xxaux
\end{array} \right) - \ccb}_2^2.
\]
\end{claim}
As a Corollary of Claim~\ref{clm:exactReduction}, we observe the
following:
\begin{lemma}
  Given $\LSA$ problem instance $(\AA,\cca,\epsa)$ where $\AA \in
\genZtwoCl$, suppose
\[
(\BB,\ccb,\epsb) = \algGZTtoMCT(\AA,\cca,\epsa).
\]
Then $\BB \in \mctwoCl$ and if $\xxb = \left( \begin{array}{c}
\xxa\\ 
\xxaux
\end{array} \right)$ is a solution to the \emph{exact} {\LSA}  problem
$(\BB,\ccb,0)$, then $\xxa$ is a solution to the \emph{exact}
{\LSA} problem $(\AA,\cca,0)$.
\end{lemma}

\begin{proof}
 This follows immediately from minimizing over $\xxa$ on both
 sides of the equation established by Claim~\ref{clm:exactReduction}
 and then applying and Fact~\ref{fac:projIsQFmin}.
\end{proof}
Before proving Claim~\ref{clm:exactReduction}
we first note a basic guarantee obtained by \algGZTtoMCT:
\begin{align}
\label{eqn:eqnInvariant}
\one^{\trp} \left( \left( \begin{array}{c}
\AAhat_i \\
\BBhat_i
\end{array} \right)  \xxb
-  \left( \begin{array}{c}
\cca_i \\
{\bf 0}
\end{array} \right) \right)
= 
\AA_i \xxa - \cca_i
.
\end{align}

To verify this guarantee, we consider two cases separately.
The first case is when the condition
$\abs{\mathcal{I}^{+1}} = 1 \text{ AND } \abs{\mathcal{I}^{-1}} = 1$ is true (see
Algorithm~\ref{alg:GZ2toMC2}, Line~\ref{lin:mcAlready}).
In this case, the main constraint in the output
corresponding to row $i$ is  $\{\AA_{i j_+} \xxb_t - \AA_{i
  j_+} \xxb_{j_-} = \cca_i\}$,
while the auxiliary constraints contain only a single row $\{\AA_{i
  j_+} \xxb_{j_+} - \AA_{i j_+}\xxb_t = 0 \}$, and adding these
proves the guarantee for this case.
The second case is when the condition in Algorithm~\ref{alg:GZ2toMC2},
Line~\ref{lin:mcAlready}) is false.
We consider the case $s = +1$. The case $s = -1$ is proved similarly.
Note that we will refer to variables $j_k$ and $l_k$ only in the
context of a fixed value of $t$, which always ensures that they are
unambiguosly defined.
In the case $s = +1$,
each time we modify $\mathcal{A}_i$ by adding
$2^{r} \left( 2 \xxb_{t+1} - \left( \xxb_{j_k} + \xxb_{l_k}
  \right)\right) = 0$,
we we also use $\mctwoCl$Gadget to create auxiliary constraints that sum up to exactly 
$-2^{r} \left( 2 \xxb_{t+1} - \left( \xxb_{j_k} + \xxb_{l_k}
  \right)\right)$, so adding these together will cancel out the changes.

\begin{proof}[Proof of Claim~\ref{clm:exactReduction}]

For convenience, we also write
\begin{enumerate}
\item $\delta_j \defeq \BBhat_j \xxb$,
\item $\hat{\delta}_i \defeq \AAhat_i \xxb - \cca_i$,
\item $\eps_i \defeq \AA_i \xxa - \cca_i$.
\end{enumerate}
Thus
\[
\norm{\left( \begin{array}{c}
\AAhat_i \\
\WW^{1/2}_i \BBhat_i
\end{array} \right)  \xxb
-  \left( \begin{array}{c}
\cca_i \\
{\bf 0}
\end{array} \right)
}_2^2 =
\hat{\delta}_i^2 + w_i \sum_{j \in S_i} \delta_j^2.
\]
Summing over all rows, we get
\[
\norm{\BB\left( \begin{array}{c}
\xxa \\
\xxaux
\end{array} \right) - \ccb
}_2^2
=
\sum_{i} 
\left( \hat{\delta}_i^2 + w_i \sum_{j \in S_i} \delta_j^2 \right).
\]
Similarly, the square of row $i$ of $\AA\xxa - \cca$ is $\eps_i^2$
and summing over all rows we get
\[
\norm{\AA\xxa - \cca}_2^2
=
\sum_i \eps_i^2
.
\]
Recalling Equation~\eqref{eqn:eqnInvariant}, we have
\[
\one^{\trp} \left( \left( \begin{array}{c}
\AAhat_i \\
\BBhat_i
\end{array} \right)  \xxb
-  \left( \begin{array}{c}
\cca_i \\
{\bf 0}
\end{array} \right) \right)
= 
\AA_i \xxa - \cca_i
.
\]
Thus
\[
\hat{\delta}_i + \sum_{j \in S_i} \delta_j = \eps_i.
\label{eqn:sum}
\]
By the Cauchy-Schwarz inequality (applied to two vectors $\aa$ and
$\bb$ given by $\aa_1 = \hat{\delta}_i, \aa_j = \sqrt{w_i} \delta_j, \bb_1 = 1, \bb_j = 1 / \sqrt{w_i}$),
\[
\eps_i^2 = \left( \hat{\delta}_i + \sum_{j \in S_i} \delta_j \right)^2
\le \left( \hat{\delta}_i^2 + w_i \sum_{j \in S_i} \delta_j^2 \right) \left( 1 + \frac{m_i}{w_i} \right).
\label{eqn:cauchy}
\]
the equality holds if and only if
\begin{align}
\label{eqn:cs_condition}
\hat{\delta}_i = w_i \delta_j\qquad \forall j \in S_i.
\end{align}
Note that we have ensured that for all $i$,
$w_i = \alpha m_i$.

By summing over rows we conclude that for every $\xxa$ and every
$\xxaux$, we have 
\begin{align}
\norm{\AA\xxa - \cca}_2^2
\leq 
\frac{\alpha+1}{\alpha}
\norm{\BB\left( \begin{array}{c}
\xxa \\
\xxaux
\end{array} \right) - \ccb
}_2.
\label{eqn:QFineq}
\end{align}
The inequality above will be an equality if Equations~\eqref{eqn:cs_condition} are
satisfied.
We now show that for every fixed $\xxa$, minimizing over
$\xxaux$ ensures that \eqref{eqn:QFineq} holds with
equality.

In particular, we will momentarily prove the following Claim.
\begin{claim}
\label{clm:errorsObtainable}
 For any fixed $\xxa$ and its associated
$\eps_i$ values, for each row $i$ of $\AA$, the linear system 
\begin{align}
\AAhat_i 
\left( \begin{array}{c}
\xxa \\
\xxaux
\end{array} \right)
 & = \cca_i + \frac{\alpha}{\alpha+1} \eps_i,
\label{eqn:error_ls_a} \\
\BBhat_j 
\left( \begin{array}{c}
\xxa  \\
\xxaux
\end{array} \right) & = \frac{1}{(\alpha+1)m_i} \eps_i, \forall j \in S_i.
\label{eqn:error_ls_b}
\end{align} 
has a solution (which may not be unique).
\end{claim}
Since every auxiliary variable is associated with only one row $i$ of
$\AA$, Claim~\ref{clm:errorsObtainable} implies that we can choose $\xxaux$ s.t. all these
linear systems are satisfied at once.

Given such a choice of $\xxaux$, we get that
Equations~\eqref{eqn:cs_condition} are satisfied so
\[
\norm{\AA\xxa - \cca}_2^2
=
\frac{\alpha+1}{\alpha}
\norm{\BB\left( \begin{array}{c}
\xxa \\
\xxaux
\end{array} \right) - \ccb
}_2^2
.
\label{eqn:QFineq}
\]
Given Claim~\ref{clm:errorsObtainable}, this completes the proof of
Claim~\ref{clm:exactReduction}.
\end{proof}

\begin{proof}[Proof of Claim~\ref{clm:errorsObtainable}]

We focus on the case of the condition $\abs{\mathcal{I}^{+1}} = 1 \text{
  AND } \abs{\mathcal{I}^{-1}} = 1$ in Algorithm~\ref{alg:GZ2toMC2},
Line~\ref{lin:mcAlready} being false. 
The case when then condition is true is very similar, but easier as it
deals with a set of just two equations.

We will construct an assignment to all the variables of
$\xxaux$ s.t. Equations~\eqref{eqn:error_ls_a}
and~\eqref{eqn:error_ls_b} are satisfied.
We start with an assignment $\xxa$ to the main variables, and we
then assign values to auxiliary variables in the order they are
created by the algorithm~{\algGZTtoMCT}.
Note that we will refer to variables $j_k$ and $l_k$ only in the
context of a fixed value of $t$, which always ensures that they are
unambiguosly defined.
When the algorithm processes pair $\xxb_{j_k},\xxb_{l_k} = \uu_{j_k},\uu_{l_k}  $,
the value of these variables will have been set already,  while
$\xxb_{t} = \uu_t$ and the other newly created auxiliary variables have not.
Every auxiliary variable is associated with only one row, so we never
get multiple assignments to a variable using this procedure.

Recall the constraints created by the $\mctwoCl$Gadget call are
    \begin{align*}
\{ \uu_{t+3} - \vv_{t+3}  - ( \uu_{t+4} - \vv_{t+4} ) &= 0,\\
\uu_{t} - \uu_{t+3} &= 0 ,\\
\uu_{t+4}  - \uu_{j_{k}} &=0 , \\
\vv_{t+3} - \vv_{t+1}   &=0,\\
\vv_{t+2} - \vv_{t+4} &=0,\\
\uu_{t+5} - \vv_{t+5}  - (\uu_{t+6} - \vv_{t+6}) &=0,\\
\uu_{t} - \uu_{t+5}   &=0,\\
\uu_{t+6}  - \uu_{l_{k}}  &=0,\\
\vv_{t+5} -  \vv_{t+2}   &=0,\\
\vv_{t+1} - \vv_{t+6} &=0 \}
    \end{align*}


Let $\zz = \left(\uu_{j_k},\uu_{l_k}, \uu_{t}, \uu_{t+1}, \uu_{t+2},
  \uu_{t+3}, \ldots, \uu_{t+6},  
\vv_{t}, \vv_{t+1},\ldots, \vv_{t+6}\right) \in
\mathbb{R}^{16}$.
Let $\GG \in \mathbb{R}^{10 \times 16} $ be the matrix s.t. $\GG \zz =\zero$ corresponds to the constraints
listed above.
Note that all coefficients of $\uu_{t+1}, \uu_{t+2}$ and $\vv_{t}$ are
zero. We set these three variables to zero.


For some $\eps$, which we will fix later,
we choose $\uu_t$ such that 
\[
\uu_t
=
\frac{1}{2} \left( \uu_{j_k} + \uu_{l_k}\right) + 5\eps.
\]
Again, for the same $\eps$, we fix the following values for 
$\uu_{t+i}, \vv_{t+j}, 3\le i\le 6, 1\le j\le 6$ 
\begin{align*}
\uu_{t+3} & = \uu_{t} - \eps, \\
\uu_{t+4} & =  \uu_{j_k} + \eps, \\
\uu_{t+5} & = \uu_{t} - \eps, \\
\uu_{t+6} & =  \uu_{l_k}  + \eps, \\
\vv_{t+1} &= 0, \\
\vv_{t+2} &= 5\eps -(\uu_{t} - \uu_{j_k}), \\
\vv_{t+3} &= \eps, \\
\vv_{t+4} &= 4\eps -(\uu_{t} - \uu_{j_k}), \\
\vv_{t+5} &= 6\eps -(\uu_{t} - \uu_{j_k}), \\
\vv_{t+6} &= -\eps.
\end{align*}
This ensures $\GG\zz = \eps \one$.
Note that for some $r$,
$\BBhat_{i} \xxb =2^r \GG\zz =  2^r \eps \one$, 
so by choosing
$\eps = 2^{-r} \frac{1}{(\alpha+1)m_i} \eps_i$, we can
ensure~Equations~\eqref{eqn:error_ls_b} are satisfied.



Also
\[
\cca_i + \eps_i = \one^{\trp} \left( \begin{array}{c}
\AAhat_i \\
\BBhat_i
\end{array} \right) \xxb
= \AAhat_i \xxb + \one^{\trp} \BBhat_i \xxb
= \AAhat_i \xxb + \frac{1}{\alpha+1} \eps_i.
\]
Thus, we have
\[
\AAhat_i \xxb = \cca_i + \frac{\alpha}{\alpha+1} \eps_i,
\]
which is~Equation~\ref{eqn:error_ls_a}.

%

This completes the proof of the claim.
\end{proof}

{\bf Remark.} The optimal solutions for $\min_{\xxb} \norm{\BB\xxb -
  \ccb}_2$ and $\min_{\xxa} \norm{\AA\xxa - \cca}_2$ have a
one-to-one map, however, the optimal values are different:
%

From the proof of Claim~\ref{clm:exactReduction} and Equation~\eqref{eqn:minquad}
\[
\norm{\cca - \PPi_{\AA}\cca}_2^2
 =
 \left(1 + \frac{1}{\alpha} \right) 
 \norm{\ccb - \PPi_{\BB} \ccb}_2^2,
\]
where we set the weight $w_i = \alpha m_i, \forall i$.
Note when $\alpha \rightarrow \infty$, the two optimal values approach
the same value.

\subsection{Relationship Between Schur Complements}
\label{subsec:SCRelation}


\begin{definition}[Schur Complement]
Let $\CC \in \mathbb{R}^{n \times n}$ be a symmetric PSD matrix. Write
$\CC = \left(\begin{array}{cc}
\CC_{11} & \CC_{12} \\
\CC_{12}^{\trp} & \CC_{22}
\end{array} \right)$ where $\CC_{11}, \CC_{12}, \CC_{12}$ are block matrices, 
the \emph{Schur complement} of $\CC$ is 
\[
\CC_S \defeq \CC_{11} - \CC_{12}\CC^{-1}_{22}\CC_{12}^{\trp}.
\]
If $\CC_{22}$ is not invertible, then we use pseudo-inverse, that is,
$\CC_S \defeq \CC_{11} - \CC_{12}\CC^{\dagger}_{22}\CC_{12}^{\trp}$.
\end{definition}


The following fact is important for Schur complement.
\begin{fact}[Schur complement]
For any fixed vector $\xx$, 
\[
\min_{\yy} \left( \begin{array}{cc}
\xx^{\trp} & \yy^{\trp}
\end{array} \right) \CC \left( \begin{array}{c}
\xx \\
\yy
\end{array} \right) 
= \xx^{\trp} \CC_S\xx.
\]
\label{clm:schur}
\end{fact}
\begin{proof}
We expand the left hand side,
\begin{align}
\left( \begin{array}{cc}
\xx^{\trp} & \yy^{\trp}
\end{array} \right) \CC \left( \begin{array}{c}
\xx \\
\yy
\end{array} \right) 
= \xx^{\trp}\CC_{11}\xx
+ 2\xx^{\trp} \CC_{12}\yy + \yy^{\trp} \CC_{22} \yy.
\label{eqn:schur_lhs}
\end{align}
Taking derivative w.r.t. $\yy$ and setting it to be 0 give that
\[
2\CC_{22}\yy + 2 \CC_{12}^{\trp} \xx = {\bf 0}.
\]
Plugging $\yy = - \CC_{22}^{\dagger} \CC_{12}^{\trp} \xx$ into~\eqref{eqn:schur_lhs}, 
\[
\min_{\yy} \left( \begin{array}{cc}
\xx^{\trp} & \yy^{\trp}
\end{array} \right) \CC \left( \begin{array}{c}
\xx \\
\yy
\end{array} \right) 
= \xx^{\trp}\CC_{11}\xx - \xx^{\trp} \CC_{12}\CC_{22}^{\dagger} \CC_{12}^{\trp} \xx
= \xx^{\trp} \CC_S \xx.
\]
This completes the proof.
\end{proof}

Recall that $\BB$ is the output coefficient matrix of Algorithm~\ref{alg:GZ2toMC2}.
Write $\BB = \left( \begin{array}{cc}
\BB_1 & \BB_2
\end{array} \right) $, where $\BB_1$ is the submatrix corresponding to $\xxa$ and $\BB_2$ is the submatrix corresponding to $\xxaux$. Then,
\[
\BB^{\trp}\BB = \left( \begin{array}{cc}
\BB_1^{\trp} \BB_1 & \BB_1^{\trp} \BB_2 \\
\BB_2^{\trp} \BB_1 & \BB_2^{\trp} \BB_2
\end{array} \right).
\]

\begin{claim}
$\frac{\alpha}{\alpha+1}\AA^{\trp}\AA$ is the Schur complement of $\BB^{\trp}\BB$.
\label{clm:schur_b}
\end{claim}
\begin{proof}
By~\eqref{clm:exactReduction}, for any fixed $\xxa$, 
\[
\min_{\xxaux} \xx^{\trp} \BB^{\trp} \BB \xx
= \frac{\alpha}{\alpha+1} \left(\xxa\right)^{\trp} \AA^{\trp} \AA \xxa.
\]
By Claim~\ref{clm:schur}, 
$\frac{\alpha}{\alpha+1}\AA^{\trp}\AA$ is the Schur complement of $\BB^{\trp}\BB$.
\end{proof}



\newcommand{\zzn}{\tilde{\zz}}
\newcommand{\zzp}{\hat{\zz}}

\subsection{Approximate solvers}
\label{subsec:GZ2ToMC2apx}

We now show that approximate solvers for $\BB$ also
translate to approximate solvers for $\AA$.
The following Lemma about the length of projections
involving integral matrices is crucial for our bounds.

\begin{lemma}
\label{lem:ProjLengthLower}
Let $\AA$ and $\cc$ a matrix and a vector 
such that $\norm{\AA\cc}_2^2$ is integral.
If $\AA^{\top} \cc \neq {\bf 0}$,
then
\[
\norm{\PPi_{\AA}\cc}_2^2 
\geq \frac{1}{\sigma^{2}_{\max} \left( \AA \right)}
= \frac{1}{\lambda_{\max}\left( \AA^{\trp} \AA \right)}.
\]
\end{lemma}

\begin{proof}
As $\norm{\AA\cc}_2^2$ is integral, the condition of
$\AA^{\trp} \cc \neq 0$ implies $\norm{\AA^{\trp}\cc}_2 \ge 1$.
Consider the SVD of $\AA$:
\[
\AA = \sum_i \sigma_i \rr^i (\ss^i)^{\trp}.
\]
Substituting it in for $\AA^{\trp}$ gives
\[
\AA^{\trp} \cc = \sum_i \left( \sigma_i (\rr^i)^{\trp} \cc \right) \ss^i,
\]
and we will use $\alpha_i = \sigma_i (\rr^i)^{\trp} \cc$ to denote
the coefficients of $\cc$ against the singular vectors of $\AA$.
Note that $\sigma_i = 0$ implies $\alpha_i = 0$.
The norm condition on $\AA^{\trp} \cc$ also means
$\sum_i \alpha_i^2 = \sum_{i : \sigma_i \neq 0} \alpha_i^2 \ge 1$, which then gives
\[
\norm{\PPi_{\AA}\cc}_2^2 
= \cc^{\trp}\AA \left( \AA^{\trp}\AA \right)^{\dagger} \AA^{\trp} \cc \\
 = \left( \sum_i \alpha_i (\ss^i)^{\trp} \right)
\left( \sum_{i : \sigma_i \neq 0} \frac{1}{\sigma_i^2} \ss^i (\ss^i)^{\trp} \right)
\left( \sum_i \alpha_i \ss^i \right) \\
 = \sum_{i : \sigma_i \neq 0} \frac{\alpha_i^2}{\sigma_i^2} \\
 \ge \frac{1}{\sigma_{\max}^2(\AA)}.
\]
This completes the proof.
\end{proof}

\begin{lemma}
\label{lem:mctogenerr}
Let $\BB\xxb = \ccb$ be the linear system returned by a call to
$\algGZTtoMCT(\AA, \cca, \epsa)$, and let $\epsb$ be the
the error parameter returned by this call.
Let $\xxb$ be a vector such that
\begin{align}
\norm{\BB\xxb - \PPi_{\BB} \ccb}_2 \le \epsb \norm{\PPi_{\BB}\ccb}_2.
\label{eqn:error_b}
\end{align}
Let $\xxatil$ be the vector returned by a call to $\algMCTtoGZTsoln(\AA, \BB, \cca, \xxb)$. Then,
\[
\norm{\AA\xxatil - \PPi_{\AA} \cca}_2 \le \epsa \norm{\PPi_{\AA} \cca}_2.
\label{eqn:error_a}
\]
\end{lemma}

\begin{proof}
We first consider the case $\AA^{\trp} \cca = {\bf 0}$, then Algorithm~\ref{alg:GZ2toMC2solnback} returns $\xxatil = {\bf 0}$, which gives
\[
\norm{\AA\xxatil - \PPi_{\AA}\cca}_2 = {\bf 0},
\]
and completes the proof for this case.
When $\AA^{\trp} \cca \neq \zero$, the
Algorithm~\ref{alg:GZ2toMC2solnback} returns  $\xxatil = \xxa$, where
$\xxb =
\begin{pmatrix}
  \xxa \\
  \xxaux
\end{pmatrix}
$.
Let 
\[
\xxbopt
\defeq
\left( \begin{array}{c}
\xxaopt\\ 
\xxauxopt
\end{array} \right)
\in
\argmin_{\xx} \norm{\BB\xx - \ccb}_2.
\]
By Fact~\ref{fac:projIsQFmin}, the condition \eqref{eqn:error_b} is equivalent to
\[
\norm{\xxb - \xxbopt}_{\BB^{\trp}\BB} \le \epsb \norm{\PPi_{\BB} \ccb}_2.
\label{eqn:diff_b}
\]
By Claim~\eqref{clm:exactReduction}, 
$\xxaopt \in \argmin_{\xx} \norm{\AA\xx - \cca}_2$. 

We now upper bound the error of $\norm{\xxa - \xxaopt}_{\AA^{\trp}\AA}$.
%
Note 
\[
\xxb - \xxbopt
= \left( \begin{array}{c}
\xxa - \xxaopt \\ 
\xxaux - \xxauxopt 
\end{array} \right).
\]
By Claim~\ref{clm:schur}  and Claim~\ref{clm:schur_b}, 
\[
\norm{\xxa - \xxaopt}_{\AA^{\trp}\AA}^2
 \le 
\left(1 + \frac{1}{\alpha} \right) \norm{\xxb - \xxbopt }_{\BB^{\trp}\BB}^2.
\]
Then, we lower bound $\norm{\PPi_{\AA}\cca}_2$.
By Claim~\ref{clm:exactReduction}, 
\[
\norm{\cca - \PPi_{\AA} \cca}_2^2
= \left( 1 + \frac{1}{\alpha} \right) 
\norm{\ccb - \PPi_{\BB} \ccb}_2^2.
\]
Since $\PPi_{\AA} \cca \perp \cca - \PPi_{\AA} \cca$, and $\PPi_{\BB} \ccb \perp \ccb - \PPi_{\BB} \ccb$, we have
\begin{align*}
\norm{\PPi_{\AA}\cca}_2^2 &= \norm{\cca}_2^2 - \norm{\cca - \PPi_{\AA} \cca}_2^2 \\
& = \norm{\ccb}_2^2 - \norm{\ccb - \PPi_{\BB} \ccb}_2^2 - \frac{1}{\alpha+1} \norm{\cca - \PPi_{\AA}\cca}_2^2 \\
& = \norm{\PPi_{\BB} \ccb}_2^2  - \frac{1}{\alpha+1} \norm{\cca - \PPi_{\AA}\cca}_2^2. 
\end{align*}
Thus, 
\[
\norm{\PPi_{\BB} \ccb}_2^2 
= \norm{\PPi_{\AA}\cca}_2^2
+ \frac{1}{\alpha+1} \norm{\cca - \PPi_{\AA}\cca}_2^2.
\]
Since $\AA^{\trp} \cca \neq \zero$,
and because $\AA^{\trp} \cca$ is integral, 
by Lemma~\ref{lem:ProjLengthLower}, we have
\[
\norm{\PPi_{\AA}\cca}_2^2 
 \ge \frac{1}{\sigma_{\max}^2(\AA)}.
\]
Thus,
\[
\norm{\PPi_{\BB} \ccb}_2^2 
\le
\norm{\PPi_{\AA}\cca}_2^2
+ \frac{1}{\alpha+1} \norm{\cca}_2^2
\le
\left( 1 + \frac{\norm{\cca}_2^2 \sigma^2_{\max}(\AA)}{\alpha+1} \right)
\norm{\PPi_{\AA}\cca}_2^2.
\]
Therefore, 
\[
\norm{\AA\xxa - \PPi_{\AA}\cca}_2 
\le \left(1 + \frac{1}{\alpha} \right)^{1/2} \left( 1 + \frac{\norm{\cca}_2^2 \sigma^2_{\max}(\AA)}{\alpha+1} \right)^{1/2}  \epsb
\norm{\PPi_{\AA}\cca}_2.
\]
Finally, by the setting
\[
\epsb =
\frac{\epsa}
{
\left(1 + \frac{1}{\alpha} \right)^{1/2}
\left( 1 + \frac{\norm{\cca}_2^2 \sigma^2_{\max}(\AA)}{\alpha+1}
\right)^{1/2}
} ,
\]
we have
\[
\norm{\AA\xxa - \PPi_{\AA}\cca}_2 
\le \epsa
\norm{\PPi_{\AA}\cca}_2.
\]
This completes the proof.
\end{proof}

\subsection{Bounding Condition Number of the New Matrix}
\label{subsec:GZToMC2ConditionNumber}

In this section, we show that the condition number of $\BB$ is upper bounded by the condition number of $\AA$ with a poly($n$) multiplicative factor.

We first characterize the null space of $\BB$.
%
%
Recall that in Equation~\eqref{eq:MC2BDecomposeX}, we write  $\xxb = \left(  \begin{array}{c}
\xxa \\
\xxaux
\end{array}
\right)$. 
In the following, we will employ a \emph{different indexing} of
$\xxaux$ than the one defined at the beginning of
Section~\ref{sec:GZ2toMC2}.
We also reorder the columns of $\BB$ so that $\BB \xxb
= \zero$ represents the same equations as before.
For appropriately chosen indices $g_1$ and $g_2$, we define
\begin{enumerate}
\item $\xxaux_{1:g_1}$ corresponds to the $\uu$-coordinates of the auxiliary variables created in $\mctwoCl$-gadgets \emph{whose coefficients are nonzero}.
\item $\xxaux_{g_1+1:g_2}$ corresponds to the $\vv$-coordinates of the auxiliary variables created in $\mctwoCl$-gadgets  \emph{whose coefficients are nonzero}.
\item $\xxaux_{g_2+1: \nb-\na}$ corresponds to the coordinates of the auxiliary variables created in $\mctwoCl$-gadgets \emph{whose coefficients are zero}.
\end{enumerate}
Using this ordering, for $0 \le i \le (g_2-g_1)/6$,
$\xxaux_{g_1+6i+1: g_1+6i+6}$ corresponds to the $\vv$-coordinates of
the auxiliary variables with non-zero coefficients in a single $\mctwoCl$-gadget.

Given any fixed $\xxa \in \nulls(\AA)$, we extend $\xxa$ to a vector in dimension $\nb$:
\begin{align}
\ppxa \defeq \left( \begin{array}{c}
\xxa \\
\xxaux
\end{array} \right).
\label{eqn:GZ2toMC2apx_pxa_def}
\end{align}
We assign the values of the auxiliary variables of $\ppxa$ in the order that they are created in Algorithm~\ref{alg:GZ2toMC2}. 
In a $\mctwoCl$-gadget, suppose the values of variables $\uu_{j_k}$ and $\uu_{l_k}$ have already been assigned. Let $\uu_t, \uu_{t+i}, \vv_{t+j}, 3 \le i \le 6, 1\le j \le 6$ be the auxiliary variables created in this gadget.
We assign values as follows,
\begin{align*}
\uu_t & = (\uu_{j_k} + \uu_{l_k}) / 2, \\
\uu_{t+3} & = \uu_t.  \\
\uu_{t+4} & = \uu_{j_k}, \\
\uu_{t+5} & = \uu_t, \\
\uu_{t+6} & = \uu_{l_k}, \\
\vv_{t+1} &= 0, \\
\vv_{t+2} &= \uu_t - \uu_{j_k}, \\
\vv_{t+3} &= 0, \\
\vv_{t+4} &= \uu_t - \uu_{j_k}, \\
\vv_{t+5} &= \uu_t- \uu_{j_k}, \\
\vv_{t+6} &= 0.
\end{align*}
This gives the first $g_2$ entries of $\xxaux$, and we set all the rest of the entries to be 0.

Let $\ee_i \in \mathbb{R}^{(g_2-g_1)/6}$ be the $i$th standard basis vector and $\one$ be the all-one vector in $6$ dimensions.
We define $\vecppi \in \mathbb{R}^{\nb}$ to be a vector whose nonzero entries are given by
\begin{align}
\vecppi_{\na+g_1+1:\na+g_2} = \ee_i \otimes \one.
\label{eqn:GZ2toMC2apx_pi}
\end{align}
Let $\ee_j \in \mathbb{R}^{{\nb -\na - g_2}}$ is the $j$th standard basis vector.
We define $\qqj \in \mathbb{R}^{\nb}$ to be a vector whose nonzero entries are given by 
\begin{align}
\label{eqn:GZ2toMC2apx_qj}
\qqj_{\na+g_2+1:\nb} = \ee_j.
\end{align}

\begin{lemma}
\label{lem:GZtoMC2Null}
$\nulls(\BB) = \Span{ \ppxa, \vecppi, \qqj: \xxa \in \nulls(\AA), 1 \le i \le \frac{g_2-g_1}{6}, 1\le j \le \nb-g_2}$.
\end{lemma}
\begin{proof}
Let $\SS$ be the subspace of $\Span{\ppxa, \vecppi, \qqj: \xxa \in \nulls(\AA), 1 \le i \le \frac{g_2-g_1}{6}, 1\le j \le \nb-g_2}$.

According to definitions of the vectors, we can check that $\SS \subseteq \nulls(\BB)$.

It remains to show that $\SS \supseteq \nulls(\BB)$.
Let $\xx \in \nulls(\BB)$.
By Claim~\ref{clm:exactReduction} with $\cca = {\bf 0}$, we have
\[
\AA\xxa = {\bf 0},
\]
that is, $\xxa \in \nulls(\AA)$.
According to the $\mctwoCl$-gadget constraints in Algorithm~\ref{alg:MC2gadget}, we have
\[
\vv_{t+1} = \vv_{t+3} = \vv_{t+6} = \gamma, 
\mbox{ and }
\vv_{t+2} = \vv_{t+4} = \vv_{t+5} = \theta, 
\]
where $\gamma - \theta = (\uu_{j_k} - \uu_{j_l})/2$.
Besides, since all entries in $\xx_{\na+g_2+1:\nb}$ have zero coefficients, they are free to choose.
Thus, $\xx \in \SS$, that is, $\SS \supseteq \nulls(\BB)$.

This completes the proof.
\end{proof}

We upper bound the largest singular value of $\BB$ in the following claim.

\begin{lemma}
\label{lem:GZToMC2LambdaMax}
$\lambda_{\max} \left( \BB^{\trp} \BB \right) = O \left( \alpha \nnz(\AA)^2 \norm{\AA}_{\infty}^2 \log^{2} \norm{\AA}_{\infty} \right)
$.
\end{lemma}

\begin{proof}
By the Courant-Fischer Theorem, 
\[
\lambda_{\max} \left(\BB^{\trp} \BB \right)
= \max_{\xx} \frac{\xx^{\trp} \BB^{\trp}  \BB \xx}{\xx^{\trp}\xx}.
\]

We write $\xx$ as $(\uu_1, \vv_2, \ldots, \uu_{\nb}, \vv_{\nb})$, so that each $(\uu_i, \vv_i)$ corresponds to the two coordinates of vertex $i$ in the graph. 
Expanding the right hand side, we get
\[
\frac{\sum_{(i,j) \in E_1} w_{(i,j)}(\uu_i - \uu_j)^2 + \sum_{(i,j) \in E_2} w_{(i,j)}(\vv_i - \vv_j)^2+ \sum_{(i,j) \in E_{1+2}} w_{(i,j)}(\uu_i - \vv_i - (\uu_j - \vv_j))^2 }{\sum_i \uu_i^2 + \vv_i^2}
\]
where $w_{(i,j)}$'s are the edge weights.
Let 
\[
w_{\max} = \max_{(i,j) \in E_1 \cup E_2 \cup E_{1+2}} w_{(i,j)}.
\]
According to our construction of $\BB$ in Algorithm~\ref{alg:GZ2toMC2}, 
\begin{align}
w_{\max} \le \alpha \nnz(\BB) \norm{\AA}_{\infty}^2.
\label{eqn:eig_b_max_weight}
\end{align}
Thus, $\lambda_{\max} \left( \BB^{\trp}  \BB \right)$ is at most 
\[
 w_{\max} \max_{\uu}  
\frac{\sum_{(i,j) \in E_1} (\uu_i - \uu_j)^2 + \sum_{(i,j) \in E_2}(\vv_i - \vv_j)^2+ \sum_{(i,j) \in E_{1+2}} (\uu_i - \vv_i - (\uu_j - \vv_j))^2 }{\sum_i \uu_i^2 + \vv_i^2}.
\]
We upper bound each term of the numerator.
Let $d_1, d_2, d_{1+2}$ be the maximum vertex degree of the graphs $G_1, G_2, G_{1+2}$, respectively.
By Cauchy-Schwarz inequality,
\[
\sum_{(i,j) \in E_1} (\uu_i - \uu_j)^2
\le 2 \sum_{(i,j) \in E_1} \uu_i^2 + \uu_j^2
\le 2 d_1  \left( \sum_i \uu_i^2 + \vv_i^2 \right),
\]
\[
\sum_{(i,j) \in E_2} (\vv_i - \vv_j)^2
\le 2 d_2  \left( \sum_i \uu_i^2 + \vv_i^2 \right),
\]
and
\begin{align*}
\sum_{(i,j) \in E_{1+2}} (\uu_i - \vv_i - (\uu_j - \vv_j))^2 
& \le 4  \sum_{(i,j) \in E_{1+2}} \uu_i^2 + \vv_i^2 + \uu_j^2 + \vv_j^2 \\
& \le 4d_{1+2} \left( \sum_i  \uu_i^2 + \vv_i^2 \right).
\end{align*}
Plugging the above inequalities to the expansion, we have
\begin{align*}
\lambda_{\max}\left( \BB^{\trp}  \BB \right)
&\le w_{\max} \cdot 8 \max \{ d_1, d_2, d_{1+2}\} \\
&\le 8w_{\max} \nnz(\BB).
\end{align*}
By the upper bound of $\nnz(\BB)$ in Lemma~\ref{lem:GZToMC2NNZ} and the upper bound of $w_{\max}$ in Equation~\eqref{eqn:eig_b_max_weight}, we have
\[
\lambda_{\max}\left( \BB^{\trp}  \BB \right)
\le O \left( \alpha \nnz(\AA)^2 \norm{\AA}_{\infty}^2 \log^{2} \norm{\AA}_{\infty} \right).
\]
This completes the proof.
\end{proof}


Recall that, for a vector $\xx \in \mathbb{R}^{\nb}$,
$\xx_{\na+1:\na+g_1}$ corresponds to the auxiliary $\uu$-variables
with non-zero coefficients in $\mctwoCl$-gadgets, $\xx_{\na+1+g_1+6i:
  \na+6+g_1+6i}$ for $0 \le i \le (g_2-g_1)/6$ corresponds to the
auxiliary $\vv$-variables with non-zero coefficients in a single $\mctwoCl$-gadget, and
$\xx_{n+g_2+1:\nb}$ corresponds to the variables with zero coefficient.

\begin{lemma}
Let $0 \le \epsilon \le \left( 100 \nb \log \norm{\AA}_{\infty} \right)^{-1}$.
Let $\xxb \in \mathbb{R}^{\nb}$ satisfying
\begin{enumerate}
\item $\xxb_{\na+g_2+1:\nb} = {\bf 0}$,
\item $\xxb_{\na+g_1+6i+1} = 0, \forall 0 \le i \le (g_2-g_1)/6$, and
\item $\norm{\BB\xxb}_2 \le \epsilon \norm{\xxb}_2$.
\end{enumerate}
Then,
\[
\norm{\AA\xxa}_2 \le 
4 \left(1 + \frac{1}{\alpha} \right)^{1/2} \nb^{1/2} \epsilon
\norm{\xxa}_2.
\]
\label{lem:GZ2toMC2apx_small_eig}
\end{lemma}
\begin{proof}
We first show that under the conditions of the Lemma, 
\begin{align}
 \norm{\xxa}_{\infty} \ge \frac{1}{4} \norm{\xxb}_{\infty}.
\label{eqn:GZ2toMC2_xa}
\end{align}

Let $\delta \defeq \epsilon \norm{\xxb}_2$.
By the 3rd condition, each entry of $\BB\xxb$ has absolute value at most $\delta$.

We first show that all $\uu$-variables in $\xxaux$  
cannot be large. 
According to Algorithm~\ref{alg:GZ2toMC2}, for each row of $\BB$, all nonzero coefficients have same absolute value, which is at least 1. 
Based on this fact and the $\mctwoCl$-gadget constructed in Algorithm~\ref{alg:MC2gadget}, we have
\begin{align*}
\frac{\uu_{j_k} + \uu_{l_k}}{2} - 5 \delta
&\le  \uu_t \le \frac{\uu_{j_k} + \uu_{l_k}}{2} + 5 \delta \\
\frac{\uu_{j_k} + \uu_{l_k}}{2} - 6\delta 
&\le \uu_{t+3} \le \frac{\uu_{j_k} + \uu_{l_k}}{2} + 6\delta \\
\uu_{j_k} - \delta & \le \uu_{t+4} \le \uu_{j_k} + \delta \\
\frac{\uu_{j_k} + \uu_{l_k}}{2} - 6\delta 
&\le \uu_{t+5} \le \frac{\uu_{j_k} + \uu_{l_k}}{2} + 6\delta \\
\uu_{l_k} - \delta & \le \uu_{t+6} \le \uu_{l_k} + \delta 
\end{align*}
where $\uu_{j_k}, \uu_{l_k}$ being paired-and-replaced,
and all others are entries of $\xxaux$.
By the triangle inequality,
\[
\abs{\uu_t}, \abs{\uu_{t+k}} \le \frac{1}{2} \left( \abs{\uu_{j_k}} + \abs{\uu_{l_k}} \right) + 6 \delta, \quad \forall 3 \le k \le 6.
\]
Note that the sum of coefficients on both sides are equal. 
We can repeat this type of substitution on the right hand side until $\uu_{j_k}$ and $\uu_{l_k}$ are variables of $\xxa$. 
At the $i$th iteration of Algorithm~\ref{alg:GZ2toMC2} line~\ref{lin:varReplace}, we have
\[
\abs{\uu_t}, \abs{\uu_{t+k}} \le 
\sum_j \alpha_j \abs{\xxa_j} + 6r \delta,
\qquad  \forall 3 \le k \le 6,
\]
where $\alpha_j \ge 0$ and $\sum_j \alpha_j = 1$.
By the H\"{o}lder inequality,
\begin{align}
\abs{\uu_t}, \abs{\uu_{t+k}} \le 
\norm{\xxa}_{\infty} + 6r \delta,
\qquad \forall 3 \le k \le 6.
\label{eqn:GZ2toMC2apx_xi}
\end{align}

We then argue that all $\vv$-variables in $\xxaux$ 
cannot be large. 
Note that by the 2nd condition, in the $\mctwoCl$-gadget, we have $\vv_{t+1} = 0$.
According to the equations in Algorithm~\ref{alg:MC2gadget}, we have
\begin{align*}
\uu_{t+3} - \uu_{t+4} - 3\delta & \le \vv_{t+2} \le \uu_{t+3} - \uu_{t+4} + 3\delta \\
-\delta & \le \vv_{t+3}  \le \delta \\
\uu_{t+3} - \uu_{t+4} - 2\delta & \le \vv_{t+4} \le \uu_{t+3} - \uu_{t+4} + 2\delta \\
\uu_{t+3} - \uu_{t+4} - 4\delta & \le \vv_{t+5} \le \uu_{t+3} - \uu_{t+4} + 4\delta \\
-\delta & \le \vv_{t+6} \le \delta
\end{align*}
By the triangle inequality,
\[
\abs{\vv_{t+k}} \le \abs{\uu_{t+3}} + \abs{\uu_{t+4}} + 4 \delta, 
\qquad \forall 2 \le k \le 6.
\]
By~\eqref{eqn:GZ2toMC2apx_xi}, at the $r$th iteration of Algorithm~\ref{alg:GZ2toMC2} line~\ref{lin:varReplace},
\[
\abs{\vv_{t+k}} \le 2\norm{\xxa}_{\infty} + 12r\delta, \qquad \forall 2 \le k \le 6.
\]
Since there are at most $\log \norm{\AA}_{\infty}$ iterations, the above inequality together with Equation~\eqref{eqn:GZ2toMC2apx_xi} implies 
\[
\norm{\xxaux}_{\infty} \le 2\norm{\xxa}_{\infty} + 12 \delta \log \norm{\AA}_{\infty}.
\]
Adding $\norm{\xxa}_{\infty}$ on both sides and substituting $\delta = \epsilon \norm{\xxb}_2$ gives
\begin{align*}
\norm{\xxb}_{\infty} &\le 3\norm{\xxa}_{\infty} + 12 \epsilon \norm{\xxb}_2 \log \norm{\AA}_{\infty} \\
&\le 3 \norm{\xxa}_{\infty}
 + 12 \epsilon \sqrt{\nb} \norm{\xxb}_{\infty} \log \norm{\AA}_{\infty}.
\end{align*}
Given $\epsilon \le \left( 100 \sqrt{\nb} \log \norm{\AA}_{\infty} \right)^{-1}$, we have Equation~\eqref{eqn:GZ2toMC2_xa}.
This further implies 
\begin{align}
\norm{\xxb}_2 \le 4 \sqrt{\nb} \norm{\xxa}_2.
\label{eqn:GZ2toMC2apx_norm}
\end{align}
%
%

By Claim~\ref{clm:exactReduction} with $\ccb = {\bf 0}$, we have
\[
\norm{\AA\xxa}_2 \le \sqrt{\frac{\alpha+1}{\alpha}} \norm{\BB\xxb}_2.
\]
By the 3rd condition and the bound~\eqref{eqn:GZ2toMC2apx_norm}, we have
\[
\norm{\AA\xxa}_2 \le 4 \left(1 + \frac{1}{\alpha} \right)^{1/2} \nb^{1/2} \epsilon  \norm{\xxa}_2.
\]
This completes the proof.
\end{proof}

We use the above lemma to lower bound the smallest nonzero eigenvalue of $\BB^{\trp}\BB$.

\begin{lemma}
$\lambda_{\min} \left(\BB^{\trp}\BB \right) = \Omega \left(\frac{\min\{1, \lambda_{\min}(\AA^{\trp}\AA) \}}{\nb^2} \right)$.
\label{lem:GZToMC2LambdaMin}
\end{lemma}

\begin{proof}

Let 
\begin{align}
\delta \defeq \min\{1, \lambda_{\min}(\AA^{\trp}\AA) \},
\label{eqn:GZ2toMC2apx_delta_def}
\end{align}
and
\begin{align}
\epsilon \defeq \frac{\delta}{1616(1+\alpha^{-1}) \nb^2}.
\label{eqn:GZ2toMC2apx_eps_def}
\end{align}

The goal is to prove that $\lambda_{\min} (\BB^{\trp}\BB) \ge \epsilon$.
Assume by contradiction, there exists an $\xx \perp \nulls(\BB)$ such that
\begin{align}
(\xxb)^{\trp}\BB^{\trp}\BB\xxb < \epsilon (\xxb)^{\trp}\xxb.
\label{eqn:abc}
\end{align}
We show a contradiction by case analysis according to whether $\xxa$ is orthogonal to the null space of $\AA$.

\noindent 
{\bf Case 1:}
Suppose $\xxa \perp \nulls(\AA)$.

Recall that $\xxb_{\na+1: \na + g_1}$ corresponds to the auxiliary $\uu$-variables in the constraints, $\xxb_{\na+g_1+1: \na+g_2}$ corresponds to the auxiliary $\vv$-variables in the constraints, and $\xxb_{\na + g_2+1: \nb}$ corresponds to the variables with zero coefficient.
By Lemma~\ref{lem:GZtoMC2Null}, 
we have $\xxb \perp \qqj, \forall 1 \le j \le \na - g_2$, where $\qqj$ is defined in~\eqref{eqn:GZ2toMC2apx_qj}.
Thus, $\xx_{\na+g_2+1:\nb} = {\bf 0}$.

Let $\rr \defeq \sum_{0 \le i \le (g_2-g_1)/6} \xxb_{\na+g_1+6i+1} \vecppi$, where $\vecppi$ is defined in~\eqref{eqn:GZ2toMC2apx_pi}.
By Lemma~\ref{lem:GZtoMC2Null}, $\rr \in \nulls(\BB)$. 
Let $\yyb \defeq \xxb - \rr$.
Similarly, we write 
\[
\yyb = \left( \begin{array}{c}
\yya \\
\yyaux
\end{array} \right),
\]
where $\yya$ corresponds to original variables and $\yyaux$ corresponds to auxiliary variables.
Note $\yya = \xxa$.
Since $\rr \in \nulls(\BB)$, we have
\[
(\yyb)^{\trp}\BB^{\trp}\BB\yyb = (\xxb)^{\trp}\BB^{\trp}\BB\xxb.
\]
Since $\xxb \perp \rr$, we have
\[
(\yyb)^{\trp}\yyb  
= (\xxb)^{\trp}\xxb + \rr^{\trp}\rr
\ge (\xxb)^{\trp}\xxb.
\]
By assumption~\eqref{eqn:abc}, we have
\[
(\yyb)^{\trp}\BB^{\trp}\BB\yyb < \epsilon (\yyb)^{\trp}\yyb.
\]
By our definition of $\yyb$, $\yyb$ satisfies the conditions of Lemma~\ref{lem:GZ2toMC2apx_small_eig}. By Lemma~\ref{lem:GZ2toMC2apx_small_eig} and the definition of $\epsilon$ in Equation~\eqref{eqn:GZ2toMC2apx_eps_def}, we have
\[
\left( \xxa \right)^{\trp} \AA^{\trp} \AA \xxa <
\delta \left( \xxa \right)^{\trp} \xxa.
\]
Since $\xxa \perp \nulls(\AA)$ and the definition of $\delta$ in Equation~\ref{eqn:GZ2toMC2apx_delta_def}, we have
\[
\lambda_{\min} \left( \AA^{\trp} \AA \right)
< \delta = \min\left\{ 1, \lambda_{\min}(\AA^{\trp}\AA) \right\},
\]
which is a contradiction.

\noindent
{\bf Case 2:}
Suppose $\xxa \in \nulls(\AA)$. 

Recall that $\ppxa$ is the extension of vector $\xxa$ defined in Equation~\eqref{eqn:GZ2toMC2apx_pxa_def}.
Let $\rr = \sum_{0 \le i \le (g_2-g_1)/6} \alpha_i \vecppi$ be a vector such that 
\[
\yyb \defeq \xxb - \ppxa - \rr
\] 
satisfying $\yyb_{\na+g_1+6i+1} = {\bf 0}, \forall 0 \le i < (g_2-g_1)/6$.
By this definition, we have 
\begin{align}
\yya = {\bf 0}.
\label{eqn:GZ2toMC2apx_xa2}
\end{align}
Since $\ppxa + \rr$ is in the null space of $\BB$,
by assumption~\eqref{eqn:abc},
\[
(\yyb)^{\trp}\BB^{\trp}\BB\yyb < \epsilon (\yyb)^{\trp}\yyb.
\]
Note $\yyb$ satisfies the conditions of Lemma~\ref{lem:GZ2toMC2apx_small_eig}.
By Equations \eqref{eqn:GZ2toMC2_xa}
and~\eqref{eqn:GZ2toMC2apx_xa2}, we have 
\[
\norm{\yyb}_{\infty} = 0.
\]
This implies that $\xxb = \ppxa + \rr$ is in the null space of $\BB$, which contradicts that $\xxb \perp \nulls(\BB)$.

\noindent
{\bf Case 3:}
Suppose $\xxa = \zzn + \zzp$, where $\zzn \neq {\bf 0}$ is in $\nulls(\AA)$ and $\zzp \neq {\bf 0}$ is orthogonal to $\nulls(\AA)$. 

Similarly,
let $\rr_1 = \sum_i \alpha_i \vecppi$ such that $\yyb \defeq \xxb - \rr_1$ satisfying $\yyb_{\na+g_1+6i+1} = 0, \forall i$.
By this definition, we have 
\[
\yya = \xxa.
\]
Then by assumption~\eqref{eqn:abc},
\[
(\yyb)^{\trp}\BB^{\trp}\BB\yyb < \epsilon (\yyb)^{\trp}\yyb.
\]
$\yyb$ satisfies the conditions of Lemma~\ref{lem:GZ2toMC2apx_small_eig}.
By Lemma~\ref{lem:GZ2toMC2apx_small_eig}, we have 
\[
(\yya)^{\trp} \AA^{\trp} \AA \yya \le 16 \left(1 + \frac{1}{\alpha} \right) \nb \epsilon (\yya)^{\trp} \yya.
\]
Since $\yya = \zzn + \zzp$ with $\zzn \in \nulls(\AA)$ and $\zzp \perp \nulls(\AA)$, we have
\[
\zzp^{\trp} \AA^{\trp}\AA \zzp \le  16 \left(1 + \frac{1}{\alpha} \right) \nb \epsilon \left( \zzn^{\trp}\zzn + \zzp^{\trp}\zzp \right).
\]
If $\zzp^{\trp}\zzp > \zzn^{\trp}\zzn / (100 \nb)$,  then we have a contradiction with~\eqref{eqn:GZ2toMC2apx_delta_def} and we have done. 
Otherwise, we have
\[
\norm{\zzp}_2^2 \le \frac{\norm{\zzn}_2^2}{100 \nb}.
\]
Let $\rr_2 = \sum_i \beta_i \vecppi$ such that 
\[
\zzb \defeq \yyb - \pp(\zzn) - \rr_2
\] 
satisfies $\zzb_{\na+g_1+6i+1} = 0, \forall i$.
By assumption~\eqref{eqn:abc},
\[
(\zzb)^{\trp}\BB^{\trp}\BB\zzb < \epsilon (\zzb)^{\trp} \zzb.
\]
Vector $\zzb$ satisfies the conditions of Lemma~\ref{lem:GZ2toMC2apx_small_eig}.
By Equation~\eqref{eqn:GZ2toMC2_xa}, we have
\[
\norm{\zzb}_2 \le 4 \sqrt{\nb} \norm{\zzp}_2. 
\]
Since $\zzp^{\trp} \zzp \le \zzn^{\trp}\zzn / (100 \nb)$,
we have 
\[
\norm{\zzb}_2 \le \norm{\zzn}_2.
\]
On the other hand,
\begin{align*}
\norm{\zzb}_2^2
&= \norm{\xxb}_2^2 + \norm{\pp(\zzn)+ \rr_1 + \rr_2}_2^2  \\
& > \norm{\pp(\zzn) + \rr_1 + \rr_2}_2^2 \\
& \ge \norm{\xxa}_2^2 \\
& = \norm{\zzn}_2^2 + \norm{\zzp}_2^2 \\
& >  \norm{\zzn}_2^2.
\end{align*}
The first equality is due to $\xxb \perp \left( \pp(\zzn) + \rr_1 + \rr_2 \right)$, 
the third inequality is due to the $\xxa$ part of $\rr_1 + \rr_2$ is {\bf 0},
and the fourth equality is due to that $\zzn \perp \zzp$.
Thus, we get a contradiction.

This completes the proof.
\end{proof}

Lemma~\ref{lem:GZToMC2LambdaMax} and Lemma~\ref{lem:GZToMC2LambdaMin} immediately imply the following lemma.

\begin{lemma}
$\kappa(\BB) = O \left( \frac{\nnz(A)^2 \norm{\AA}_{\infty} \log^2 \norm{\AA}_{\infty}}{ \min \{1, \sigma_{\min}(\AA) \}} \right)$.
\label{lem:GZ2toMC2apx_cond}
\end{lemma}

\subsection{Putting it All Together}

\begin{proof}[Proof of Lemma~\ref{lem:zeroSumTwoToMcTwo}]
We set $\alpha = 1$ in Algorithm~\ref{alg:GZ2toMC2}.

By Lemma~\ref{lem:GZToMC2NNZ} and Claim~\ref{clm:para_infty_norm}, we have
\[
\nnz(\BB) = O \left( s \log(sU) \right).
\]
According to our reduction in Algorithm~\ref{alg:GZ2toMC2}, the largest entry and the smallest nonzero entry of the right hand side vector does not change.
Besides, all nonzero entries of $\BB$ have absolute value at least 1.
The largest entry of $\BB$ is upper bounded by Equation~\eqref{eqn:eig_b_max_weight}
\[
\sqrt{w_{\max}} = O \left( \sqrt{\nnz(\BB)} \norm{\AA}_{\infty} \right).
\]
By Claim~\ref{clm:para_infty_norm}, we have
\[
\sqrt{w_{\max}} = O \left( s^{3/2} U \log^{1/2} (sU) \right).
\]

By Lemma~\ref{lem:GZ2toMC2apx_cond} and Lemma~\ref{clm:para_infty_norm}, \ref{clm:para_eig}, we have
\[
\kappa(\BB) = O \left( \frac{\nnz(A)^2 \norm{\AA}_{\infty} \log^2 \norm{\AA}_{\infty}}{ \min \{1, \sigma_{\min}(\AA) \}} \right)
= O \left( s^4 U^2 K \log^2 (sU) \right).
\]

By Lemma~\ref{lem:mctogenerr}, we have
\[
(\epsb)^{-1} = (\epsa)^{-1} O \left( 1+ \norm{\cca}_2 \sigma_{\max}(\AA) \right)
= (\epsa)^{-1}  O \left( sU^2 \right).
\]
This completes the proof.
\end{proof}

\newcommand{\algMCTtoMCTS}{\ensuremath{\textsc{Reduce}\mctwoCl\textsc{To}\mctwostrictCl}}
\newcommand{\algMCTStoMCTsoln}{\ensuremath{\textsc{MapSoln}\mctwostrictCl\textsc{To}\mctwoCl}}
\newcommand{\ccbtwostrict}{\cc^{\text{B}^{>0}}}
\newcommand{\xxbtwostrictopt}{\xx^{\text{B}^{>0}*}}
\newcommand{\wstrict}{\delta}

\section{$\mctwoCl$ Efficiently Reducible to $\mctwostrictCl$}
\label{sec:McTwoToMcTwoStrict}

The matrices generated by interior point methods (see Section~\ref{sec:ipm} for details)
are more restrictive than $\mctwoCl$: for every edge present, the weights
of all three types of matrices are non-zero.
Formally, the class of matrices $\mctwostrictCl$ consists of matrices of
the form
\[
\LL^1 \otimes \left( \begin{array}{cc}
1 & 0 \\
0 & 0
\end{array} \right)
+ \LL^2 \otimes \left( \begin{array}{cc}
0 & 0 \\
0 & 1
\end{array} \right)
+ \LL^{1+2} \otimes \left( \begin{array}{cc}
1 & -1 \\
-1 & 1
\end{array} \right),
\]
where the non-zero support of all three matrices, $\LL^{1}$, $\LL^{2}$,
and $\LL^{1 + 2}$ are the same (but the weights may vary greatly).
On the other hand, the matrices that we generate in Section~\ref{sec:GZ2toMC2}
can be transformed into such a matrix by adding a small value,
$\wstrict$ to the weight of all edges where one of the types of edges
 have non-zero support.
We will describe this construction in
Section~\ref{subsec:MCTwoToMCTwoStrictConstruction},
and bound its condition number in Section~\ref{subsec:CondBoundMCTwoStrict}.

\subsection{Construction}
\label{subsec:MCTwoToMCTwoStrictConstruction}

In this section, we show the construction from an instance of $\mctwoCl$ to an instance of $\mctwostrictCl$.
The strategy is to add extra edges with a sufficiently small weight, such that $\LL^1, \LL^2, \LL^{1+2}$ have identical nonzero stricture and the solution of the linear system does not change much.

%
%
%

The reduction from $\mctwoCl$ to $\mctwostrictCl$,
with pseudo-code in Algorithm~\ref{alg:MCTtoMCTS},
simply adds edges with weight $\wstrict$ to all the missing edges.
The transformation of solutions of the corresponding instance of
$\mctwostrictCl$ back to a solution of the original $\mctwoCl$ instance
in Algorithm~\ref{alg:MCTStoMCTsolnback} simply returns the same vector $\xx$.

\begin{algorithm}[H]
\renewcommand{\algorithmicrequire}{\textbf{Input:}}
\renewcommand\algorithmicensure {\textbf{Output:}}
\caption{\label{alg:MCTtoMCTS}\algMCTtoMCTS}
    \begin{algorithmic}[1]
\REQUIRE{$(\BB,\ccb,\eps_{1})$ where $\BB \in \mctwoCl$ is an $m \times
  n$ matrix, $\ccb \in \mathbb{R}^m$, and $0 < \eps_{1} < 1$.}
\ENSURE{$(\BBtwostrict,\ccbtwostrict,\eps_{2})$ where $\BBtwostrict \in \mctwostrictCl$ is an $m' \times
  n'$ matrix, $\ccbtwostrict \in \mathbb{R}^{m'}$, and $0 < \eps_{2} < 1$.}
\STATE $\wstrict \leftarrow  \frac{\eps_1}{100\kappa(\BB)\sigma_{\max}(\BB) \norm{\ccb}_2}$
\label{lin:delta}
\STATE $\widehat{\mathcal{B}} \leftarrow \emptyset$  \hfill
\COMMENT{New equations for $\mctwostrictCl$}
\FOR {each pair of vertices $i$, $j$ whose blocks are involved in some equation in $\BB$}
\IF{$\BB$ does not contain a type $1$ equation between $i$ and $j$}
\STATE $\widehat{\mathcal{B}}\leftarrow \widehat{\mathcal{B}}
	\cup \{ \uu_i - \uu_j = 0\}$.
\ENDIF
\IF{$\BB$ does not contain a type $2$ equation betwen $i$ and $j$}
\STATE $\widehat{\mathcal{B}} \leftarrow \widehat{\mathcal{B}}
	\cup \{ \vv_i - \vv_j = 0\}$.
\ENDIF
\IF{$\BB$ does not contain a type $1 + 2$ equation betwen $i$ and $j$}
\STATE $\widehat{\mathcal{B}} \leftarrow \widehat{\mathcal{B}}
	\cup \{\uu_i + \vv_i - (\uu_j + \vv_j) = 0 \}$.
\ENDIF
\ENDFOR 
\STATE Let $\BBtil$ be the coefficient matrix of equations
in $\widehat{\mathcal{B}}$.
\RETURN
\[
\BBtwostrict =
\left( \begin{array}{c}
\BB \\
\delta \BBtil
\end{array}  \right),
\qquad
\ccbtwostrict =
\left( \begin{array}{c}
\ccb \\
{\bf 0}
\end{array}  \right),
\qquad
\epsilon_2 = \frac{\eps_1}{100}.
\]
\end{algorithmic}
\end{algorithm}

\begin{algorithm}[H]
\renewcommand{\algorithmicrequire}{\textbf{Input:}}
\renewcommand\algorithmicensure {\textbf{Output:}}
    \caption{\label{alg:MCTStoMCTsolnback}\algMCTStoMCTsoln}
    \begin{algorithmic}[1]
	\REQUIRE $m \times n$ matrix $\BB \in \mctwoCl$, 
    $m' \times n'$ matrix $\BBtwostrict \in \mctwostrictCl$, 
    vector $\ccb \in \mathbb{R}^m$
    vector $\xx \in \mathbb{R}^n$.
    \ENSURE Vector $\yy \in \mathbb{R}^{n}$.
    \IF {$\BB^{\trp} \ccb = {\bf 0}$}
    	\RETURN $\yy \assign {\bf 0}$
    \ELSE
	    \RETURN $\yy \assign \xx$
	\ENDIF
\end{algorithmic}
\end{algorithm}

Let $\BBtwostrict\xx = \ccbtwostrict$ be the linear system returned by a call to $\algMCTtoMCTS(\BB, \ccb, \eps_1)$. 
As we only add up to two edges per original edge in $\BB$,
and do not introduce any new variables,
the size of $\BBtwostrict$ is immediate from this routine.
According to the algorithms, $\ccbtwostrict$ is simply \[
\ccbtwostrict = \left( \begin{array}{c}
\ccb \\
 {\bf 0}
\end{array} \right) 
= \left( \begin{array}{c}
\cca \\
{\bf 0}
\end{array} \right).
\]
Note the two zero vectors in the above equation have different dimensions.

For a given matrix $\BB$, we will choose the additive term
to ensure non-zeros to be:
\begin{equation}
\label{eq:wEpsilon}
\wstrict
\defeq
\frac{\eps_1}{100\kappa(\BB) \sigma_{\max}(\BB) \norm{\ccb}_2}.
\end{equation}
Note that we can bound $\kappa(\BB), \sigma_{\max}(\BB)$ by condition number of $\AA$ (i.e, the linear system instance of $\genCl$).

\begin{lemma}
	\label{lem:MCTwoStrictSize}
	$\nnz(\BBtwostrict) = O\left( \nnz(\BB) \right)$.
By our setting of $\wstrict$ in~\eqref{eq:wEpsilon}, the largest entry of $\BBtwostrict$ does not change, the smallest entry of $\BBtwostrict$ is at least $\wstrict$.
\end{lemma}


Note that the addition of all three types of edges means that
the null space of $\BBtwostrict$ is now given by the connected
components in its graph theoretic structure, and is likely
significantly different from the null space of $\BB$.
In order to solve the Linear System Approximation Problem(\LSA)~problems
\begin{align*}
& \min_{\xx} \norm{\BB\xx - \ccb}_2,\\
& \min_{\xx} \norm{\BBtwostrict \xx - \ccbtwostrict}_2,
\end{align*}
we need to solve the two linear systems:
\begin{align}
& \BB^{\trp} \BB \xx = \BB^{\trp} \ccb, 
\label{eqn:McTwoToMcTwoStrict_ls_mctwo}
\\
& (\BBtwostrict)^{\trp} \BBtwostrict \xx = (\BBtwostrict)^{\trp} \ccbtwostrict.
\label{eqn:McTwoToMcTwoStrict_ls_mctwostrict}
\end{align}

First note that the RHS the two equations are the same because:
\[
(\BBtwostrict)^{\trp} \ccbtwostrict
= \left( \begin{array}{cc}
\BB^{\trp} &
\wstrict \BBtil^{\trp}
\end{array} \right)
\left( \begin{array}{c}
\ccb \\
{\bf 0}
\end{array} \right)
= \BB^{\trp} \ccb.
\]

This means that for a sufficiently small choice of $\wstrict$,
the differences between the solutions of these two linear systems is small.

\begin{lemma}
\label{lem:McTwoToMcTwoStrict_sol_diff}
Let $\BBtwostrict \xx =  \ccbtwostrict$ be the linear system returned by
a call to $\algMCTtoMCTS(\BB, \ccb, \eps_1)$.
Let $\xxbopt \in \argmin_{\xx} \norm{\BB\xx - \ccb}_2$ and $\xxbopt \perp \nulls(\BB)$. 
Let $\xxbtwostrictopt \in \argmin_{\xx} \norm{\BBtwostrict \xx - \ccbtwostrict }_2$.
%
Then we have:
\[
\norm{\xxbopt - \xxbtwostrictopt}_{\BB^{\trp} \BB}
\leq
O\left(   \wstrict (1+\wstrict) \kappa(\BB) \norm{\ccb}_2 \right).
\]
\end{lemma}

\begin{proof}
The desired distance can be written as
\[
\norm{\xxbopt - \xxbtwostrictopt}_{\BB^{\trp} \BB}
= \norm{\BB \left( \xxbopt - \xxbtwostrictopt \right)}_2.
\]
Also, the optimality of $\xxbopt$ means that
$\BB \xxbopt - \ccb$ is perpendicular to anything in the rank space of $\BB$,
in particular,
\[
\BB \left( \xxbopt - \xxbtwostrictopt \right)
\perp
\BB \xxbopt - \ccb,
\]
which in turn gives:
\begin{align}
\norm{\BB\left(\xxbopt -  \xxbtwostrictopt\right)}_2^2
= 
\norm{\BB \xxbtwostrictopt - \ccb}_2^2 - 
\norm{\BB\xxbopt - \ccb}_2^2.
\label{eqn:McTwoToMcTwoStrict_exact_lhs}
\end{align}

We now bound the right hand side.
Since $\xxbtwostrictopt$ is a minimizer of $\norm{\BBtwostrict \xx - \ccbtwostrict}_2$, we have
\[
\norm{\BBtwostrict \xxbtwostrictopt - \ccbtwostrict}_2^2 \le \norm{\BBtwostrict \xxbopt - \ccbtwostrict}_2^2,
\]
where we can extract out the $\wstrict \BBtil$ term
in $\BBtwostrict$ separately to get:
\[
\norm{\BB\xxbtwostrictopt - \ccb}_2^2 + \wstrict^2 \norm{\BBtil \xxbtwostrictopt}_2^2
\le \norm{\BB \xxbopt - \ccb}_2^2 + \wstrict^2 \norm{\BBtil \xxbopt}_2^2.
\]
Together with Equation~\eqref{eqn:McTwoToMcTwoStrict_exact_lhs}, this then gives
\[
\norm{\xxbopt - \xxbtwostrictopt}_{\BB^{\trp}\BB}
\leq \wstrict \norm{\BBtil \xxbopt}_2.
\]
It remains to upper bound $\norm{\BBtil \xxbtwostrictopt}_2$.
By the assumption of $\xxbopt \perp \nulls(\BB)$, we get:
\[
\frac{\norm{\BBtil \xxbopt}_2^2}{\norm{\BB\xxbopt}_2^2}
\le \frac{\lambda_{\max} \left( \BBtil^{\trp}\BBtil \right)}{\lambda_{\min}\left( \BB^{\trp} \BB \right) }.
\]

By a proof similar to Lemma~\ref{lem:GZToMC2LambdaMax},
which upper bounds $\lambda_{\max}\left( \BB^{\trp} \BB\right)$, we have:
\[
\sigma_{\max}\left( \BBtil \right)
= O (1+\wstrict) \sigma_{\max} (\BB),
\]
which implies
\[
\norm{\BBtil \xxbopt}_2 
\le O(1+\wstrict) \kappa(\BB) \norm{\BB\xxbopt}_2.
\]
Since $\xxbopt$ is a minimizer of $\min_{\xx} \norm{\BB\xx - \ccb}_2$, we have 
\[
\norm{\BB\xxbopt}_2 = \norm{\PPi_{\BB}\ccb}_2 \le \norm{\ccb}_2.
\]
Therefore, 
\[
\norm{\xxbopt - \xxbtwostrictopt}_{\BB^{\trp}\BB} \le O\left(
 \wstrict (1+\wstrict) \kappa(\BB)  \right)  \norm{\cc}_2,
\]
which completes the proof.
\end{proof}


We now check that the approximate solutions of the two linear systems~\eqref{eqn:McTwoToMcTwoStrict_ls_mctwo} and~\eqref{eqn:McTwoToMcTwoStrict_ls_mctwostrict} are also close to each other.

\begin{lemma}
\label{lem:GZToMC2NullSpaceClose}
Let $\BBtwostrict \xx =  \ccbtwostrict$ be the linear system returned by
a call to $\algMCTtoMCTS(\BB, \ccb, \eps_1)$, and let $\eps_2$ be the error parameter returned by this call.
Let $\xx $ be a vector such that 
\[
\norm{\BBtwostrict\xx - \PPi_{\BBtwostrict}\ccbtwostrict }_2
\le \eps_2 \norm{\PPi_{\BBtwostrict} \ccbtwostrict}_2.
\]
Let $\yy$ be the vector returned by a call to $\algMCTStoMCTsoln(\BB, \BBtwostrict, \ccb, \xx)$.
Then
\[
\norm{\BB\yy - \PPi_{\BB}\ccb}_2
\le \eps_1 \norm{\PPi_{\BB} \ccb}_2.
\]
\end{lemma}

\begin{proof}
In Algorithm~\ref{alg:MCTStoMCTsolnback}, if the condition $\BB^{\trp}\ccb = {\bf 0}$ is true, then $\yy = {\bf 0}$. This implies
\[
\norm{\BB\yy - \PPi_{\BB}\ccb}_2
= {\bf 0}.
\]

In the following, we assume that $\BB^{\trp}\ccb \neq {\bf 0}$, in which case $\yy = \xx$.
We first show that the construction implies
$\norm{\BB\xx - \PPi_{\BB}\ccb }_2
\leq \norm{\BBtwostrict\xx - \PPi_{\BBtwostrict}\ccbtwostrict }_2$.
Once again, let $\xxbopt$ and $\xxbtwostrictopt$
be the vectors that minimize
$\norm{\BB \xx - \ccb}_2$ and $\norm{\BBtwostrict \xx - \ccbtwostrict}_2$ respectively.
This choice gives:
\[
\norm{\BB\xx - \PPi_{\BB}\ccb }_2
= \norm{\xx - \xxbopt}_{\BB^{\trp}\BB}.
\]
As $\norm{\cdot}_{\BB^{\trp} \BB}$ is a norm, by triangle inequality we have:
\[
\norm{\xx - \xxbopt}_{\BB^{\trp}\BB}
\leq \norm{\xx - \xxbtwostrictopt}_{\BB^{\trp}\BB}
+ \norm{\xxbtwostrictopt - \xxbopt}_{\BB^{\trp}\BB}.
\]
We will bound these two terms separately.

Since $\BB^{\trp}\BB \pleq \BB^{\trp}\BB + \wstrict^2 \BBtil^{\trp} \BBtil =  (\BBtwostrict)^{\trp} \BBtwostrict$,
the first term is less than its norm in the $(\BBtwostrict)^{\trp} \BBtwostrict$ norm:
\[
\norm{\xx - \xxbtwostrictopt}_{\BB^{\trp}\BB}
\leq
\norm{\xx - \xxbtwostrictopt}_{\left(\BB^{\trp}\BB + \wstrict^2 \BBtil^{\trp} \BBtil \right)}
= \norm{\BBtwostrict\xx - \PPi_{\BBtwostrict}\ccbtwostrict}_2,
\]
while the second term is precisely the distances between the two
optimums, which we just bounded in 
Lemma~\ref{lem:McTwoToMcTwoStrict_sol_diff}.
Combining these bounds then gives:
\begin{align}
\norm{\BB\xx - \PPi_{\BB}\ccb}_2
\le
\norm{\BBtwostrict\xx - \PPi_{\BBtwostrict}\ccbtwostrict}_2 + 
O\left(  \wstrict (1+\wstrict) \kappa(\BB) \norm{\ccb}_2 \right).
\label{eqn:McTwoToMcTwoStrict_approx_errb}
\end{align}

As the equations in $\BBtwostrict$ is a superset of the ones in $\BB$,
we have
\[
\norm{\BBtwostrict\xxbtwostrictopt - \ccbtwostrict}_2 \ge
\norm{\BB\xxbtwostrictopt - \ccb}_2
\ge
\norm{\BB\xxbopt - \ccb}_2.
\]
Substituting
$\norm{\BB \xxbopt - \ccb}_2^2
= \norm{\cca}_2^2 - \norm{\PPi_{\BB}\ccb}_2^2$
and its equivalent in $\BBtwostrict$ gives
\[
\norm{\PPi_{\BBtwostrict}\ccbtwostrict}_2 \le \norm{\PPi_{\BB}\ccb}_2.
\]
Together with the condition of the lemma,
\[
\norm{\BBtwostrict\xx - \PPi_{\BBtwostrict}\ccbtwostrict}_2
\le \epsilon_2 \norm{\PPi_{\BB}\ccb}_2.
\]
Plugging this into Equation~\eqref{eqn:McTwoToMcTwoStrict_approx_errb},
we have
\[
\norm{\BB\xx - \PPi_{\BB}\ccb}_2
= O \left( \eps_2 \norm{\PPi_{\BB}\ccb}_2
+ \wstrict \kappa\left( \BB \right) \norm{\ccb}_2 \right).
\]

It remains to upper bound $\norm{\ccb}_2$ by a function of  $\norm{\PPi_{\BB} \ccb}_2$. 
By our construction in Algorithm~\ref{alg:GZ2toMC2} and assumption, $\norm{\BB\ccb}_2^2$ is a positive integer.
By Lemma~\ref{lem:ProjLengthLower},
\[
\norm{\PPi_{\BB}\ccb}_2 \ge \frac{1}{\sigma_{\max}(\BB)}.
\]
Thus,
\[
\norm{\BB\xx - \PPi_{\BB}\ccb}_2
= O \left( \epsilon_2 + \wstrict \kappa(\BB) \sigma_{\max}(\BB) \norm{\ccb}_2 \right) \norm{\PPi_{\BB} \ccb}_2.
\]
By our setting of $\wstrict$ in Equation~\eqref{eq:wEpsilon}, we have
\[
\norm{\BB\xx - \PPi_{\BB}\ccb}_2
\le \eps_1 \norm{\PPi_{\BB}\ccb}_2.
\]
This completes the proof.
%
\end{proof}

\subsection{Bounding Condition Number of the New
	System in $\mctwostrictCl$}
\label{subsec:CondBoundMCTwoStrict}

We now establish bounds on the numerical quantities related
to $\BBtwostrict$.
By a proof similar to the upper bound on 
$\lambda_{\max}\left( \BB^{\trp}\BB \right)$
in Lemma~\ref{lem:GZToMC2LambdaMax}, we have:
\[
\lambda_{\max}\left((\BBtwostrict)^{\trp} \BBtwostrict \right) = 
O \left( \lambda_{\max} \left( \BB^{\trp} \BB \right) \right).
\]
As a result, we focus on the lower bound here:
\begin{lemma}
\label{lem:MC2ToMC2StrictLambdaMin}
The matrix $\BBtwostrict$ from the linear system returned by
a call to $\algMCTtoMCTS(\BB, \ccb, \epsilon_1, \wstrict)$,
where $\wstrict$ is set according to Equation~\ref{eq:wEpsilon},
satisfies
\[
\sigma_{\min}\left( \BBtwostrict\right) = \Omega \left( \frac{\epsilon_1}{\kappa(\BB) \sigma_{\max}(\BB) \nb} \right).
\]
\end{lemma}

\begin{proof}
Let $G$ be a unit-edge weight graph whose vertex set and edge set are same as the underlying graph of $\BBtwostrict$.
Let $\LL_G$ be the associated Laplacian matrix of $G$, and let $\MM := \LL_G \otimes \CC$, where
\begin{align*}
\CC = \left( \begin{array}{cc}
2 & -1 \\
-1 & 2
\end{array} \right)
\end{align*}
is a symmetric PSD matrix.
Note 
\begin{align*}
\MM = \LL_G \otimes \left( \begin{array}{cc}
1 & 0 \\
0 & 0
\end{array} \right)
+ \LL_G \otimes \left( \begin{array}{cc}
0 & 0 \\
0 & 1
\end{array} \right)
+ \LL_G \otimes \left( \begin{array}{cc}
1 & -1 \\
-1 & 1
\end{array} \right),
\end{align*}
which means $\MM$ has the same null space as $(\BBtwostrict)^{\trp}\BBtwostrict$.

\[
\lambda_{\min}\left( (\BBtwostrict)^{\trp} \BBtwostrict \right)
= \min_{\xx \perp \nulls(\BBtwostrict)} \frac{\xx^{\trp} (\BBtwostrict)^{\trp} \BBtwostrict \xx}{\xx^{\trp} \xx} 
\ge \delta^2 \min_{\xx \perp \nulls(\MM)} \frac{\xx^{\trp} \MM \xx}{\xx^{\trp} \xx},
\]
the last inequality is due to that $\delta$ is the minimum edge weight.

Note that the sum of the 3 types of blocks is positive
definite and has eigenvalue at least $1$: the type $1$
and $2$ blocks already sum to $\II$.
Formally:
\[
\lambda_{\min} (\MM)
= \lambda_{\min} (\LL_G) \cdot  \lambda_{\min} (\CC)
= \lambda_{\min} (\LL_G).
\]

The result then follows from the folklore bound that
the minimum non-zero eigenvalue of a unit weighted graph
is at least $\frac{1}{n^2}$.
One way to see this is via Cheeger's inequality
(see e.g.~\cite{Spielman07:survey}
applied to each block: this decomposition is equivalent
to spearating the matrix into its diagonal blocks
based on the connected components, and then invoking
the fact that the minimum weight of a cut is at least
$1$, and there are at most $n$ vertices.
Together these tools imply
$\lambda_{\min} (\LL_G) = \Omega \left( \frac{1}{\nb^2} \right)$,
and in turn:
\[
\sigma_{\min} \left( \BBtwostrict \right)
\geq \Omega\left( \frac{\delta}{n} \right)
= \Omega \left( \frac{\epsilon_1}{\kappa(\BB) \sigma_{\max}(\BB) \nb} \right).
\]
%
\end{proof}

This also implies a bound on the condition number of $\BBtwostrict$:
\begin{lemma}
\label{lem:MC2ToMC2StrictCond}
$\kappa(\BBtwostrict) = O \left( \eps_1^{-1} \sigma_{\max}^2(\BB) \kappa(\BB) \nb  \right)$.
\end{lemma}

\subsection{Putting it All Together}
\begin{proof}[Proof of Lemma~\ref{lem:McTwoToMcTwoStrict}]
By Lemma~\ref{lem:MCTwoStrictSize}, we have
\[
\nnz(\BBtwostrict) = O(s),
\]
The largest entry is $U$, and the smallest nonzero entry is $\wstrict$. By Equation~\eqref{eq:wEpsilon} and Lemma~\ref{clm:para_l2_norm} and~\ref{clm:para_eig},
\[
\wstrict = \Omega \left( \frac{\eps}{KU^2} \right).
\]
By Lemma~\ref{lem:MC2ToMC2StrictCond} and Lemma~\ref{clm:para_l2_norm} and~\ref{clm:para_eig},
the condition number is
\[
\kappa(\BBtwostrict) = O \left( \frac{s^2 U^2 K}{\epsilon} \right).
\]
By Lemma~\ref{lem:GZToMC2NullSpaceClose},
the accuracy error for $\BBtwostrict$ is $O(\epsilon)$.
\end{proof}

\newcommand{\algMCTStoMCTSZ}{\ensuremath{\textsc{Reduce}\mctwostrictCl\textsc{To}\mctwostrictintCl}}
\newcommand{\algMCTSZtoMCTSsoln}{\ensuremath{\textsc{MapSoln}\mctwostrictintCl\textsc{To}\mctwostrictCl}}

\newcommand{\BBint}{\BB^{int}}
\newcommand{\ccbtwostrictint}{\cc^{\BBtwostrictint}}

\section{Rounding and Scaling Weights to Integers}

In this section, we show a reduction from the linear system with strict 2-commodity matrix to the linear system with integral strict 2-commodity matrix.

Algorithm~\ref{alg:MCTStoMCTSZ} presents the pseudo-code for the algorithm $\algMCTStoMCTSZ$. Given an instance $(\BBtwostrict, \ccbtwostrict, \eps_1)$ where $\BBtwostrict \in \mctwostrictCl$, the call $\algMCTStoMCTSZ(\BBtwostrict, \ccbtwostrict, \eps_1)$ returns an instance $(\BBtwostrictint, \ccbtwostrictint, \eps_2)$ where $\BBtwostrictint \in \mctwostrictintCl$.
Algorithm~\ref{alg:MCTSZtoMCTSsolnback} provides the (trivial) pseudo-code for $\algMCTSZtoMCTSsoln$ which maps a solution of an instance over $\mctwostrictintCl$ to a solution of an instance over $\mctwostrictCl$.

\begin{algorithm}[ht]
\renewcommand{\algorithmicrequire}{\textbf{Input:}}
\renewcommand\algorithmicensure {\textbf{Output:}}
\caption{\label{alg:MCTStoMCTSZ}\algMCTStoMCTSZ}
    \begin{algorithmic}[1]
\REQUIRE{$(\BBtwostrict,\ccbtwostrict,\eps_{1})$ where $\BBtwostrict \in \mctwostrictCl$ is an $m \times
  n$ matrix, $\ccbtwostrict \in \mathbb{R}^m$, and $0 < \eps_{1} < 1$}
\ENSURE{$(\BBtwostrictint,\ccbtwostrictint,\eps_{2})$ where $\BBtwostrictint \in \mctwostrictintCl$ is an $m' \times
  n'$ matrix, $\ccbtwostrictint \in \mathbb{R}^{m'}$, and $0 < \eps_{2} < 1$.}
\STATE $k \assign \left\lceil \log_2
 \frac{\nnz(\BBtwostrict)}{\eps_1} \right\rceil + 2$
\STATE $\BBtwostrictint \assign \BBtwostrict$
\STATE $\ccbtwostrictint \assign \ccbtwostrict$
\FOR {each entry $(\BBtwostrictint)_{ij}$}
	\STATE $(\BBtwostrictint)_{ij} \assign  \left\lceil (\BBtwostrictint)_{ij} \cdot 2^k \right\rceil$
	\label{lin:BBtwostrictint_entry}
\ENDFOR
\FOR {each entry $\ccbtwostrictint_i$} 
	\STATE $\ccbtwostrictint_i \assign \ccbtwostrictint_i \cdot 2^k$
\ENDFOR	
\RETURN $(\BBtwostrictint, \ccbtwostrictint, \eps_1/3)$.
\end{algorithmic}
\end{algorithm}

\begin{algorithm}[htb]
\renewcommand{\algorithmicrequire}{\textbf{Input:}}
\renewcommand\algorithmicensure {\textbf{Output:}}
    \caption{\label{alg:MCTSZtoMCTSsolnback}\algMCTSZtoMCTSsoln}
    \begin{algorithmic}[1]
	\REQUIRE $m \times n$ matrix $\BBtwostrict \in \mctwoCl$, 
    $m' \times n'$ matrix $\BBtwostrictint \in \mctwostrictCl$, 
    vector $\xx \in \mathbb{R}^{n'}$.
    \ENSURE Vector $\yy \in \mathbb{R}^{n}$.
    \RETURN $\xx$
\end{algorithmic}
\end{algorithm}

For simplicity, we analyze an intermediate linear system $\BBint\yy = \ccbtwostrict$, where 
\[
(\BBint)_{ij} \defeq 2^{-k} \cdot (\BBtwostrictint)_{ij}.
\]
Recall the definition of $\BBtwostrictint$ in Algorithm~\ref{alg:MCTStoMCTSZ} line~\ref{lin:BBtwostrictint_entry}, we have
\begin{align}
(\BBint)_{ij}
= 2^{-k} \cdot \left\lceil (\BBtwostrict)_{ij} \cdot 2^k \right\rceil.
\label{eqn:rounding_rule}
\end{align}
The linear system $\BBint\yy = \ccbtwostrict$ is exactly the linear system $\BBtwostrictint \yy = \ccbtwostrictint$ multiplying a factor $2^{-k}$ on both sides.


Note the condition number and the eigen-space of a matrix, the optimal solutions of the projection problems are invariant under scaling.
Let $\yy^* \in \argmin_{\yy} \norm{\BBtwostrictint \yy - \ccbtwostrictint}_2 = \argmin_{\yy} \norm{\BBint \yy - \ccbtwostrict}_2$.
The following two inequalities: 
\[
\norm{\yy - \yy^*}_{(\BBtwostrictint)^{\trp} \BBtwostrictint}
\le \epsilon \norm{\PPi_{\BBtwostrictint} \ccbtwostrictint}_2
\]
and
\[
\norm{\yy - \yy^*}_{(\BBint)^{\trp} \BBint}
\le \epsilon \norm{\PPi_{\BBint} \ccbtwostrict}_2
\]
are equivalent.
Thus, it suffices to analyze the linear system $\BBint \yy = \ccbtwostrict$.

Solving the two projection problems
\begin{align*}
& \min_{\xx} \norm{\BBtwostrict\xx - \ccbtwostrict}_2, \\
& \min_{\yy} \norm{\BBint\yy - \ccbtwostrict}_2,
\end{align*}
is equivalent to solving the following two linear systems:
\begin{align}
& (\BBtwostrict)^{\trp} \BBtwostrict \xx = (\BBtwostrict)^{\trp} \ccbtwostrict,  
\label{eqn:McTwoStrictToMcTwoInt_ls1} \\
& (\BBint)^{\trp} \BBint \yy = (\BBint)^{\trp} \ccbtwostrict. \nonumber
\end{align}
Note $\ccbtwostrict = (\cca; {\bf 0})$, and all entries of the rows of $\BBtwostrict$ corresponding to the original rows of $\AA$ are integers.
Thus, 
\[
(\BBint)^{\trp} \ccbtwostrict = (\BBtwostrict)^{\trp} \ccbtwostrict.
\]
That is, the 2nd linear system is equivalent to
\begin{align}
(\BBint)^{\trp} \BBint \yy = (\BBtwostrict)^{\trp} \ccbtwostrict.
\label{eqn:McTwoStrictToMcTwoInt_ls2}
\end{align}

Let 
\[
\MM \defeq (\BBtwostrict)^{\trp} \BBtwostrict
\mbox{ and }
\MMhat \defeq (\BBint)^{\trp} \BBint.
\]

We bound the eigenvalues of $\MMhat$ by the following lemma.

\begin{lemma}
$\lambda_{\max}(\MMhat) \le (1+2^{-k+2}) \lambda_{\max}(\MM)$ and $\lambda_{\min}(\MMhat) \ge \lambda_{\min}(\MM)$.
Furthermore, $\kappa(\MMhat) \le (1+2^{-k+2}) \kappa(\MM)$.
\label{lem:McTwoStrictToMcTwoInt_cond}
\end{lemma}

\begin{proof}

Note $\MM$ can be written as $\BB^{\trp} \WW \BB$, where $\BB$ is the incidence-structured block matrix with \emph{unit} nonzero entries and $\WW$ is the diagonal matrix with edge weights.
Let $\xx$ be an arbitrary vector, and let $\yy \defeq \BB \xx$.
\[
\xx^{\trp} \MM \xx
= \yy^{\trp} \WW \yy = \sum_i \WW_{ii} \yy_i^2,
\]
and 
\[
\xx^{\trp} \MMhat \xx
= \sum_i \WWhat_{i} \yy_i^2.
\]
Let 
\[
\alpha \defeq \min_i \frac{\WWhat_{ii}}{\WW_{ii}}
\mbox{ and }
\beta \defeq \max_i \frac{\WWhat_{ii}}{\WW_{ii}}.
\]
We have
\[
\alpha \sum_i \WW_{i} \yy_i^2 \le \sum_i \WWhat_{i} \yy_i^2 \le \beta \sum_i \WW_{i} \yy_i^2.
\]
This implies that
\begin{align}
\alpha \MM \pleq \MMhat \pleq \beta \MM.
\label{eqn:McTwoStrictToMcTwoInt_spetral_inequality}
\end{align}

Now we bound the values of $\alpha, \beta$.
According to Equation~\eqref{eqn:rounding_rule}, 
\begin{align}
\alpha = \min_i \frac{\WWhat_{ii}}{\WW_{ii}}
\ge \frac{ \left( 2^{-k} \cdot \WW_{ii}^{1/2} \cdot 2^k \right)^2}{ \WW_{ii}}
= 1.
\label{eqn:McTwoStrictToMcTwoInt_alpha_val}
\end{align}
Similarly,
\begin{align}
\beta = \max_i \frac{\WWhat_{ii}}{\WW_{ii}}
\le \frac{\left( 2^{-k} \left( \WW_{ii}^{1/2} \cdot 2^k + 1 \right) \right)^2}{ \WW_{ii}}
\le 1 + 2^{-k+2}.
\label{eqn:McTwoStrictToMcTwoInt_beta_val}
\end{align}
The last inequality is due to $\WW_{ii} \ge 1$.
Thus,
\[
\lambda_{\max}(\MMhat) \le (1+2^{-k+2}) \lambda_{\max} (\MM)
\mbox{ and }
\lambda_{\min}(\MMhat) \ge \lambda_{\min} (\MM).
\]
The bound on the condition number 
$\kappa(\MMhat) \le (1+2^{-k+2}) \kappa(\MM)$ immediately follows the above two inequalities.
\end{proof}

We show the exact solutions of the two linear systems in Equation~\eqref{eqn:McTwoStrictToMcTwoInt_ls1} and~\eqref{eqn:McTwoStrictToMcTwoInt_ls2} are close, by the following lemma.

\begin{lemma}
Let 
\[
\xx^* \defeq \MM^{\dagger} (\BBtwostrict)^{\trp} \ccbtwostrict
\mbox{ and }
\yy^* \defeq \MMhat^{\dagger} (\BBtwostrict)^{\trp} \ccbtwostrict.
\]
$\xx^*$ and $\yy^*$ are exact solutions of linear system~\eqref{eqn:McTwoStrictToMcTwoInt_ls1} and~\eqref{eqn:McTwoStrictToMcTwoInt_ls2} respectively.
\[
\norm{\yy^* - \xx^*}_{\MM}
\le 2^{-k+2} \norm{(\BBtwostrict)^{\trp} \ccbtwostrict}_{\MM^{\dagger}}.
\]
\label{lem:McTwoStrictToMcTwoInt_exact}
\end{lemma}

\begin{proof}
Note $\MM, \MMhat$ are both symmetric.
Expanding the left hand side,
\begin{align*}
\norm{\yy^* - \xx^*}_2^2
&= (\ccbtwostrict)^{\trp} \BBtwostrict \left( \MMhat^{\dagger} - \MM^{\dagger} \right) \MM \left( \MMhat^{\dagger} - \MM^{\dagger} \right) (\BBtwostrict)^{\trp} \ccbtwostrict \\
&=  (\ccbtwostrict)^{\trp} \BBtwostrict \MM^{\dagger 1/2} \left( \MM^{ 1/2} \MMhat^{\dagger} \MM^{1/2} - \II \right)^2  \MM^{\dagger 1/2} (\BBtwostrict)^{\trp} \ccbtwostrict \\
&\le  \norm{\MM^{ 1/2} \MMhat^{\dagger} \MM^{1/2} - \II}_2^2
\norm{\MM^{\dagger 1/2} (\BBtwostrict)^{\trp} \ccbtwostrict}_2^2.
\end{align*}
The last inequality is by the Courant-Fischer theorem.

By Equation~\eqref{eqn:McTwoStrictToMcTwoInt_spetral_inequality} and the fact that $\MM, \MMhat$ have same null space,
\[
\left( \beta^{-1} - 1 \right) \II \pleq \MM^{1/2} \MMhat^{\dag} \MM^{1/2} - \II \pleq \left( \alpha^{-1} - 1 \right) \II.
\]
By Equation~\eqref{eqn:McTwoStrictToMcTwoInt_alpha_val} and~\eqref{eqn:McTwoStrictToMcTwoInt_beta_val}, 
\[
\norm{\MM^{ 1/2} \MMhat^{\dagger} \MM^{1/2} - \II}_2 \le 2^{-k+2}.
\]
Therefore, 
\[
\norm{\yy^* - \xx^*}_{\MM}
\le 2^{-k+2} \norm{(\BBtwostrict)^{\trp} \ccbtwostrict}_{\MM^{\dagger}}.
\]
%
\end{proof}

We then show the approximate solutions of the two linear systems in Equation \eqref{eqn:McTwoStrictToMcTwoInt_ls1} and~\eqref{eqn:McTwoStrictToMcTwoInt_ls2} are close.

\begin{lemma}
\label{lem:McTwoStrictToMcTwoInt_apx}
Let $\xx$ be a vector such that
\[
\norm{\BBint \xx - \PPi_{\BBint}\ccbtwostrict}_2 \le \epsilon_2 \norm{\PPi_{\BBint} \ccbtwostrict}_2.
\]
Then,
\[
\norm{\BBtwostrict \xx - \PPi_{\BBtwostrict}\ccbtwostrict}_2 \le \left( \epsilon_2  + 2^{-k+2} (1+\epsilon_2) \right) \norm{\PPi_{\BBtwostrict} \ccbtwostrict}_2.
\]
\end{lemma}

\begin{proof}
Let 
\[
\xx^* \defeq \MM^{\dagger} (\BBtwostrict)^{\trp} \ccbtwostrict
\mbox{ and }
\yy^* \defeq \MMhat^{\dagger} (\BBtwostrict)^{\trp} \ccbtwostrict.
\]
Expanding the left hand side,
\begin{align*}
\norm{\BBtwostrict \xx - \PPi_{\BBtwostrict}\ccbtwostrict}_2
& = \norm{\xx - \xx^*}_{\MM} \\
& \le \norm{\xx - \yy^*}_{\MM}
+ \norm{\yy^* - \xx^*}_{\MM}.
\end{align*}
The last inequality is due to the triangle inequality.
By Equation~\eqref{eqn:McTwoStrictToMcTwoInt_spetral_inequality}, the first term can be upper bounded by
\[
\norm{\xx - \yy^*}_{\MM} \le  \alpha^{-1/2} \norm{\xx - \yy^*}_{\MMhat}.
\]
The second term is upper bounded by Lemma~\ref{lem:McTwoStrictToMcTwoInt_exact}.
Thus, we have
\[
\norm{\BBtwostrict \xx - \PPi_{\BBtwostrict}\ccbtwostrict}_2
\le \alpha^{-1/2} \epsilon_2 \norm{\PPi_{\BBint} \ccbtwostrict}_2 + 2^{-k+2} \norm{\PPi_{\BBtwostrict} \ccbtwostrict}_2.
\]
Note that 
\[
\norm{\PPi_{\BBint} \ccbtwostrict}_2
= \norm{\BBtwostrict \ccbtwostrict}_{\MMhat^{\dag}}
\mbox{ and }
\norm{\PPi_{\BBtwostrict} \ccbtwostrict}_2
= \norm{\BBtwostrict \ccbtwostrict}_{\MM^{\dag}}.
\]
Again, by Equation~\eqref{eqn:McTwoStrictToMcTwoInt_spetral_inequality}, 
\[
\norm{\BBtwostrict \ccbtwostrict}_{\MMhat^{\dag}}
\le \alpha^{-1/2} \norm{\BBtwostrict \ccbtwostrict}_{\MM^{\dag}}.
\]
By Equation~\eqref{eqn:McTwoStrictToMcTwoInt_alpha_val} and~\eqref{eqn:McTwoStrictToMcTwoInt_beta_val}, 
\[
\norm{\BBtwostrict \xx - \PPi_{\BBtwostrict}\ccbtwostrict}_2
\le \left( \epsilon_2  + 2^{-k+2} (1+\epsilon_2) \right) 
\norm{\PPi_{\BBtwostrict} \ccbtwostrict}_2.
\]
This completes the proof.
\end{proof}

\begin{proof}[Proof of Lemma~\ref{lem:McTwoStrictToMcTwoStrictInt}]
We set 
\[
k \defeq \left\lceil \log_2 \frac{s}{\epsilon} \right\rceil + 2.
\]
After rounding and scaling, the number of nonzero entries does not change.
The value of the smallest nonzero entry does not decrease.
The value of the largest entry is multiplied by $2^k = s \epsilon^{-1}$, which is upper bounded by $s\eps^{-1} U$.

By Lemma~\ref{lem:McTwoStrictToMcTwoInt_cond}, 
\[
\kappa(\BBtwostrictint) \le (1 + s^{-1/2}\epsilon^{1/2}) \kappa(\BBtwostrict).
\]

By Lemma~\ref{lem:McTwoStrictToMcTwoInt_apx}, the accuracy is bounded by $\epsilon /3$.
\end{proof}
\newcommand{\algGtoGZ}{\textsc{Reduce\,}$\genCl$\textsc{to}$\genZCl$}
\newcommand{\algGZtoGZT}{\textsc{Reduce\,}$\genZCl$\textsc{to}$\genZtwoCl$}
\newcommand{\algGZtoGsoln}{\textsc{MapSoln\,}$\genZCl$\textsc{to}$\genCl$}
\newcommand{\algGZTtoGZsoln}{\textsc{MapSoln\,}$\genZtwoCl$\textsc{to}$\genZCl$}

\newcommand{\xxztwoopt}{\xx^{\text{Z},2*}}
\newcommand{\xxzopt}{\xx^{\text{Z}*}}
\newcommand{\xxopt}{\xx^*}

\section{$\genCl$ Efficiently Reducible to $\genZtwoCl$}
\label{sec:genToZeroSumTwo}
In this section, we prove Lemma~\ref{lem:genToZeroSumTwo}.
\begin{definition}
We let $\genCl_{\text{z}}$ denote the class of all matrices with integer
    valued entries s.t. there is at least one non-zero entry in every
    row and column, and every row has zero row sum.
\end{definition}
%
%
%
%

\subsection{$\genCl \leq_{f} \genZCl$}

In this section, we show the reduction from $\genCl$ to $\genZCl$.
Given an instance of linear system $(\AA, \cca, \eps_1)$ with $\AA \in \genCl$, the goal is to construct an instance of linear system $(\AAZ, \ccz, \eps_2)$ with $\AAZ \in \genCl_{\text{z}}$ such that, there exists a map between the solutions of the two linear systems.

Algorithm~\ref{alg:GtoGZ} shows the construction of $(\AAZ, \ccz, \eps_2)$, and Algorithm~\ref{alg:GtoGZsolnback} shows the transform from a solution of $(\AAZ, \ccz, \eps_2)$ to a solution of $(\AA, \cca, \epsilon_1)$.
Recall that $\AA$ has $\na$  columns. According to the algorithm, $\AAZ$ has $(\na+1)$ columns, and $\ccz = \cc$.

\begin{algorithm}[htb]
\renewcommand{\algorithmicrequire}{\textbf{Input:}}
\renewcommand\algorithmicensure {\textbf{Output:}}
    \caption{\label{alg:GtoGZ}\algGtoGZ}
    \begin{algorithmic}[1]
\REQUIRE{$(\AA,\cca,\eps_{1})$ where $\AA_{1} \in \genCl$ is an $m \times
  n$ matrix, $\cca \in \mathbb{R}^m$, and $0 < \eps_{1} < 1$.}
\ENSURE{$(\AAZ,\ccz,\eps_{2})$ where $\AAZ \in \genZCl$ is an $m' \times
  n'$ matrix, $\ccz \in \mathbb{R}^{m'}$, and $0 < \eps_{2} < 1$.}
\RETURN $\left( \left( \begin{array}{cc}
\AA & -\AA \one
\end{array} \right), \cca, O \left( \frac{\eps_1}{n} \right) \right)$.
\end{algorithmic}
\end{algorithm}

\begin{algorithm}[htb]
\renewcommand{\algorithmicrequire}{\textbf{Input:}}
\renewcommand\algorithmicensure {\textbf{Output:}}
    \caption{\label{alg:GtoGZsolnback}\algGZtoGsoln}
    \begin{algorithmic}[1]
	\REQUIRE $m \times n$ matrix $\AA \in \genCl$, 
    $m' \times n'$ matrix $\AAZ \in \genZCl$,
    vector $\xxz \in \mathbb{R}^{n'}$.
    \ENSURE Vector $\xx \in \mathbb{R}^{n}$.
    \RETURN $\xx \leftarrow \xxz_{1:n} - \xxz_{n+1}\one$
\end{algorithmic}
\end{algorithm}


\begin{lemma}
Let $\AAZ\xxz = \ccz$ be the linear system returned by a call to 
\algGtoGZ$(\AA, \cca, \epsilon_1)$.
Then,
\[
\nnz(\AAZ) = O \left( \nnz(\AA) \right),
\]
and the largest entry of $\AAZ$ is at most $\norm{\AA}_{\infty}$.
\label{lem:gnToZeroSumTwo_dim}
\end{lemma}

We show the relation between the exact solvers of the two linear systems.

\begin{claim}
Let $\xxzopt \in \argmin_{\xx} \norm{\AAZ\xx - \ccz}_2$.
Let $\xxopt$ be the output vector of Algorithm~\ref{alg:GtoGZsolnback}, that is,
\[
\xxopt = \xxzopt_{1:\na} - \xxzopt_{\na+1} \one.
\]
Then,
\[
\xxopt \in \argmin_{\xx} \norm{\AA\xx - \cca}_2.
\]
\label{clm:GtoGZ_exact}
\end{claim}
\begin{proof}
Note $\one \in Null(\AAZ)$, thus,
\begin{align*}
\AAZ\xxzopt = \AAZ \left( \xxzopt - \xxzopt_{\na+1} \one \right).
\end{align*}
Expanding it gives
\[
\AAZ\xxzopt =
\left( \begin{array}{cc}
\AA & -\AA\one
\end{array} \right) \left( \begin{array}{c}
\xx^* \\
{\bf 0}
\end{array} \right)
= \AA\xx^*.
\]
Note $\ccz = \cc$.
Therefore, $\xxopt \in \argmin_{\xx} \norm{\AA\xx - \cca}_2$.
\end{proof}

We then show the relation between the approximate solvers.

\begin{lemma}
Let $\xxz$ be a vector such that $\norm{\AAZ\xxz - \PPi_{\AAZ}\cca}_2 \le \epsilon_2 \norm{\PPi_{\AAZ} \cca}_2$. 
Let $\xx$ be the output vector of Algorithm~\ref{alg:GtoGZsolnback}, that is, $\xx = \xxz_{1:\na} - \xxz_{\na+1} \one$.
Then, 
\[
\norm{\AA\xx - \PPi_{\AA}\cca}_2 \le O \left(  \frac{\sqrt{\na}\sigma_{\min}(\AA)}{\sigma_{\min}(\AAZ)} \right) \epsilon_2 \norm{\PPi_{\AA} \cca}_2.
\]
\label{lem:genToZeroSum_eps}
\end{lemma}

\begin{proof}
Let $\xxzopt \in \argmin_{\xx} \norm{\AAZ \xx - \ccz}$ whose last entry is 0.
Expanding the left hand side norm, 
\[
\norm{\AAZ\xxz - \PPi_{\AAZ}\cca}_2
= \norm{\AAZ\left( \xxz - \xxzopt \right)}_2.
\]
By Claim~\ref{clm:GtoGZ_exact}, $\xxopt \in \argmin_{\xx} \norm{\AA\xx - \cca}_2$.
Thus,
\[
\norm{\AA\left( \xx - \xx^*\right)}_2 = \norm{\AAZ(\xxz - \xxzopt)}_2.
\]
Note that
\[
\AA\xx = \AAZ \xxz
\mbox{ and }
\AA\xxopt = \AAZ \xxzopt.
\]
Together with the Lemma condition, this gives 
\begin{align}
\norm{\AA\xx - \PPi_{\AA}\cca}_2 
= \norm{\AAZ\xxz - \PPi_{\AAZ}\ccz}_2 
\le \epsilon_2 \norm{\PPi_{\AAZ}\cca}_2.
\label{eqn:GtoGZ_lhs}
\end{align}

It remains to upper bound $\norm{\PPi_{\AAZ}\cca}_2$ by a function of $\norm{\PPi_{\AA}\cca}_2$.
Note
\begin{align}
\norm{\PPi_{\AA}\cca}_2 = \norm{\AA^{\trp} \cca}_{\left( \AA^{\trp} \AA \right)^{\dagger}}
\mbox{ and }
\norm{\PPi_{\AAZ}\cca}_2 = \norm{(\AAZ)^{\trp} \cca}_{\left( (\AAZ)^{\trp} \AAZ \right)^{\dagger}}.
\label{eqn:GtoGZ_rhs}
\end{align}
It suffices to work on $\norm{\AA^{\trp} \cca}_{\left( \AA^{\trp} \AA \right)^{\dagger}}$ and $\norm{(\AAZ)^{\trp} \cca}_{\left( (\AAZ)^{\trp} \AAZ \right)^{\dagger}}$.
Since $\AAZ = \left( \begin{array}{cc}
\AA & -\AA\one
\end{array} \right)$, we have
\[
\norm{(\AAZ)^{\trp} \cca}_2^2
= \norm{\AA^{\trp} \cca}_2^2 + \norm{\one^{\trp} \AA^{\trp}\cca}_2^2.
\]
By Cauchy-Schwarz inequality,
\[
\norm{\one^{\trp} \AA^{\trp}\cca}_2 \le \norm{\one}_2 \norm{\AA^{\trp}\cca}_2
= \sqrt{\na} \norm{\AA^{\trp}\cca}_2.
\]
Plugging this gives
\[
\norm{(\AAZ)^{\trp} \cca}_2^2 
\le (\na+1) \norm{\AA^{\trp}\cca}_2^2.
\]
By Courant-Fischer theorem,
\begin{align*}
\norm{(\AAZ)^{\trp}\cca}_{\left((\AAZ)^{\trp}\AAZ \right)^{\dagger}}^2 \le \lambda_{\max} \left( \left((\AAZ)^{\trp}\AAZ \right)^{\dagger} \right) \norm{(\AAZ)^{\trp}\cca}_2^2
\le \lambda_{\min}^{-1} \left( (\AAZ)^{\trp}\AAZ \right) \norm{(\AAZ)^{\trp}\cca}_2^2
\end{align*}
and since $\AA^{\trp}\cca$ is in the eigenspace of $\AA^{\trp}\AA$, 
\[
\norm{\AA^{\trp}\cca}_{\left( \AA^{\trp}\AA \right)^{\dagger}}^2 \ge \lambda_{\min}\left( \left( \AA^{\trp}\AA \right)^{\dagger}\right) \norm{\AA^{\trp}\cca}_2^2
= \lambda_{\max}^{-1}\left(  \AA^{\trp}\AA \right) \norm{\AA^{\trp}\cca}_2^2.
\]
Combining the above two inequalities,
\[
\norm{(\AAZ)^{\trp}\cca}_{\left((\AAZ)^{\trp}\AAZ \right)^{\dagger}}^2
\le \frac{\lambda_{\max} \left( \AA^{\trp}\AA \right)}{ \lambda_{\min} \left( (\AAZ)^{\trp}\AAZ \right)} (\na+1) \norm{\AA^{\trp}\cca}_{\left( \AA^{\trp}\AA \right)^{\dagger}}^2.
\]
Given equations~\eqref{eqn:GtoGZ_lhs} and~\eqref{eqn:GtoGZ_rhs}, and condition $\norm{\AAZ\xxz - \cca}_2 \le \epsilon \norm{\PPi_{\AAZ}\cca}_2$, we have
\begin{align*}
\norm{\AA\xx - \PPi_{\AA}\cca}_2 
\le \epsilon_2 \frac{\sigma_{\max}(\AA)}{\sigma_{\min}(\AAZ)} \sqrt{\na+1} \norm{\PPi_{\AA}\cca}_2.
\end{align*}
This completes the proof.
\end{proof}

We compute the nonzero condition number of $\AAZ$. Note
\begin{align}
(\AAZ)^{\trp}\AAZ =  \left( \begin{array}{cc}
\MM & -\MM \one \\
- \one^{\trp} \MM & \one^{\trp} \MM \one
\end{array} \right),
\end{align}
where $\MM = \AA^{\trp}\AA$.

\begin{claim}
$\lambda_{\max} \left( (\AAZ)^{\trp}\AAZ \right) \le 2\na \lambda_{\max}(\AA^{\trp} \AA)$.
\label{claim:max_eig}
\end{claim}
\begin{proof}
Let $\lambda_1$ be the largest eigenvalue of $ (\AAZ)^{\trp}\AAZ $, and $\yy = (\tilde{\yy}; \alpha)$ be the associated eigenvector of unit length. 
Let $\mu_1$ be the largest eigenvalue of $\AA^{\trp} \AA = \MM$.
\begin{align*}
\lambda_1 = 
\yy^{\trp} (\AAZ)^{\trp}\AAZ \yy &= \left(\begin{array}{cc}
\tilde{\yy}^{\trp} &
\alpha
\end{array} \right) 
\left( \begin{array}{cc}
\MM & -\MM \one \\
- \one^{\trp} \MM & \one^{\trp} \MM \one
\end{array} \right)
\left( \begin{array}{c}
\tilde{\yy} \\
\alpha
\end{array} \right)  \\
& =
\tilde{\yy}^{\trp} \MM \tilde{\yy} - 2\alpha \one^{\trp} \MM  \tilde{\yy}
+ \alpha^2 \one^{\trp}\MM \one.
\end{align*}
By Courant-Fischer Theorem, 
\[
\tilde{\yy}^{\trp} \MM \tilde{\yy} \le \mu_1 \tilde{\yy}^{\trp}\tilde{\yy},
\ \mbox{ and } \ 
\one^{\trp} \MM \one \le \mu_1 \na.
\]
By Cauchy-Schwarz inequality,
\[
\abs{\one^{\trp} \MM \tilde{\yy}}  \le \norm{\one}_2 \norm{\MM\tilde{\yy}}_2
\le \sqrt{\na} \mu_1 \norm{\tilde{\yy}}_2.
\]
Putting all together,
\begin{align*}
\lambda_1 & \le \mu_1 \left( \norm{\tilde{\yy}}_2^2 + 2\abs{\alpha}\sqrt{\na} \norm{\tilde{\yy}}_2
+ \alpha^2\na \right) \\
& = \mu_1 \left( \norm{\tilde{\yy}}_2 + \abs{\alpha}  \sqrt{\na} \right)^2.
\end{align*}
Since $\norm{\tilde{\yy}}_2, \abs{\alpha} \le 1$, we have
\[
\lambda_1 \le 2\mu_1 \na.
\]
This completes the proof.
\end{proof}

Before lower bounding the smallest nonzero singular value of $\AAZ$, we characterize the null space of $\AAZ$ by the following Claim.
\begin{claim}
\[
\nulls(\AAZ) = \Span{\one, \left( \begin{array}{c}
\zz \\
0
\end{array} \right): \zz \in \nulls(\AA) }.
\]
\end{claim}
\begin{proof}
Let $\SS$ be the subspace of  $\Span{\one, \left( \begin{array}{c}
\zz \\
0
\end{array} \right): \zz \in \nulls(\AA) }$.

We first show that $\SS \subseteq \nulls(\AAZ)$. Clearly, 
\[
\AAZ \one = \left( \begin{array}{cc}
\AA & -\AA \one
\end{array} \right) \one = {\bf 0}.
\]
For each $\zz \in \nulls(\AA)$, we have 
\[
\AAZ \left( \begin{array}{c}
\zz \\
0
\end{array} \right) 
= \left( \begin{array}{cc}
\AA & -\AA\one
\end{array}  \right) 
\left( \begin{array}{c}
\zz \\
0
\end{array} \right)
= \AA\zz = {\bf 0}.
\]
Thus, $\SS \subseteq \nulls(\AAZ)$.

We then show that $\SS \supseteq \nulls(\AAZ)$. 
For any $\zz' \in \nulls(\AAZ)$, $\zz' - \zz'_{\na+1} \one \in \nulls(\AAZ)$.
Let $\zz \in \mathbb{R}^n$ such that the $i$th entry of $\zz$ is $\zz'_i - \zz'_{\na+1}$.
Thus,
\[
{\bf 0 } =\AAZ \left( \zz' - \zz'_{\na+1}\one \right)
= \left( \begin{array}{cc}
\AA & -\AA\one
\end{array}  \right) 
\left( \begin{array}{c}
\zz \\
0
\end{array} \right) 
= \AA\zz.
\]
That is, any vector in $\SS$ can be written as a linear combination of $\one$ and $\left( \begin{array}{c}
\zz \\
0
\end{array} \right)$ where $\zz \in \nulls(\AA)$.
Thus, $\SS \supseteq \nulls(\AAZ)$.

Therefore, $\nulls(\AAZ) = \SS$.
\end{proof}

By the above claim, we know that $\AA$ and $\AAZ$ have same rank.
We bound the smallest nonzero singular value of $\AAZ$ by the following Claim.

\begin{claim}
$\lambda_{\min}\left((\AAZ)^{\trp}\AAZ \right) \ge  \lambda_{\min}(\AA^{\trp} \AA) / (\na+1)$.
\end{claim}
\begin{proof}
Let $\lambda_k$ be the smallest nonzero eigenvalue of $(\AAZ)^{\trp} \AAZ$, 
and $\yy = \left( \tilde{\yy}; \alpha \right)$ be the associated eigenvector of unit length.
Let $\mu_k$ be the smallest nonzero eigenvalue of $\AA^{\trp} \AA = \MM$, and $\xx$
be the associated eigenvector of unit length.
\[
\lambda_k = \tilde{\yy}^{\trp} \MM \tilde{\yy} - 2\alpha \one^{\trp} \MM  \tilde{\yy}
+ \alpha^2 \one^{\trp}\MM \one.
\]
Since $\MM$ is symmetric PSD, it can be written as $\MM = \MM^{1/2} \MM^{1/2}$.
\[
\lambda_k = \norm{\MM^{1/2} \left(\tilde{\yy} - \alpha\one \right) }_2^2.
\]
We decompose $\tilde{\yy} - \alpha \one := \zz_1 + \zz_2$, where $\zz_1 \in \nulls(\MM)$ and $\zz_2 \perp \nulls(\MM)$.
Then, $\lambda_k = \norm{\MM^{1/2} \zz_2}_2^2$.
By Courant-Fischer Theorem
\[
\lambda_k 
\ge \mu_k \norm{\zz_2}_2^2.
\]
To lower bound $\lambda_k$, it suffices to lower bound $\norm{\zz_2}_2$.

Since $\yy = \left( \tilde{\yy}; \alpha \right)$, 
\[
\norm{\tilde{\yy} - \alpha \one}_2^2 = \norm{\yy - \alpha\left(\begin{array}{c}
\one \\
1
\end{array} \right)}_2^2.
\]
Since $\yy \perp \one$, which is in $\nulls(\AAZ)$, we have
\[
\norm{\tilde{\yy} - \alpha \one}_2^2 
= \norm{\yy}_2^2 + \alpha^2(1+\na)
= \norm{\tilde{\yy}}_2^2 + \alpha^2(2+\na).
\]
%
Vectors $\tilde{\yy}, -\alpha \one, \tilde{\yy} - \alpha\one$ form a triangle. Let $\theta$ be the angle between $\tilde{\yy}$ and $\alpha\one$. See Figure~\ref{fig:z2}.
\begin{figure}[htb]
\centering
\includegraphics[scale=0.5]{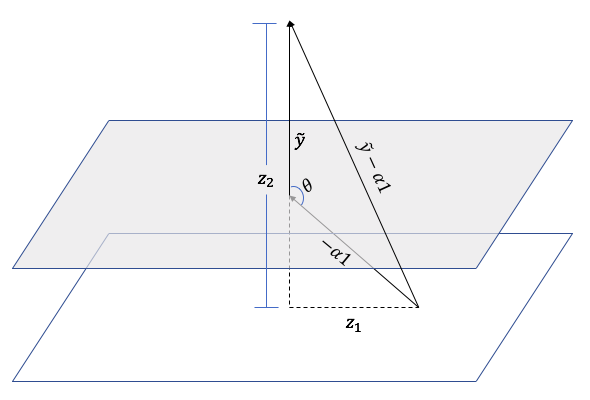}
\caption{The relation between vectors $\tilde{\yy}, -\alpha \one$ and $\tilde{\yy} - \alpha \one$.}
\label{fig:z2}
\end{figure} 

Then,
\[
\cos(\theta) = \frac{\norm{\tilde{\yy}}_2^2 + \norm{\alpha\one}_2^2 - \norm{\tilde{\yy} - \alpha\one}_2^2}{2\norm{\tilde{\yy}}_2 \norm{\alpha\one}_2}
= -\frac{\abs{\alpha}}{\norm{\tilde{\yy}}_2 \sqrt{\na}}.
\]
Thus,
\[
\norm{\zz_2}_2 = \norm{\tilde{\yy}}_2 +  \norm{\alpha \one}_2 \cos(\pi - \theta) = \norm{\tilde{\yy}}_2 +  \frac{\alpha^2}{\norm{\tilde{\yy}}_2} 
\ge \norm{\tilde{\yy}}_2.
\]
On the other hand,
\[
1 = \alpha^2 + \norm{\tilde{\yy}}_2^2
= \left( \sum_i \tilde{\yy}_i \right)^2 + \sum_i \tilde{\yy}_i^2
\le (\na+1) \sum_i \tilde{\yy}_i^2
= (\na+1) \norm{\tilde{\yy}}_2^2.
\]
Therefore,
\[
\lambda_k \ge \mu_k \norm{\tilde{\yy}}_2^2 \ge \frac{\mu_k}{\na+1}.
\]
This completes the proof.
\end{proof}

The above lemmas give an upper bound on the condition number of $\AAZ$.
\begin{lemma}
$\kappa(\AAZ) \le O(\na^{3/2}) \kappa(\AA)$.
\label{lem:genToZeroSum_cond}
\end{lemma}

\subsection{$\genZCl \leq_{f} \genZtwoCl$}

In this section, we show the reduction from $\genZCl$ to $\genZtwoCl$.
Given an instance of linear system $(\AAZ, \ccz, \epsilon_1)$ with $\AAZ \in \genCl_{\text{z}}$, the goal is to construct an instance of linear system $(\AAZtwo, \ccztwo, \eps_2)$ with $\AAZtwo \in \genCl_{\text{z},2}$ such that, there is a map between the solutions of these two linear systems.

Algorithm~\ref{alg:GZtoGZT} shows the construction of $(\AAZtwo, \ccztwo, \eps_2)$, and Algorithm~\ref{alg:GZTtoGZsolnback} shows the transform from a solution of $(\AAZtwo, \ccztwo, \eps_2)$ to a solution of $(\AAZ, \ccz, \epsilon_1)$.
Note $\ccztwo = (\ccz; 0)$.

\begin{algorithm}[htb]
\renewcommand{\algorithmicrequire}{\textbf{Input:}}
\renewcommand\algorithmicensure {\textbf{Output:}}
    \caption{\label{alg:GZtoGZT}\algGZtoGZT}
    \begin{algorithmic}[1]
\REQUIRE{$(\AAZ,\ccz,\eps_{1})$ where $\AAZ \in \genCl_{\text{z}}$ is an $m \times
  n$ matrix, $\ccz \in \mathbb{R}^m$, and $0 < \eps_{1} < 1$.}
\ENSURE{$(\AAZtwo,\ccztwo,\eps_{2})$ where $\AAZtwo \in \genCl_{\text{z},2}$ is an $m' \times
  n'$ matrix, $\ccztwo \in \mathbb{R}^{m'}$, and $0 < \eps_{2} < 1$.}
\STATE $w \leftarrow 3\eps_1^{-1} \sqrt{m}  \sigma_{\max}(\AAZ)\norm{\AAZ}_{\infty} \norm{\ccz}_2$ \label{lin:GZtoGZT_weight}
\STATE $k^* \leftarrow \min \left\{ k \in \mathbb{Z}: 2^k \ge \norm{\AAZ}_{\infty} \right\}$
\STATE Let $\aa \in \mathbb{R}^n$ 
\FOR{$i \leftarrow 1$ to $n$}
\STATE $\aa_i \leftarrow 2^{k^*} - \norm{\AAZ_i}_1 / 2$
\label{lin:GZtoGZT_vec_a}
\ENDFOR
\STATE $\AAZtwo = \left( \begin{array}{ccc}
\AAZ & \aa & -\aa \\
{\bf 0} & w & -w
\end{array} \right) $
\STATE $\ccztwo \leftarrow \left( \begin{array}{c}
\ccz \\
0
\end{array} \right)$
\RETURN $\left( \AAZtwo, \ccztwo, O \left( \frac{\eps_1}{\sigma_{\max}(\AAZ) \norm{\ccz}_2} \right) \right)$. 
\end{algorithmic}
\end{algorithm}

\begin{algorithm}[htb]
\renewcommand{\algorithmicrequire}{\textbf{Input:}}
\renewcommand\algorithmicensure {\textbf{Output:}}
    \caption{\label{alg:GZTtoGZsolnback}\algGZTtoGZsoln}
    \begin{algorithmic}[1]
	\REQUIRE $m \times n$ matrix $\AAZ \in \genCl_{\text{z}}$, 
    $m' \times n'$ matrix $\AAZtwo \in \genCl_{\text{z},2}$, vector $\ccz \in \mathbb{R}^m$,
    vector $\xx \in \mathbb{R}^{n'}$.
    \ENSURE Vector $\yy \in \mathbb{R}^{n}$.
    \IF{$(\AAZ)^\trp \ccz = {\bf 0}$}
    	\RETURN $\yy \assign {\bf 0}$
    \ELSE
	    \RETURN $\yy \assign \xx_{1:n}$
	\ENDIF
\end{algorithmic}
\end{algorithm}

\begin{lemma}
Let $\AAZtwo\xxztwo = \ccztwo$ be the linear system returned by a call to \algGZtoGZT$(\AAZ, \ccz, \epsilon_1)$.
Then,
\[
\nnz(\AAZtwo) = O \left( \nnz(\AAZ) \right).
\]
and the largest entry of $\AAZtwo$ is at most $\eps_1^{-1} \sqrt{m} \sigma_{\max}(\AAZ)\norm{\AAZ}_{\infty} \norm{\ccz}_2$.
\label{lem:gnToPowerTwo_dim}
\end{lemma}

\begin{claim}
$\nulls(\AAZtwo) = \Span{\left( \begin{array}{c}
{\bf 0} \\
1 \\
1
\end{array} \right),
\left( \begin{array}{c}
\xx \\
0 \\
0
\end{array} \right): \xx \in \nulls(\AAZ)}$.
\label{clm:GZtoGZT_null}
\end{claim}
\begin{proof}
Let $\yy = \left( \xx; \alpha; \beta \right) \in \mathbb{R}^{\na+2}$ such that $\yy \in \nulls(\AAZtwo)$, that is,
\[
\AAZtwo \yy = \left( \begin{array}{ccc}
\AA & \aa & -\aa \\
{\bf 0} & w & -w
\end{array} \right) 
\left( \begin{array}{c}
\xx \\
\alpha \\
\beta
\end{array} \right)
= \left( \begin{array}{c}
\AAZ\xx + (\alpha - \beta) \aa  \\
(\alpha - \beta) w
\end{array} \right)
= {\bf 0}.
\]
Since $w > 0$, we have
\[
\alpha = \beta.
\]
and thus
\[
\AAZ\xx = {\bf 0}.
\]
Therefore, 
\[
\nulls(\AAZtwo) = \Span{\left( \begin{array}{c}
{\bf 0} \\
1 \\
1
\end{array} \right),
\left( \begin{array}{c}
\xx \\
0 \\
0
\end{array} \right): \xx \in \nulls(\AAZ)}.
\]
This completes the proof.
%
%
\end{proof}


We write $\xxztwo$ as $\left(\xx; \alpha, \beta \right)$.
By Claim~\ref{clm:GZtoGZT_null}, the vector $\left({\bf 0}; 1; 1 \right)$ is in the null space of $\AAZtwo$. 
Thus, 
\[
\AAZtwo \xxztwo = 
\AAZtwo\left( \xxztwo - \frac{\alpha+\beta}{2} \left( \begin{array}{c}
{\bf 0} \\
1 \\
1
\end{array} \right) \right) =
\AAZtwo \left( \begin{array}{cc}
\xx \\
\frac{\alpha - \beta}{2} \\
\frac{\beta - \alpha}{2}
\end{array} \right).
\]
Without loss of generality, we assume $\alpha = \beta$, that is, $\xxztwo = \left( \xx; \alpha; -\alpha \right)$.

We first show that the exact solutions of the two linear systems are close.

\begin{lemma}
Let $\xxztwoopt \in \argmin_{\xx} \norm{\AAZtwo\xx - \ccztwo}_2$.
Write $\xxztwoopt = \left( \ss^*; \alpha; -\alpha \right)$.
Then,
\[
\ss^* \in \argmin_{\xx} \norm{\AAZ\xx - \ccz}_{\PP},
\]
where $\PP = \II - \frac{\aa\aa^{\trp}}{w^2 + \norm{\aa}_2^2}$.
\label{lem:genToPower2_exact}
\end{lemma}
\begin{proof}
Without loss of generality, we write $\xxztwo = \left( \xx; \alpha; -\alpha \right)$.
We expand $\AAZtwo\xxztwo - \ccztwo$, 
\begin{align*}
& \min_{\xxztwo} \norm{\AAZtwo\xxztwo - \ccztwo}_2^2 \\
=& \min_{\xx, \alpha}\norm{\left( \begin{array}{ccc}
\AAZ & \aa & -\aa \\
{\bf 0} & w & -w
\end{array} \right)\left( \begin{array}{c}
\xx \\
\alpha \\
-\alpha
\end{array} \right) - \left(\begin{array}{c}
\ccz \\
0
\end{array} \right) }_2^2 \\
=& \min_{\xx, \alpha}\norm{\AAZ\xx - \ccz + 2\alpha\aa}_2^2
+ 4w^2 \alpha^2 \\
= & \min_{\xx} \min_{\alpha} 4\left( w^2 + \norm{\aa}_2^2 \right) \left( \alpha + \frac{\aa^{\trp}(\AAZ\xx-\ccz)}{2\left(w^2 + \norm{\aa}_2^2 \right) } \right)^2
+ \norm{\AAZ\xx- \ccz}_2^2 - \frac{\left(\aa^{\trp}(\AAZ\xx-\ccz) \right)^2}{w^2 + \norm{\aa}_2^2 }.
\end{align*}
Since the minimization over $\xx$ and $\alpha$ is independent of each other, for any fixed $\xx$, the above value is minimized when 
\[
\alpha = - \frac{\aa^{\trp}(\AAZ\xx-\ccz)}{2\left(w^2 + \norm{\aa}_2^2 \right) }.
\]
Plugging this value of $\alpha$ gives
\begin{align*}
\min_{\xxztwo} \norm{\AAZtwo\xxztwo - \ccztwo}_2^2
&= \min_{\xx} \norm{\AAZ\xx- \ccz}_2^2 - \frac{\left(\aa^{\trp}(\AAZ\xx-\ccz) \right)^2}{w^2 + \norm{\aa}_2^2 } \\
& = \min_{\xx} \left( \AAZ\xx - \ccz \right)^{\trp} \left( \II 
- \frac{\aa\aa^{\trp}}{w^2 + \norm{\aa}_2^2} \right) \left( \AAZ\xx - \ccz \right) \\
& = \min_{\xx} \norm{\AAZ\xx - \ccz}_{\PP}^2,
\end{align*}
where $\PP = \II - \frac{\aa\aa^{\trp}}{w^2 + \norm{\aa}_2^2}$.
The fact $\xxztwoopt \in \argmin_{\xx} \norm{\AAZtwo\xx - \ccztwo}_2$ implies that
\[
\ss^*
\in \argmin_{\xx} \norm{\AAZ \xx - \ccz}_{\PP},
\]
which completes the proof.
\end{proof}

\begin{lemma}
Let $\xxztwoopt \in \argmin_{\xx} \norm{\AAZtwo\xx - \ccztwo}_2$ such that $\xxztwoopt$ has the form $\left( \ss^*; \alpha; -\alpha \right)$ and $\norm{\xxztwoopt}_2$ is minimized.
Let $\xx^* \in \argmin_{\xx} \norm{\AAZ\xx - \ccz}_2$ such that $\norm{\xx^*}_2$ is minimized.
Then,
\[
\norm{\AAZ \left( \ss^* - \xx^* \right)}_2 \le \frac{\norm{\aa}_2}{w} \sigma_{\max} (\AAZ) \norm{\ccz}_2 \norm{\PPi_{\AAZ} \ccz}_2.
\]
\label{lem:genToPower2_exact_reduction}
\end{lemma}
\begin{proof}
Let $\gamma \defeq \frac{\norm{\aa}^2_2}{w^2 + \norm{\aa}_2^2}$.
By the definition of $\PP$ in Lemma~\ref{lem:genToPower2_exact}, 
\[
\left(1 - \gamma \right) \II \pleq \PP \pleq \II.
\] 
Thus,
\begin{align}
(1 - \gamma) \norm{\AAZ\ss^* - \ccz}^2_2
\le \norm{\AAZ\ss^* - \ccz}^2_{\PP}
\le \norm{\AAZ\xx^* - \ccz}^2_{\PP}
\le \norm{\AAZ\xx^* - \ccz}^2_2.
\label{eqn:genToZeroSumTwo_relation}
\end{align}
The second inequality is due to $\ss^* \in \argmin_{\xx} \norm{\AAZ\xx - \ccz}_{\PP}$.
Note that $\AAZ\xx^* - \ccz$ is orthogonal to the column space of $\AAZ$, in particular, it is orthogonal to $\AAZ(\ss^* - \xx^*)$.
\[
\norm{\AAZ(\ss^* - \xx^*)}_2^2
= \norm{\AAZ\ss^* - \ccz}_2^2 - \norm{\AAZ\xx^* - \ccz}_2^2.
\]
By Equation~\eqref{eqn:genToZeroSumTwo_relation}, 
\[
\norm{\AAZ(\ss^* - \xx^*)}_2^2
\le \frac{\gamma}{1-\gamma} \norm{\AAZ\xx^* - \ccz}_2^2
\le \frac{\gamma}{1-\gamma} \norm{\ccz}_2^2.
\]
If $\PPi_{\AAZ} \ccz = {\bf 0}$, then $\ss^* = \xx^* = {\bf 0}$ and the claim holds. Otherwise,
by Lemma~\ref{lem:ProjLengthLower}, 
\[
\norm{\AAZ(\ss^* - \xx^*)}_2
\le \sqrt{\frac{\gamma}{1-\gamma}} \sigma_{\max}(\AAZ) \norm{\ccz}_2 \norm{\PPi_{\AAZ} \ccz}_2.
\]
Plugging the value of $\gamma$ completes the proof.
\end{proof}

We then show the approximate solvers of the two linear systems are close.

\begin{lemma}
Let $\eps_2$ be the error parameter returned by a call to \algGZtoGZT$(\AAZ, \ccz, \eps_1)$ (Algorithm~\ref{alg:GZtoGZT}).
Let $\xxztwo$ be a vector such that $\norm{\AAZtwo\xxztwo - \PPi_{\AAZtwo}\ccztwo}_2 \le \epsilon_2 \norm{\PPi_{\AAZtwo}\ccztwo}_2$. 
Let $\xx$ be the vector returned by Algorithm~\ref{alg:GZTtoGZsolnback}. 
Then,
\[
\norm{\AAZ\xx - \PPi_{\AAZ}\ccz}_2 \le \eps_1
\norm{\PPi_{\AAZ}\ccz}_2.
\]
\label{lem:genToPower2_eps}
\end{lemma}

\begin{proof}
If $\PPi_{\AAZ} \ccz = {\bf 0}$, then $\xx = {\bf 0}$ and the statement holds.
If $\PPi_{\AAZ} \ccz \neq {\bf 0}$, then $\xx = \xxztwo_{1:\na+1}$.
Let $\xx^* \in \argmin_{\xx} \norm{\AAZ\xx - \ccz}_2$. 
We now bound the difference between our solution $\xx$ and $\xx^*$.
Let $\xxztwoopt \in \argmin_{\xx} \norm{\AAZtwo\xx - \ccztwo}_2$ of the form $(\ss^*; \alpha; -\alpha)$.
By the triangle inequality,
\begin{align}
\norm{\AAZ\left(\xx - \xx^*\right)}_2
\le \norm{\AAZ\left(\xx - \ss^*\right)}_2 + \norm{\AAZ\left(\ss^* - \xx^*\right)}_2.
\label{eqn:GZtoGZT_approx_all}
\end{align}

The second term is upper bound by Lemma~\ref{lem:genToPower2_exact_reduction}.
It remains to upper bound the first term.
Let 
\begin{align}
\delta \defeq \epsilon_2 \norm{\PPi_{\AAZtwo} \ccztwo}_2.
\label{eqn:genToPower2_delta_def}
\end{align}

Without loss of generality, we write $\xxztwo$ as $\left( \xx; \beta; -\beta \right)$, where $\xx$ is the output of Algorithm~\ref{alg:GZTtoGZsolnback}.
\begin{align*}
\norm{\AAZtwo\xxztwo - \PPi_{\AAZtwo}\ccztwo}_2^2
& = \norm{\AAZtwo \left( \xxztwo - \xxztwoopt \right)}_2^2 \\
& = \norm{\left( \begin{array}{ccc}
\AAZ & \aa & -\aa \\
{\bf 0} & w & -w
\end{array} \right) 
\left( \begin{array}{c}
\xx - \ss^* \\
\beta - \alpha \\
\alpha - \beta
\end{array} \right)
}_2^2 \\
& = \norm{\left( \begin{array}{c}
\AAZ \left( \xx - \ss^* \right) + 2\left( \beta - \alpha \right) \aa \\
2w\left( \alpha - \beta \right)
\end{array} \right)}_2^2 \\
& = \norm{\AA \left( \xx - \ss^* \right) + 2\left( \beta - \alpha \right) \aa}_2^2
+ 4w^2 \left( \alpha - \beta \right)^2.
\end{align*}
Since $\norm{\AAZtwo\xxztwo - \PPi_{\AAZtwo}\ccztwo}_2 \le \delta$, we have
\[
4w^2 \left( \alpha - \beta \right)^2 \le \delta^2,
\]
that is,
\begin{align}
\left( \alpha - \beta \right)^2 \le \frac{\delta^2}{4w^2}.
\label{eqn:GZtoGZT_approx_diff}
\end{align}
Similarly, we have
\[
\norm{\AAZ \left( \xx - \ss^* \right) + 2\left( \beta - \alpha \right) \aa}_2 \le \delta.
\]
By the triangle inequality,
\[
\norm{\AAZ \left( \xx - \ss^* \right)}_2 - 2\norm{\left( \beta - \alpha \right) \aa}_2
\le \delta. 
\]
Plugging~\eqref{eqn:GZtoGZT_approx_diff} and~\eqref{eqn:genToPower2_delta_def} into the above inequality, and rearranging it,
\[
\norm{\AAZ \left( \xx - \ss^* \right)}_2
\le  \left( 1 + \frac{\norm{\aa}_2}{w} \right) \epsilon_2 
\norm{\ccztwo}_2.
\label{eqn:GZtoGZT_approx_x_diff}
\]
By Lemma~\ref{lem:ProjLengthLower}, 
\[
\norm{\AAZ \left( \xx - \ss^* \right)}_2
\le  \left( 1 + \frac{\norm{\aa}_2}{w} \right) \epsilon_2 
\sigma_{\max}(\AAZ) \norm{\ccztwo}_2 \norm{\PPi_{\AAZ} \ccz}_2.
\]

Together with Lemma~\ref{lem:genToPower2_exact_reduction} and Equation~\eqref{eqn:GZtoGZT_approx_all}, we have
\[
\norm{\AAZ(\xx - \xx^*)}_2
\le \left( \frac{\norm{\aa}_2}{w}  + \left( 1 + \frac{\norm{\aa}_2}{w} \right) \epsilon_2 \right)
\sigma_{\max}(\AAZ) \norm{\ccz}_2\norm{\PPi_{\AAZ}\ccz}_2.
\]
According to our setting of $w$ and $\aa$ in line~\ref{lin:GZtoGZT_weight} and line~\ref{lin:GZtoGZT_vec_a} of Algorithm~\ref{alg:GZtoGZT},
\[
w 
\ge \frac{\norm{\aa}_2 }{\epsilon_2}.
\]
This implies that
\[
\norm{\AAZ\xx - \PPi_{\AAZ}\ccz}_2 \le \eps_1
\norm{\PPi_{\AAZ}\ccz}_2,
\]
which completes the proof.
\end{proof}

We then bound the singular values of $\AAZtwo$.

\begin{claim}
$\lambda_{\max}\left((\AAZtwo)^{\trp}\AAZtwo \right) \le 
O\left( \eps_1^{-2} m \lambda_{\max}((\AAZ)^\trp \AAZ) \norm{\AAZ}_{\infty}^2 \norm{\ccz}_2^2  \right)
$.
\end{claim}
\begin{proof}
Let $\lambda_1$ bet he largest eigenvalue of $(\AAZtwo)^{\trp}\AAZtwo$, and $\yy = \left( \tilde{\yy}; \alpha; -\alpha \right)$ be the associated eigenvector of unit length.
Let $\mu_1$ be the largest eigenvalue of $(\AAZ)^{\trp}\AAZ$.
\begin{align*}
\lambda_1 &= \norm{ \left( \begin{array}{ccc}
\AAZ & \aa & -\aa \\
{\bf 0} & w & -w
\end{array} \right)
\yy
}_2^2 \\
& = \norm{ \left( \begin{array}{ccc}
\AAZ & \aa & -\aa \\
{\bf 0} & 1 & -1
\end{array} \right)
\yy + 
\left( \begin{array}{ccc}
{\bf 0} & {\bf 0} & {\bf 0} \\
{\bf 0} & w-1 & -w+1
\end{array} \right) \yy
}_2^2 \\
& = \tilde{\yy}^{\trp}(\AAZ)^{\trp}\AAZ\tilde{\yy} + 4\alpha \aa^{\trp}\AAZ\tilde{\yy} + 4\alpha^2 \left( \aa^{\trp}\aa + 1 \right)
+ 8 \alpha^2 (w-1) + 4\alpha^2 (w-1)^2.
\end{align*}
By the Courant-Fischer Theorem,
\[
\tilde{\yy}^{\trp}(\AAZ)^{\trp}\AAZ\tilde{\yy} \le \mu_1 \tilde{\yy}^{\trp}\tilde{\yy} 
\le \mu_1.
\]
By Cauchy-Schwarz inequality,
\[
\abs{\aa^{\trp}\AAZ\tilde{\yy}} \le \norm{\aa}_2 \norm{\AAZ\tilde{\yy}}_2
\le \norm{\AA}_{\infty}\sqrt{m} \sqrt{\mu_1}. 
\]
Thus,
\begin{align*}
\lambda_1 &\le \mu_1 + 4\abs{\alpha} \norm{\AAZ}_{\infty}\sqrt{m} \sqrt{\mu_1} + 4\alpha^2 \left( \norm{\AAZ}_{\infty}^2 m + 1 \right) + 12\alpha^2 w^2 \\
&\le \left( \sqrt{\mu_1} + 2 \norm{\AAZ}_{\infty} \sqrt{m} \right)^2 +  12\alpha^2 w^2 \\
& \le 2\mu_1 + 8 m \norm{\AAZ}_{\infty}^2 + 12 \alpha^2 w^2.
\end{align*}
Our setting of $w$ in line~\ref{lin:GZtoGZT_weight} of Algorithm~\ref{alg:GZtoGZT} gives
\[
\lambda_{\max}\left((\AAZtwo)^{\trp}\AAZtwo \right) \le 
2\mu_1 + 8 m \norm{\AAZ}_{\infty}^2 + 108 \epsilon_1^{-2} m \mu_1 \norm{\AAZ}_{\infty}^2
\norm{\ccz}_2^2.
\]
This completes the proof.
\end{proof}

\begin{claim}
$\lambda_{\min}((\AAZtwo)^{\trp}\AAZtwo) \ge \frac{2\lambda_{\min}((\AAZ)^{\trp}\AAZ)}{2\left( \norm{\AAZ}_{\infty}^2 m + 1 \right) + \lambda_{\min}((\AAZ)^{\trp}\AAZ)}$.
\end{claim}
\begin{proof}
Let 
\[
\CC = \left( \begin{array}{ccc}
\AAZ & \aa & -\aa \\
{\bf 0} & 1 & -1
\end{array} \right).
\]
Note 
\[
\AAZtwo = \CC + \left( \begin{array}{ccc}
{\bf 0} & {\bf 0} & {\bf 0} \\
{\bf 0} & w-1 & -w+1
\end{array} \right).
\]
We can check that
\[
(\AAZtwo)^{\trp} \AAZtwo = \CC^{\trp}\CC + \left( \begin{array}{ccc}
{\bf 0} & {\bf 0} & {\bf 0} \\
{\bf 0} & w^2 - 1 & - (w^2 -1) \\
{\bf 0} & -(w^2 -1) & w^2 -1
\end{array} \right).
\]
By our setting, $w^2 - 1 \ge 0$.
The second matrix is a rank-one PSD matrix. 
Since $\AAZtwo$ and $\CC$ has same null space, 
\[
\lambda_{\min}\left( (\AAZtwo)^{\trp}\AAZtwo \right) \ge \lambda_{\min}\left(\CC^{\trp}\CC \right).
\]
It suffices to lower bound the smallest nonzero eigenvalue of $\CC^{\trp}\CC$.

Let $\lambda_k \defeq \lambda_{\min} \left( \CC^{\trp}\CC \right)$, and $\yy = \left( \tilde{\yy}; \alpha; -\alpha \right)$ be the associated eigenvector of unit length.
Let $\mu_k \defeq \lambda_{\min}\left((\AAZ)^{\trp}\AAZ \right)$, and $\xx$ be the associated eigenvector of unit length.
\begin{equation}
\begin{split}
\lambda_k & =  \tilde{\yy}^{\trp}(\AAZ)^{\trp}\AAZ\tilde{\yy} + 4\alpha \aa^{\trp}\AAZ\tilde{\yy} + 4\alpha^2 \left( \aa^{\trp}\aa + 1 \right) \\
& = \norm{\AAZ\tilde{\yy} + 2\alpha \aa}_2^2 + 4\alpha^2 \\
& \ge 4\alpha^2.
\end{split}
\label{eqn:lambdak}
\end{equation}
On the other hand,
\begin{align*}
\lambda_k & = \left( 2\alpha \sqrt{\aa^{\trp}\aa + 1} + \frac{\aa^{\trp}\AAZ\tilde{\yy}}{\sqrt{\aa^{\trp}\aa+1}} \right)^2
+ \tilde{\yy}^{\trp}(\AAZ)^{\trp}\AAZ\tilde{\yy} - \frac{\tilde{\yy}^{\trp}(\AAZ)^{\trp}\aa\aa^{\trp}\AAZ\tilde{\yy}}{\aa^{\trp}\aa+1} \\
& \ge \tilde{\yy}^{\trp}(\AAZ)^{\trp}\left( \II - \frac{\aa\aa^{\trp}}{\aa^{\trp}\aa+1} \right) \AAZ \tilde{\yy}.
\end{align*}
Take eigen-decomposition of the matrix in the middle,
\[
\lambda_k \ge \tilde{\yy}^{\trp}(\AAZ)^{\trp} \QQ \DD \QQ^{\trp} \AAZ\tilde{\yy},
\]
where $\DD := \diag \left( 1-\frac{\aa^{\trp}\aa}{\aa^{\trp}\aa +1} , 1,\ldots, 1 \right)$.
By Claim~\ref{clm:GZtoGZT_null}, $\yy \perp \nulls(\AAZtwo)$ implies 
$\tilde{\yy} \perp \nulls(\AAZ)$. 
By the Courant-Fischer Theorem,
\[
\lambda_k \ge \frac{1}{\aa^{\trp}\aa + 1} \norm{\AAZ\tilde{\yy}}_2^2
\ge \frac{1}{\aa^{\trp}\aa + 1} \mu_k \norm{\tilde{\yy}}_2^2.
\]
Together with~\eqref{eqn:lambdak}, 
\[
\lambda_k \ge \max \left\{ 4\alpha^2, \frac{\mu_k}{\aa^{\trp}\aa+1} \norm{\tilde{\yy}}_2^2 \right\}.
\]
Since $\norm{\yy}_2^2 = \tilde{\yy}_2^2 + 2\alpha^2 = 1$,
\[
\lambda_k \ge \frac{2\mu_k}{2\left( \aa^{\trp}\aa + 1 \right) + \mu_k}
\ge \frac{2\mu_k}{2\left( \norm{\AA}_{\infty}^2 m + 1 \right) + \mu_k}.
\]
This completes the proof.
\end{proof}

The above two lemmas give the following bound on the condition number of $\AAZtwo$.
\begin{lemma}
$\kappa(\AAZtwo) = O \left( \eps_1^{-1} \sqrt{m} \norm{\AAZ}_{\infty} \norm{\ccz}_2 \left( \sigma_{\max}(\AAZ) + \kappa(\AAZ) \sqrt{m} \norm{\AAZ}_{\infty} \right) \right)
$.
\label{lem:genToPower2_cond}
\end{lemma}

\begin{proof}[Proof of Lemma~\ref{lem:genToZeroSumTwo}]
By Lemma~\ref{lem:gnToZeroSumTwo_dim} and~\ref{lem:gnToPowerTwo_dim}, we have
\[
\nnz(\AAZtwo) = O(s).
\]
The smallest nonzero entry does not change.
The largest entry of $\AAZtwo$ is at most 
\[
\max \{ O(\norm{\AA}_{\infty}), w \}
= O \left( \epsilon^{-1} n^{2} \sqrt{m}\sigma_{\max}(\AA)  \norm{\AA}_{\infty} \norm{\cca}_2 \right).
\]
By Claim~\ref{clm:para_infty_norm} and~\ref{clm:para_l2_norm}, the largest entry is upper bounded by
\[
O \left( \epsilon^{-1} s^{9/2} U^3 \right).
\]

By Lemma~\ref{lem:genToZeroSum_cond} and~\ref{lem:genToPower2_cond}, and Claim~\ref{clm:para_infty_norm} and~\ref{clm:para_l2_norm}
\[
\kappa(\AAZtwo) = O\left( \epsilon^{-1} s^8 U^3 K \right).
\]
By Lemma~\ref{lem:genToZeroSum_eps} and~\ref{lem:genToPower2_eps}, we have
\[
\eps_2^{-1} = O \left( s^{5/2} U^2 \right) \eps^{-1}.
\]
This completes the proof.
\end{proof}


\section{2D Trusses}
\label{sec:trusses}

\newcommand{\yytil}{\tilde{\yy}}
\newcommand{\Adelta}{A_{\delta}}
In this section, 
we show that the matrix $\BB$ constructed by $\algGZTtoMCT$ (Algorithm~\ref{alg:GZ2toMC2}), is a 2D Truss Incidence Matrix defined in Definition~\ref{def:trusses}.
It follows that for any function $f$, $\genCl \leq_{f} \mctwoCl$ implies $\genCl \leq_{f}
\trusstwoCl$.
We assume that the algorithm $\algGZTtoMCT$ is called on a matrix
$\AA$ with no two identical rows: if there are identical rows, these
rows can be collapsed into one row without changing the associated
normal equations, by reweighting the resulting row, similar to the
technique used in the proof of Claim~\ref{clm:ipm_distinct_rows}. Details are left
to the reader.
A key step is to show
that a 2-commodity
gadget in the reduction corresponds to a 2D truss subgraph,
which we call the 2D-truss gadget.

Without loss of generality, we let $\uu$-variables correspond to the
horizonal axis and $\vv$-variables to the vertical axis of the 2D plane.
According to Definition~\ref{def:MC2} and~\ref{def:trusses}:
\begin{enumerate}
\item an equation $\uu_i - \uu_j = 0$ in a 2-commodity linear system corresponds to a horizontal edge in the 2D plane;
\item an equation $\vv_i - \vv_j = 0$ in a 2-commodity linear system corresponds to a vertical edge in the 2D plane;
\item an equation $\uu_i - \vv_i - (\uu_j - \vv_j) = 0$ in a 2-commodity linear system corresponds to a diagonal edge in the 2D plane.
\end{enumerate}
Note that our reduction here heavily relies on the
ability to choose arbitrary weights.
In particular, the weights on the elements are not
related at all with the distances between the
corresponding vertices.

Our strategy for picking the coordinates of the 
vertices of the constructed 2D truss is the following:
we first pick the coordinates of the original $n$ vertices randomly, and then determine the coordinates of the new vertices constructed in the reduction to satisfy all the truss equations.

For the $n$ original vertices, we pick their $\uu$-coordinates arbitrarily and pick their $\vv$-coordinates randomly.
We pick 
an $n$-dimensional random vector $\yy$
 uniformly distributed on the $n$-dimensional sphere centered at the
 origin and with radius $R =\norm{\AA}_1 n^{10}$.
We then round each entry of $\yy$ to have precision $\delta =
10^{-10}$, so that the total number of bits used to store an entry is
at most $O(\log(n \norm{\AA}_1) )$.
Let $\yytil$ be the vector after rounding.
We assign the $\vv$-coordinate of the $i$th vertex to be the $i$th entry of $\yytil$.

We then pick the coordinates of the new vertices in the order they are created.
Note that each time we replace two vertices in the current equations,
say $\ss^{j_1}, \ss^{j_2}$, whose coordinates have already been
determined, we create a 2D truss gadget with 7 new vertices, say
$\ss^{t}, \ss^{t+1}, \ldots, \ss^{t+6}$ (See Algorithm~\ref{alg:MC2gadget}
$\algMCToGZGadget$ for the construction.).
According to the construction of this gadget, the new vertices
$\ss^{t+1}, \ldots, \ss^{t+6}$ only appear in this single gadget, whose coordinates do not affect other vertices.
Figure~\ref{fig:trussGadget} is the corresponding subgraph which satisfies all the equations in the 2D truss gadget.
Note the two triangles $(\ss^{t+3}, \ss^{t+5}, \ss^{t+6})$ and $(\ss^{t+3}, \ss^{t+4}, \ss^{t+5})$ need to be isosceles right triangles, 
which implies $\vv_t = (\vv_{j_1} + \vv_{j_2})/2$.
Note also that we can assign $\uu$-coordinates to the new vertices
which are not between the $\uu$-coordinates of $\ss^{j_1}$ and
$\ss^{j_2}$.
In fact, using an appropriate choice of $\uu$-coordinates
and edge weights, we can always place $\ss^{t}, \ss^{t+1}, \ldots, \ss^{t+6}$ to get the
desired equations, provided $\vv_{j_1} \neq \vv_{j_2}$, which we
later will argue holds with high probability using
Lemma~\ref{lem:trussCoord}.
First, however, we state and prove Lemma~\ref{lem:trussCoord}.


\begin{figure}[ht]

\begin{center}
\begin{tikzpicture}
\draw (-1,0) -- (1,0);
\draw (-1,0) -- (-1,2);
\draw (-1,0) -- (1,-2);
\draw (-1,0) -- (-1,-2);
\draw (1,2) -- (1,0);
\draw (-1,2) -- (-3,2);
\draw (1,0) -- (-1,2);
\draw (1,0) -- (-3,0);
\draw (1,0) -- (1,-2);
\draw (1,-2) -- (3,-2);
\filldraw[black] (-1,2) circle (2pt) node[anchor=south] {$\ss^{t+6}$};
\filldraw[black] (-1,0) circle (2pt) node[anchor=south east] {$\ss^{t+3}$};
\filldraw[black] (-1,-2) circle (2pt) node[anchor=north] {$\ss^{t+1}$};
\filldraw[black] (3,-2) circle (2pt) node[anchor=north] {$\ss^{j_1}$};
\filldraw[black] (1,2) circle (2pt) node[anchor= south] {$\ss^{t+2}$};
\filldraw[black] (-3,2) circle (2pt) node[anchor=south] {$\ss^{j_2}$};
\filldraw[black] (1,0) circle (2pt) node[anchor=west] {$\ss^{t+5}$};
\filldraw[black] (-3,0) circle (2pt) node[anchor=south] {$\ss^t$};
\filldraw[black] (1,-2) circle (2pt) node[anchor=north] {$\ss^{t+4}$};
\end{tikzpicture}
\end{center}

\caption{Geometric realization of the mutlicommodity
flow gadget generated by as a truss matrix
}

\label{fig:trussGadget}

\end{figure}
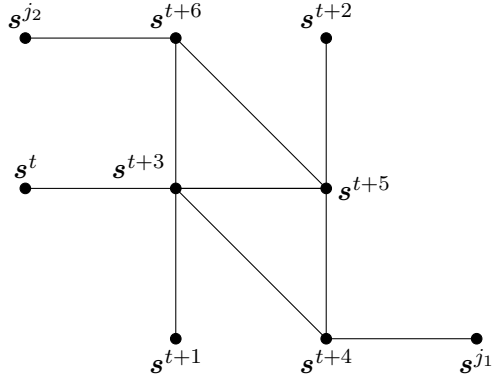

%


\begin{lemma}
Let $\aa \in \mathbb{R}^n$ be a fixed vector such that $-2 \le \aa_i
\le 2, \forall i \in [n]$ and $\aa^\top \one = 0$, and $\aa \neq \zero$.
Let $\yytil$ be a vector picked as above.
Then,
\[
\Pr \left(  \aa^{\trp} \yytil =0 \right) \le \frac{2\delta n^2}{\norm{\aa}_2 R}.
\]
\label{lem:trussCoord}
\end{lemma}
\begin{proof}
Let $\Delta \defeq \yytil - \yy$. Clearly, $-\delta \le \Delta_i \le \delta, \forall i \in [n]$.
\[
\aa^\top \yytil = \aa^\top (\yy + \Delta) \le \aa^\top \yy + \sum_{i \in [n]} \abs{\aa_i} \abs{\Delta_i} 
\le \aa^\top \yy + 2\delta n.
\]
Similarly, $\aa^\top \yytil \ge \aa^\top \yy - 2 \delta n$.
Thus,
\[
\Pr \left( \aa^{\trp} \yytil = 0 \right) \le \Pr \left( \abs{\aa^{\trp} \yy} \le 2 \delta n \right).
\]
Since the distribution of $\yy$ is rotation invariant, we assume without loss of generality $\aa = (\norm{\aa}_2,0,\ldots, 0)$.
Let $A$ be the area of the $n$-dimensional sphere. 
\[
A = \frac{2 \pi^{(n+1)/2}}{\Gamma\left( \frac{n+1}{2} \right)} \cdot R^n.
\]
Let $A_{\delta}$ be the area of $\{ \yy: \norm{\yy}_2 = R,\abs{\yy_1} \le 2\delta n / \norm{\aa}_2 \}$. Then,
\[
\Adelta \le \frac{2 \pi^{n/2}}{\Gamma\left( \frac{n}{2} \right)} \cdot R^{n-1} \cdot \frac{2\delta n}{\norm{\aa}_2}.
\]
Thus, (assume $n$ is even, the case of odd $n$ can be checked similarly)
\begin{align*}
\Pr \left( \abs{\aa^{\trp} \yy} \le \delta \right)
& = \frac{\Adelta}{A} \\
& \le \frac{2\delta n }{\norm{\aa}_2 \sqrt{\pi} R} \cdot \frac{\Gamma(\frac{n+1}{2})}{\Gamma(\frac{n}{2})} \\
& = \frac{2\delta n }{\norm{\aa}_2 R} \cdot \frac{n! \sqrt{\pi}}{ n!! (n-2)!!} \\
& \le \frac{2\delta n^2}{\norm{\aa}_2 R}.
\end{align*}
This completes the proof.
\end{proof}

By our construction of the truss, for each vertex, its
$\vv$-coordinate can be written as a fixed convex combination of
$\yytil$, say $\cc^\top \yytil$ in which $\cc^\top \one = 1$ and
$\cc_i \ge 0, \forall i \in [n]$.
Note that when algorithm $\algGZTtoMCT$ is applied to a matrix $\AA$,
it processes the $j$th row in $k$ iterations where by
Lemma~\ref{lem:GZToMC2NNZ} we have
 $k \leq \log \norm{\AA_j}_1$.
Let $p \defeq \log \norm{\AA_j}_1$.
Given how the convex combination specified by $\cc$ is formed, it
follows that $2^{p} \cc$ is an integer vector.

Next we argue that when two variables are chosen for pairing by the 
$\algGZTtoMCT$ algorithm as it processes some row $\AA_j$, these two
variables will have their $\vv$-coordinates represented as convex
combinations $\cc^\top \yytil$ and $\dd^\top \yytil$ where crucially
$\cc \neq \dd$. This ensures that $2^{p} (\cc-\dd)$ is a non-zero integer
vector,  and hence $\norm{\cc-\dd}_2 \geq 2^{-p} = 1/\norm{\AA_j}_1$.
This follows from stronger claim stated below, which we prove later.

\begin{claim}
\label{clm:trussConvexBits}
  Suppose algorithm $\algGZTtoMCT$ is processing some row
  $\AA_j$. 
  Let $V^{l} $ be the set of variables with non-zero coefficients in
  the main equation $\mathcal{A}_j$
  at the $l$th iteration of the while-loop in Line~\ref{lin:varReplace}
  of $\algGZTtoMCT$.
  Let $S^{l}$ be the set of associated vertices.
  Consider two arbitrary vertices $\ss,\tt\in S^l$, and let $\cc,\dd \in
  \rea^n$ be the non-negative vectors with  $\cc^\top \one = \dd^\top
  \one = 1$ s.t.  the $\vv$-coordinates of $\ss$ and $\tt$ represented as convex
combinations are $\cc^\top \yytil$ and $\dd^\top \yytil$ respectively.
For every $i \in [n]$, we view entries $\cc_i$ and $\dd_i$ as fixed point binary numbers: $\cc_i = \alpha^i_0 . \alpha^i_1 \alpha^i_2
\ldots \alpha^i_p $ and $\dd_i = \beta^i_0 . \beta^i_1  \beta^i_2
\ldots \beta^i_p$, or equivalently $\cc_i = \sum_{k = 0}^p \alpha^i_k
2^{-k} $ and $\dd_i = \sum_{k = 0}^p \beta^i_k
2^{-k} $.
Then there is no index $k$ s.t. $1 = \alpha^i_k = \beta^i_k$,
i.e. there is no index where the $k$th bit is 1 in both strings.
\end{claim}

These two vertices have same $\vv$-coordinate if and only if $(\cc - \dd)^\top \yytil = 0$.
Let $\aa \defeq \cc - \dd$.
Then, $-2 \le \aa_i \le 2, \forall i \in [n]$, 
$\aa^\top \one = 0$, and 
\[
\norm{\aa}_2 \geq 1/\norm{\AA_j}_1.
\]
By Lemma~\ref{lem:trussCoord}, 
\[
\Pr \left( \cc^{\top} \yytil = \dd^{\top} \yytil \right)
\le \frac{2\delta n^2 \norm{\AA_j}_1}{R}.
\]
By Lemma~\ref{lem:GZToMC2NNZ}, the total number of
the vertices in the truss is at most
\[
O\left(n^2 \log n\right).
\]
By a union bound, the probability that there exist two different vertices with same $\vv$-coordinate is at most 
\[
\frac{2\delta n^2 \norm{\AA}_1}{R} \cdot O\left( n^2 \log n \right)^2
= O \left( \frac{\log^2 n}{n^4} \right).
\]

\begin{proof}[Proof of Lemma~\ref{lem:McTwoStrictToTruss}]
Since the linear system for 2D trusses is the same as the linear system for 2-commodity, all complexity parameters of these two linear systems are the same.
\end{proof}

\begin{proof}[Proof of Claim~\ref{clm:trussConvexBits}]
  We prove the claim by induction. 
  Our induction hypothesis is simply that the claim holds in round $l$.

  Note that by Lemma~\ref{lem:GZToMC2NNZ}
  every time a new vertex is created (during the $l$th iteration of
  while-loop in Line~\ref{lin:varReplace}), it is always paired in the
  following iteration (iteration $l+1$) and then disappears (i.e. has zero coefficient)
  in the main equation in all following iterations.
 
  Observe that the convex combination vector $\aa$ for each vertex that corresponds
  to an original variable $i$ is has $\aa_i = 1$ and $\aa_h = 0$ for
  all $h \neq i$. 
  This proves the induction hypothesis for the base case of the
  variables in the main equation $\mathcal{A}_j$ before the first
  iteration of the while-loop (i.e. $l=0$).

  Suppose $\cc$,$\dd$ are the convex combination vectors for two
  variables that exist in some round $l$.
  Assume the induction hypothesis holds for round $l-1$.
  Write  the binary strings for the $i$th of both vectors as $\cc_i = \alpha^i_0 . \alpha^i_1 \alpha^i_2
  \ldots \alpha^i_p $ and $\dd_i = \beta^i_0 . \beta^i_1  \beta^i_2
  \ldots \beta^i_p$, or equivalently $\cc_i = \sum_{k = 0}^p
  \alpha^i_k 2^{-k} $ and $\dd_i = \sum_{k = 0}^p \beta^i_k 2^{-k} $.
  Suppose for the sake of contradiction that there exists some $k$
  s.t. $\alpha^i_k = \beta^i_k = 1$.
  Trivially, it cannot be the case that the such a collision occurs if
  either variable is an original variable.
  So both variables must be new variables.
  The bit string for entry $\cc_i$ is created by averaging two bit
  strings of variables from the main equation in round $l-1$, say 
  $\eta^i_0 . \eta^i_1 \eta^i_2 \ldots \eta^i_p$ and  $\gamma^i_0 . \gamma^i_1 \gamma^i_2
  \ldots \gamma^i_p$.
  Similarly, bit string for entry $\dd_i$ is created by averaging two bit
  strings of variables from the main equation in round $l-1$, say 
  $\theta^i_0 . \theta^i_1 \theta^i_2 \ldots \theta^i_p$ and $\sigma^i_0
  . \sigma^i_1 \sigma^i_2 \ldots \sigma^i_p$.
  Note that each variable can only be paired once in each iteration,
  so the four bit strings must come from distinct variables in round
  $l$.
  
  $\alpha^i_k = 1$ requires that exactly one of the following
  conditions is true:
  \begin{enumerate}
  \item 
\label{enu:caseGamma}$\gamma^i_{k-1} = 1$
  \item
\label{enu:caseDelta}
$\eta^i_{k-1} = 1$
  \item
    \label{enu:caseCarry} A ``carry'' occurred when adding strings
    $\gamma^i_k \gamma^i_{k+1} \ldots \gamma^i_p$ and 
    $\eta^i_k \eta^i_{k+1} \ldots \eta^i_p$.
  \end{enumerate}
But, Case~\ref{enu:caseCarry} immediately leads to a contradiction, as
a carry can only occur when there exists some bit position $g$
s.t. $\gamma^i_g = \eta^i_g$. But this is false by the induction
hypothesis.
Thus we must be in either Case~\ref{enu:caseGamma} or Case~\ref{enu:caseDelta}.
By similar logic, we can conclude from $\beta^i_k = 1$ that exactly
one of the following must be true:  either $\theta^i_{k-1} = 1$ or $\sigma^i_{k-1} = 1$.
All together, we have concluded that exactly two of the four bits
$\eta^i_{k-1},\gamma^i_{k-1},\theta^i_{k-1},$ and
$\sigma^i_{k-1}$ must be set to 1.
This contradicts the induction hypothesis.
Having established a contradiction whenever the induction hypothesis
fails at step $l$, we have shown that it holds at this step.
\end{proof}


\newcommand{\inc}{\NN}
\newcommand{\coef}{\MM}
\newcommand{\cp}{\zz} 
\newcommand{\cchat}{\hat{\cc}}

\section{Connections with Interior Point Methods}
\label{sec:ipm}

In this section, we discuss how applications in 
scientific computing and combinatorial optimization
produce the linear systems that we show are hard
to solve.
We first give a brief overview of interior point methods,
with focus on how they generate linear systems
in Section~\ref{subsec:IPM}.
Then we formalize the matrices that  interior point
methods produce when run on 2-commodity flow
matrices~\ref{sec:IPMMC2} and isotropic total variation
minimization~\ref{subsec:IPMIsotropicTV}.

\subsection{Brief Overview of Interior Point Methods}
\label{subsec:IPM}

Interior point methods~\cite{Wright97,Ye97:book,Nemirovski04,
BoydV04:book,DaitchS08,LeeS14,LeeS15}
can be viewed as ways of solving convex optimization problems via
a sequence of linear systems.
For simplicity, we choose the log-barrier based
interpretation from Chapter 11 of the book by Boyd and
Vandenberghe as our starting point.
The main idea is to represent a convex optimization problem
with constraints as a sequence of
linear programs with terms called barrier functions added to the objective.
Specifically, turning the problem:
\begin{align*}
\min_{\yy} \qquad & f\left( \yy \right)\\
\text{subject to:} \qquad & \coef \yy = \bb\\
& \yy \geq 0
\end{align*}
into the equality-constrained optimization problem:
\begin{align*}
\min_{\yy} \qquad & t \cdot f\left( \yy \right)
	- \sum_{i} \log\left( \yy_i \right)\\
\text{subject to:} \qquad & \coef \yy = \bb
\end{align*}
for a parameter $t$ that is gradually increased (by factors
of about $1+n^{-1/2}$ throughout the course of the algorithm).
The solution of the above optimization problem converges to an optimal solution of the original linear programming, as $t$ goes to infinity.
Between these increase steps, the algorithm performs
Newton steps on this log barrier objective, which
when combined with the equality constraint $\coef \yy = \bb$
requires solving the problem:
\begin{align*}
\min_{\Delta \yy} \qquad & -\gg\left( \yy \right)^{\trp} \Delta \yy\\
\text{subject to:} \qquad & \coef \Delta \yy = {\bf 0}\\
& \norm{\Delta \yy}_{\HH\left( \yy \right)} \leq 0.1
\end{align*}
which is to maximize the projection along the gradient
subject to the second order term being at most $0.1$
and staying in the null space.
This can in turn be interpreted as a least squares problem,
and solving the linear system:
\begin{equation}
\coef\HH(\yy)^{-1} \coef^{\trp} \xx = -t \coef\HH(\yy)^{-1} \gg(\yy).
\label{eqn:IPMNewton}
\end{equation}

%

\subsection{2-Commodity Flow}
\label{sec:IPMMC2}

\newcommand{\yyr}{\yy^{r}}
\newcommand{\yyf}{\yy^{1}}
\newcommand{\yyg}{\yy^{2}}
\newcommand{\ddf}{\dd^{1}}
\newcommand{\ddg}{\dd^{2}}

We now show that solving 2-commodity flow problems using interior
point methods as described in Subsection~\ref{subsec:IPM} can
lead to any system in the class $\mctwostrictCl$.
There are also many variants of the multicommodity flow problem~\cite{Madry10a},
and we work with the minimum cost version due to it being the
most general.
\begin{definition}[Min-cost 2-commodity flow problem]
Given a directed graph $G = (V, E)$ with $n$ vertices and $m$ edges, a positive edge-capacity vector $\cp \in \mathbb{R}^{m}$, a
positive edge-cost vector $\cc \in \mathbb{R}^{2m}$, and two vertex-demand vectors $\ddf, \ddg \in \mathbb{R}^{n}$.
The goal is to compute two flows $\yyf, \yyg$ such that, 
\begin{enumerate}
\item $\yyf$ satisfies the demand $\ddf$, and $\yyg$ satisfies the demand $\ddg$,
\item for each edge $e\in E$, the sum of the two flows on $e$ is no larger than the edge capacity $\cp_e$, and
\item the total cost of the two flows is minimized.
\end{enumerate}
\end{definition}
To formulate this as a linear program, we 
let $\yy = \left( \yyf; \yyg \right)$ be the two flows, and $\inc$ be the edge-vertex incidence matrix of graph $G$. 
\begin{align*}
\min \quad  &  \cc^{\trp} \yy \\
s.t. \quad & \inc^{\trp}\yyf = \ddf \\
& \inc^{\trp}\yyg = \ddg \\
& \yyf,\yyg, \cp - \yyf - \yyg \ge {\bf 0}
\end{align*}
Write this linear programming as the following minimization problem with a logarithmic barrier function,
\begin{align*}
\min \quad  & \cc^{\trp} \yy - \frac{1}{t} \sum_{i\in [m]}  \log \yyf_i + \log \yyg_i + \log(\cp_i - \yyf_i - \yyg_i)  \\
s.t. \quad & \inc^{\trp}\yyf = \ddf \\
& \inc^{\trp}\yyg = \ddg 
\end{align*}
where $t > 0$ is a parameter.
%
\begin{definition}
  We say a linear system $\BB^{\trp} \BB \xx = \BB^{\trp} \cc$ is a
  \emph{Minimum Cost 2-commodity Flow IPM Linear System} if it can be
  obtained from as an instance of the Newton-Step Linear System of
  Equation~\eqref{eqn:IPMNewton} for some Min-cost 2-commodity flow problem.
\end{definition}
The class of Minimum Cost 2-commodity Flow IPM Linear Systems
is appears more restrictive than $\mctwostrictCl$, but as we will see,
it is essentially equivalent, and still sufficiently expressive that
any linear system can be reduced to it.
The main result that we will sketch in this section is:
\begin{lemma}
For any linear system $\AA \xx = \cc$ (with error parameter $\eps$) with polynomially bounded sparse
parameter complexity, there exists an efficient reduction to a Minimum
Cost 2-commodity Flow IPM Linear System.
\label{lem:imp}
\end{lemma}
Our reduction from a general linear system $(\AA, \cc,\eps)$, where $\AA \in \genCl$, to a Minimum Cost 2-commodity Flow IPM Linear System is
the same as our reduction to $\mctwostrictCl$, except that
before our chain of reductions, we first multiply the linear system by
a diagonal matrix $\SS$ with diagonal entries that are $\pm 1$. We
will later specify how to choose these signs.
This gives us the linear system $\SS \AA \xx = \SS \cc$ with the same
error parameter as before. It is easy to
verify that this system has the same sparse parameter complexity as
the original linear system, and an approximate solution to this system
is an approximate solution to the original system with the same
$\eps$.

We now apply our usual chain of reductions to get a linear system in over a
matrix in $\mctwostrictCl$.
We write this system as
\[
(\BBtwostrict)^\trp \BBtwostrict \xx = (\BBtwostrict)^\trp \ccbtwostrict.
\]
The remainder of this Section is dedicated to showing that this linear
system is a Minimum Cost 2-commodity Flow IPM Linear System, thus
proving Lemma~\ref{lem:imp}.

We can pull out the edge weights from $\BBtwostrict$ by writing
$
(\BBtwostrict)^\trp \BBtwostrict = \BBhat^\trp \WW \BBhat,
$
where $\BBhat$ is the unweighted 2-commodity edge-vertex incidence matrix with the same edge structure as $\BBtwostrict$, and $\WW$ is the diagonal matrix of the edge weights.
Then, the linear system in $\mctwostrictCl$ can be written as
\begin{align}
\BBhat^\trp \WW \BBhat\xx = \BBhat^\trp \WW^{1/2} \ccbtwostrict.
\label{eq:ipm_mctwostrict_ls}
\end{align}
Before we prove Lemma~\ref{lem:imp}, we explore some properties of the above linear system.

\begin{claim}
There exist a 2-commodity edge-vertex incidence matrix $\BBtil$, a diagonal matrix $\WWtil$, and a vector $\cctil$ such that
\begin{enumerate}
\item $\BBtil^\trp \WWtil \BBtil = \BBhat^\trp \WW \BBhat$ and $\BBhat^\trp \WW^{1/2} \ccbtwostrict = \BBtil^\trp \WW^{1/2} \cctil$,
\item all the rows of $\BBtil$ are distinct.
\end{enumerate}
\label{clm:ipm_distinct_rows}
\end{claim}
\begin{proof}
If all rows of $\BBhat$ are distinct, then we set $\BBtil = \BBhat, \WWtil = \WW$ and $\cctil = \ccbtwostrict$.
Otherwise, assume the $i$th row and the $j$th row of $\BBhat$ are the same.
Note
\[
\BBhat^\trp \WW \BBhat 
= \sum_{k \neq i, j} \WW_{kk} \BBhat_{k}^\trp \BBhat_{k} + (\WW_{ii} + \WW_{jj}) \BBhat_{i}^\trp \BBhat_{i},
\]
and
\[
\BBhat^\trp \WW^{1/2} \ccbtwostrict 
= \sum_{k \neq i, j} \WW_{kk}^{1/2} \ccbtwostrict_k \BBhat_{k}^\trp + (\WW_{ii}^{1/2} \ccbtwostrict_i + \WW_{jj}^{1/2} \ccbtwostrict_j ) \BBhat_{i}^\trp.
\]
We can construct $\BBtil$ and $\WWtil$ by removing the $j$th row of $\BBhat, \WW$ and set the corresponding edge weight $\WWtil_{ii}$ to be $\WW_{ii}+ \WW_{jj}$. 
We construct $\cctil$ by removing the $j$th entry of $\ccbtwostrict$ and set $\cctil_i$ to be $(\WW_{ii}^{1/2} \ccbtwostrict_i + \WW_{jj}^{1/2} \ccbtwostrict_j ) / (\WW_{ii} + \WW_{jj})^{1/2}$.
We repeat the above process to merge the identical rows until we get the desired matrices and vector.
\end{proof}
\noindent {\bf Remark.}
Let $\xx^*$  be a minimizer of $\norm{\WW^{1/2} \BBhat \xx - \ccbtwostrict}_2$. 
We know that $\xx^*$ is a solution of the normal equations $\BBhat^\trp \WW \BBhat \xx = \BBhat^\trp \WW^{1/2} \ccbtwostrict$. 
By Claim~\ref{clm:ipm_distinct_rows}, $\xx^*$ is a solution of $\BBtil^\trp \WWtil \BBtil^\trp \xx = \BBtil^\trp \WWtil^{1/2} \cctil$, which in turn gives that $\xx^* \in \argmin_{\xx} \norm{\WWtil^{1/2} \BBtil \xx - \cctil}_2$.
Let $\xx$ be an approximate solution such that
\[
\norm{\xx - \xx^*}_{\BBhat^\trp \WW \BBhat} \le \eps \norm{\BBhat^\trp \WW^{1/2} \ccbtwostrict}_{(\BBhat^\trp \WW \BBhat)^{\dagger}}.
\]
By Claim~\ref{clm:ipm_distinct_rows}, the above equation is equivalent to
\[
\norm{\xx - \xx^*}_{\BBtil^\trp \WWtil \BBtil} \le \eps \norm{\BBtil^\trp \WWtil^{1/2} \cctil}_{(\BBtil^\trp \WWtil \BBtil)^{\dagger}}.
\]
It means that $\xx$ is an approximate solution of $\min_{\xx} \norm{\WWtil^{1/2} \BBtil \xx - \cctil}_2$ with $\eps$-accuracy, vice versa.
Thus, the two minimization problems before and after merging identical rows are equivalent.

Without the loss of generality, we assume that all the rows of $\BBhat$ are distinct.
We define \emph{the underlying graph} of $\BBhat$ to be the \emph{simple} graph over $n$ vertices where vertices $i$ and $j$ are connected if and only if $\BBhat$ has at least one type of edges between $i$ and $j$.
According to the constructions in $\algGZTtoMCT$ (Algorithm~\ref{alg:GZ2toMC2}) and $\algMCTtoMCTS$ (Algorithm~\ref{alg:MCTtoMCTS}), 
each edge of the underlying graph has exactly 3 types of edges in $\BBhat$ (that is, type 1, type 2 and type $1+2$), and
the edge weights of $\BBtwostrict$ are either at least 1 or equal to a tiny number $\delta$ assigned in line~\ref{lin:delta} of Algorithm~\ref{alg:MCTtoMCTS}.
We call edge weights which are at least 1 as \emph{large} weights.

\begin{claim}
\label{clm:ipm_single_large_weight_row}
For each edge in the underlying graph of $\BBhat$, exactly one of its type 1, type 2, and type $1+2$ edges has large weight.
\end{claim}
\begin{proof}
Note that the large-weight edges all appear in $\BB \in \mctwoCl$ constructed in Algorithm~\ref{alg:GZ2toMC2}.
An edge in $\BB$ is either a gadget edge in $\mathcal{B}$ or a non-gadget edge in $\mathcal{A}$. 
\begin{enumerate}
\item By Algorithm~\ref{alg:MC2gadget}, all edges in a single $\mctwoCl$-gadget are distinct. 
\item We show that each gadget edge is distinct from all other edges. 
By Algorithm~\ref{alg:MC2gadget}, in a 2-commodity gadget, each ``old" vertex is connected to a ``new" vertex,
and all new vertices except $\xx_t$ are independent of the new vertices created in all other gadgets.
However, in each gadget including $\xx_t$ (as a new vertex or an odd vertex), $\xx_t$ is connected to new vertices which are created in that gadget and do not appear in any other edge outside that gadget.
\item It is possible that two identical non-gadget edges exist when the original linear system instance has redundant or inconsistent constraints.
But note all non-gadget edges are type 1 edges, it implies that two identical edges correspond to two identical rows in $\BBhat$.
By Claim~\ref{clm:ipm_distinct_rows}, we can always merge such identical rows so that all non-gadget edges are distinct.
\end{enumerate}
This completes the proof.
\end{proof}

\begin{proof}[Proof of Lemma~\ref{lem:imp}]
Note the matrices $\BBhat, \WW$ and the vector $\ccbtwostrict$ are given by the reductions from the  general linear system instance $\AA\xx = \cca$.
The goal of proving Lemma~\ref{lem:imp} is to determine the
edge-vertex incidence matrix $\inc$, the cost vector $\cc$, the demand
vectors $\ddf, \ddg$, the edge capacity vector $\cp$, and the flows
$\yyf, \yyg$ such that the intermediate linear systems that arise in
an interior point method~\eqref{eqn:IPMNewton} for solving the
Min-cost 2-commodity Flow Problem is exactly the above linear system~\eqref{eq:ipm_mctwostrict_ls}.

We first make the coefficient matrices of the two linear systems in~\eqref{eqn:IPMNewton} and~\eqref{eq:ipm_mctwostrict_ls} to be the same, by choosing the matrix $\inc$ and the vectors $\yyf,\yyg,\ddf,\ddg,\cp$ properly.
The coefficient matrix of the equality constraints of the 2-commodity linear programming is
\begin{align}
\coef = \left( \begin{array}{cc}
\inc^{\trp} & \\
& \inc^{\trp}
\end{array} \right).
\label{eqn:imp_coeffcient}
\end{align}
We choose $\inc$ to be the edge-vertex incidence matrix of the underlying graph of $\BBhat$. 
By denoting the residue flow amount along an edge as
\[
\yyr \defeq \cp - \yyf - \yyg,
\]
we can rewrite the Hessian as:
\[
\HH = \frac{1}{t}
\sum_{i\in [m]}
\frac{1}{(\yyf_{i})^2} \ee_{2i-1}\ee_{2i-1}^{\trp}
+ \frac{1}{(\yyg_{i})^2} \ee_{2i} \ee_{2i}^{\trp}
+ \frac{1}{(\yyr_{i})^2} \left( \ee_{2i-1} + \ee_{2i} \right)
	\left( \ee_{2i-1} + \ee_{2i} \right)^{\trp},
\]
where $\ee_i \in \mathbb{R}^{2m}$ is the standard basis vector.

We can rearrange the rows and columns of $\HH$ such that, row $2i-1$ and column $2i-1$ correspond to the $\yyf$-flow of the $i$th edge, and row $2i$ and column $2i$ correspond to the $\yyg$-flow of the $i$th edge.
After this rearrangement,
$\HH$ becomes a block diagonal matrix,
where each edge $i$ corresponds to a one $2\times 2$ block given by
\begin{equation}
\HH_{2i-1:2i}
\defeq
\left( \begin{array}{cc}
\frac{1}{(\yyf_i)^2} + \frac{1}{(\yyr_i)^2} & \frac{1}{(\yyr_i)^2} \\
\frac{1}{(\yyr_i)^2} & \frac{1}{(\yyg_i)^2} + \frac{1}{(\yyr_i)^2}
\end{array} \right).
\label{eq:MC2HessianBlock}
\end{equation}
The block diagonal structure of $\HH$ then gives:
\[
\HH^{-1} = \diag\left( \HH_{1:2}^{-1}, \ldots, \HH_{2m-1:2m}^{-1} \right),
\]
and for each edge $i$ we have:
\begin{equation}
\HH_{2i-1:2i}^{-1}
 = \frac{1}{\alpha_i}
\left( \left( \begin{array}{cc}
(\yyf_i)^2 (\yyr_i)^2 & 0 \\
0 & (\yyg_i)^2 (\yyr_i)^2 
\end{array} \right)
+ (\yyf_i)^2 (\yyg_i)^2 \left( \begin{array}{cc}
1 & -1 \\
-1 & 1 
\end{array} \right)  \right),
\label{eq:MC2HessianBlockInv}
\end{equation}
where $\alpha_i = (\yyf_i)^2 + (\yyg_i)^2 + (\yyr_i)^2$. 
It means that 
\begin{enumerate}
\item the type 1 edge has weight $(\yyf_i)^2 (\yyr_i)^2 / \alpha_i$,
\item the type 2 edge has weight $(\yyg_i)^2  (\yyr_i)^2  / \alpha_i$, and
\item the type $1+2$ edge has weight $ (\yyf_i)^2 (\yyg_i)^2 / \alpha_i$.
\end{enumerate}

Note these three types of weights are symmetric.
Let $w_{(i,1)}, w_{(i,2)}, w_{(i,1+2)}$ be the edge weights of type 1, type 2 and type $1+2$ edge of the $i$th edge, respectively.
These edge weights are given by the corresponding diagonals of $\WW$.
We set the values of $\yyf_i, \yyg_i, \yyr_i$ such that
\begin{align*}
(\yyf_i)^2 (\yyr_i)^2 / \alpha_i &= w_{(i,1)}, \\
(\yyg_i)^2 (\yyr_i)^2 / \alpha_i &= w_{(i,2)},\\
(\yyf_i)^2 (\yyg_i)^2 / \alpha_i &= w_{(i,1+2)}.
\end{align*}
Solving the above equations gives:
\begin{equation}
\begin{split}
(\yyf_i)^2 &= \frac{w_{(i,1)}w_{(i,1+2)}}{w_{(i,2)}} + w_{(i,1)} + w_{(i,1+2)}, \\
(\yyg_i)^2 &= \frac{w_{(i,2)} w_{(i,1+2)}}{w_{(i,1)}} + w_{(i,2)} + w_{(i,1+2)}, \\
(\yyr_i)^2 &= \frac{w_{(i,1)} w_{(i,2)}}{w_{(i,1+2)}} + w_{(i,1)} + w_{(i,2)}.
\end{split}
\label{eq:HessianInverse}
\end{equation}
The $\yyf, \yyg$ and $\yyr$ determine the edge-capacity vector $\cp$, and they together with $\inc$ determine the vertex-demand vectors $\ddf$ and $\ddg$.

Since we rearranged the rows and columns of $\HH$, we need do the same rearrangement for $\coef$. 
After rearranging row and columns of $\coef$, $\coef$ is of the form that,
column $2i-1$ and column $2i$ correspond to the $\yyf$-flow and $\yyg$-flow of the $i$th edge, respectively, and row $2i-1$ and row $2i$ correspond to the $\uu$-coordinate and $\vv$-coordinate of vertex $i$, respectively. 
In $\coef$, each edge has a $2\times 2$ identity matrix for one endpoint, and a negative $2\times 2$ identity matrix for the other endpoint.

By the above setting, we can check that 
\begin{align}
\coef\HH^{-1}\coef^\trp = \BBhat^\trp \WW \BBhat.
\label{eqn:ipm_lhs}
\end{align}
Note that $\coef$ has dimension $2n \times 2m$, and $\BBhat$ has dimension $3m \times 2n$.

Secondly, we make the right hand side vectors of the two linear systems~\eqref{eqn:IPMNewton} and~\eqref{eq:ipm_mctwostrict_ls} to be the same, that is,
\begin{align}
-t \coef\HH^{-1} \gg = \BBhat^\trp \WW^{1/2} \ccbtwostrict,
\label{eqn:ipm_rhs}
\end{align}
by choosing the cost vector $\cc$ properly.
The gradient $\gg$ is given by:
\[
\gg = \cc - \frac{1}{t} \sum_{i\in [m]} \frac{1}{\yyf_{i}} \ee_{2i-1} 
+ \frac{1}{\yyg_{i}} \ee_{2i}
- \frac{1}{\yyr_{i}} \left( \ee_{2i-1}  + \ee_{2i} \right).
\]
For simplicity, let 
\[
\ff(\yy) \defeq t \left( \cc - \gg \right).
\]
$\ff(\yy)$ is a vector in $2m$-dimension such that, the $i$th edge corresponds to the $2 \times 1$ sub-vector
\[
\left( \begin{array}{c}
\frac{1}{\yyf_i} - \frac{1}{\yyr_i} \\
\frac{1}{\yyg_i} - \frac{1}{\yyr_i}
\end{array} \right).
\]
Then we have 
\[
-t\coef\HH^{-1}\gg = \coef\HH^{-1}  \left( \ff - t\cc \right).
\]

Recall the right hand side vector of Equation~\eqref{eqn:ipm_rhs} is $\BBhat^\trp \WW^{1/2} \ccbtwostrict$.
According to our construction in $\algGZTtoMCT$ (Algorithm~\ref{alg:GZ2toMC2}) and $\algMCTtoMCTS$ (Algorithm~\ref{alg:MCTtoMCTS}), 
$\ccbtwostrict = (\cca; {\bf 0})$, where $\cca$ is the right hand side vector of the general linear system instance.
Note all the nonzero entries of $\ccbtwostrict$ correspond to the type 1 edges in the main constraint set $\mathcal{A}$, see line~\ref{lin:mainConstraint} of Algorithm~\ref{alg:GZ2toMC2}.

Recall the structure of $\coef$ in Equation~\eqref{eqn:imp_coeffcient}, the edge-vertex incidence matrix $\inc$ contains all the $m$ edges of the underlying graphs, which are exactly the set of edges in $\BB$ constructed in $\algGZTtoMCT$ (Algorithm~\ref{alg:GZ2toMC2}).
Note for each edge of $\inc$, its large-weight edge can be any of the corresponding type 1, type 2 and type $1+2$ edges. 
By Claim~\ref{clm:ipm_single_large_weight_row}, we can partition the underlying graph edges into $E_1 \cup E_2 \cup E_{1+2}$ such that $E_k$ contains all the underlying edges whose large-weight edges are the corresponding type $k$ edges, $k \in \{1,2,1+2\}$.
Note $\ccbtwostrict_{(i,1)} \neq 0$ only if edge $i$ is in $E_1$.

We define a vector $\cc' \in \mathbb{R}^{2m}$ such that the entry of $\cc'$ which corresponds to the $\yyf$-flow 
of the $i$th edge 
equals to $w_{(i,1)}^{1/2} \ccbtwostrict_{(i,1)}$. 
This definition gives:
\begin{align}
\BBhat^\trp \WW^{1/2} \ccbtwostrict
= \coef \cc'.
\label{eqn:ipm_def_c'}
\end{align}
By Equation~\eqref{eqn:ipm_rhs}, 
\[
\coef\HH^{-1} \left( \ff - t\cc \right)
= \coef\cc'.
\]
Rearranging it,
\[
\coef \HH^{-1} \left( \ff - t\cc - \HH\cc' \right) = {\bf 0}.
\]
According to $\algMCToGZGadget$ (Algorithm~\ref{alg:MC2gadget}), each gadget has a cycle containing both the 2 edges in $E_{1+2}$: $(t+3, t+4, t+2, t+5, t+6, t+1, t+3)$.
Let $E'$ be the set of all the edges in such gadget cycles.
Let $\pp \in \{0,1\}^m$ such that each entry corresponding to an edge in $E'$ has value 1.
Let $\qq \defeq (\pp; \pp)$.
According to $\algGZTtoMCT$ (Algorithm~\ref{alg:GZ2toMC2}), every gadget edge $i$ has $\cc'_i = 0$.
Thus, we are free to choose the sign of a column of $\inc^\top$ which corresponds to a gadget edge, which does not change Equations~\eqref{eqn:ipm_lhs} and~\eqref{eqn:ipm_def_c'}.
\footnote{Note the demand vectors $\ddf$ and $\ddg$ change after we change $\inc$.}
Without loss of generality, we assume that after proper sign flipping, 
$\MM \qq = {\bf 0}$.
It gives that
\[
\coef \HH^{-1} \left( \ff - t\cc - \HH\cc'  + \HH\qq\right) = {\bf 0}.
\]
We set $t\cc = \ff - \HH\cc' + \lambda\HH\qq$, where $\lambda > 0$ is a sufficiently large constant to be determined later.


For each edge $i$, let $\mathcal{I}_i$ be an indicator whose value is 1 if edge $i$ in the set $E'$ and 0 otherwise.
For the $i$th edge, 
its corresponding $2 \times 1$ sub-vector in $t\cc$ is
\[
t \left( \begin{array}{c}
\cc_i \\
\cc_{m+i}
\end{array} \right) =
\left( \begin{array}{c}
\frac{1}{\yyf_i} - \frac{1}{\yyr_i} \\
\frac{1}{\yyg_i} - \frac{1}{\yyr_i}
\end{array} \right)
-
\cc'_i \left( \begin{array}{c}
\frac{1}{(\yyf_i)^2} + \frac{1}{(\yyr_i)^2} \\
\frac{1}{(\yyr_i)^2} 
\end{array} \right)
+ \mathcal{I}_i \lambda \ss^i,
\]
where
\[
\ss^i \defeq
\left( \begin{array}{c}
\frac{1}{(\yyf_i)^2} + \frac{2}{(\yyr_i)^2} \\
\frac{1}{(\yyg_i)^2} + \frac{2}{(\yyr_i)^2}
\end{array} \right).
\]
Note that $\ss^i \ge {\bf 0}$ (entry-wise).
Plugging the solution of $\yy$ from Equation~\eqref{eq:HessianInverse} gives:
\[
t \left( \begin{array}{c}
\cc_i \\
\cc_{m+i}
\end{array} \right) =
\left( \begin{array}{c}
\sqrt{\frac{w_{(i,2)}}{\gamma}} - \sqrt{\frac{w_{(i,1+2)}}{\gamma}} \\
\sqrt{\frac{w_{(i,1)}}{\gamma}} - \sqrt{\frac{w_{(i,1+2)}}{\gamma}}
\end{array} \right)
-
\cc'_i \left( \begin{array}{c}
\frac{w_{(i,2)} + w_{(i,1+2)}}{\gamma} \\
\frac{w_{(i,1+2)}}{\gamma}
\end{array} \right)
+ \mathcal{I}_i \lambda \ss^i,
\]
where $\gamma \defeq w_{(i,1)}w_{(i,1+2)} + w_{(i,1)}w_{(i,2)} + w_{(i,2)}w_{(i,1+2)}$.

By Claim~\ref{clm:ipm_single_large_weight_row}, each edge of the underlying graph has exactly one type of edges with large weight.
We deal with each of these cases separately under the condition that
two of the edge's weights are much smaller than that of the third one.
We will also denote this ratio using $\eps > 0$:
\begin{enumerate}
\item If $w_{(i,2)} = w_{(i,1+2)} = \eps w_{(i,1)}$, then
\[
t \left( \begin{array}{c}
\cc_i \\
\cc_{m+i}
\end{array} \right) 
\ge - \cc'_i \left(\begin{array}{c}
\frac{2\eps}{\mu} \\
\frac{\eps}{ \mu}
\end{array} \right)
+ \left( \begin{array}{c}
0 \\
\frac{1 - \sqrt{\eps}}{\sqrt{\mu}}
\end{array} \right),
\]
where $\mu = (2\eps + \eps^2) w_{(i,1)}$.
This corresponds to an edge of $E_1$.
If $\cc'_i \le 0$, then $t \cc_i, t\cc_{m+i} \ge 0$. 
If $\cc'_i > 0$, then $\cc'_i = w_{(i,1)}^{1/2}\ccbtwostrict_{(i,1)}$ where $\ccbtwostrict_{(i,1)}$ comes from an entry of $\cca$.
We multiply -1 on both sides of the corresponding equation in the general linear system instance $\AA\xx = \cca$ which creates this edge $i$,
and will get $\cc'_i < 0$ instead. 
\footnote{We first flip signs of equations of $\AA\xx = \cca$ so that $\cca \le {\bf 0}$ (entry-wise), which determines $\inc, \coef$ and $\cc'$. We then flip signs of columns of $\inc$, which we guarantee does not violate Equation~\eqref{eqn:ipm_lhs} and~\eqref{eqn:ipm_def_c'}.}
\item If $w_{(i,1)} = w_{(i,1+2)} = \eps w_{(i,2)}$, then
\[
t \left( \begin{array}{c}
\cc_i \\
\cc_{m+i}
\end{array} \right) 
\ge - \cc'_i \left( \begin{array}{c}
\frac{1+\eps}{\mu} \\
\frac{\eps}{\mu}
\end{array} \right)
+ \left( \begin{array}{c}
\frac{1 - \sqrt{\eps}}{\sqrt{\mu}} \\
0
\end{array} \right),
\]
where $\mu = (2\eps + \eps^2) w_{(i,2)}$.
This corresponds to an edge of $E_2$, in which $\cc'_i = 0$.
We can check that $t \cc_i, t\cc_{m+i} \ge 0$.
\item If $w_{(i,1)} = w_{(i,2)} = \eps w_{(i,1+2)}$, then
\[
t \left( \begin{array}{c}
\cc_i \\
\cc_{m+i}
\end{array} \right) 
= - \cc'_i \left( \begin{array}{c}
\frac{1+ \eps}{ \mu} \\
\frac{1}{ \mu}
\end{array}
\right)
+ \left( \begin{array}{c}
\frac{\sqrt{\eps} - 1}{ \sqrt{\mu} } \\
\frac{\sqrt{\eps} - 1}{ \sqrt{\mu} }
\end{array} \right)
 + \mathcal{I}_i \lambda \ss^i,
\]
where $\mu = (2 \eps + \eps^2) w_{(i,1+2)}$.
This corresponds to an edge of $E_{1+2}$, in which $\cc'_i = 0$ and $\mathcal{I}_i = 1$.
We choose $\lambda$ such that, for every edge $i$ in this case, $t\cc_i, t\cc_{m+i} \ge 0$.

\end{enumerate}

\end{proof}
Note in IPMs, the values of $t$ is a sequence of increasing values. Whenever $t \ge m / \epsilon'$, the optimality gap is smaller than $\epsilon'$.
This ensures that the cost vector $\cc$'s entries are in a
reasonable range.

\subsection{Isotropic Total Variation Minimization}
\label{subsec:IPMIsotropicTV}

For the isotropic total variation problem, we follow
the formulations given in~\cite{ChinMMP13}, namely
given a graph $G = (V, E, \ww)$, we partition the
sets into $S_1 \ldots S_k$, and minimize the objective
\[
\norm{\yy - \ss}_2^2 + 
\sum_{1 \leq i \leq k}
\ww_{i} \sqrt{\sum_{(u,v) \in S_i} \ww_{uv} (\yy_u - \yy_v)^2}.
\]
where $\ss$ is the input signal vector.
Let $\inc$ be the edge-vertex incidence matrix of the graph $G$,
then the dual of the grouped flow problem is:
\begin{align*}
\max \qquad & \ss^{\trp} \inc^{\trp} \ff\\
\text{subject to:} \qquad & \sum_{e \in S_i} \ww_e^{-1} \ff_e^2 \leq \ww_i^{-1}
	\qquad \forall 1 \leq i \leq k
\end{align*}
Both the primal and dual problems can be formulated into log
barriers via second order cones~\cite{GoldfarbY04}.
In the case with minimizing vertex labels,
we introduce a variable $\yy_{i}$
for each cluster $i$, and instead minimize
$\sum_{i} \ww_{i} \yy_i$ subject to the constraint
\[
\yy_{i}^2 \geq \sum_{(u,v) \in S_i} \ww_e^{-1} (\yy_u - \yy_v)^2,
\]
which in turn leads to the logarithmic barrier function
\[
\log\left(\yy_i^2 -  \sum_{(u,v) \in S_i} \ww_e^{-1} (\yy_u - \yy_v)^2 \right),
\]
while in the flow case the barrier functions for each set of edges $S_i$ is
\[
\log\left( \ww_i^{-1} - \sum_{e \in S_i} \ww_e^{-1} \ff_e^2\right).
\]

Due to the connections with the multicommodity flow problems,
we will only state the connections through the flow version here.
Recall that linear systems in $\mctwoCl$ have the form:
\begin{align*}
\LL^{1} \otimes \left( \begin{array}{cc}
1 & 0 \\
0 & 0
\end{array} \right)
+ \LL^{2} \otimes \left( \begin{array}{cc}
0 & 0 \\
0 & 1
\end{array} \right)
+ \LL^{1+2} \otimes \left( \begin{array}{cc}
1 & -1 \\
-1 & 1
\end{array} \right),
\end{align*}
where $E_1, E_2, E_{1+2}$ denotes the edges in
$\LL^1, \LL^2, \LL^{1+2}$, respectively.
Edges in $E_1 \cup E_2$ correspond to individual edges,
so it remains to show that edges in $E_{1 + 2}$ can
correspond to the Hessian matrices of clusters
containing pairs of edges.
In a cluster with two edges $\ff_1$ and $\ff_2$, the
gradient of the function
\[
\log\left( \alpha - \ww_1^{-1} \ff_1^2 - \ww_2^{-1} \ff_2^2\right)
\]
is
\[
\frac{2}{r}
\left(
\begin{array}{c}
-\ww_1^{-1} \ff_1\\
-\ww_2^{-1} \ff_2
\end{array}
\right),
\]
where we treat
\[
r \defeq \alpha - \ww_1^{-1} \ff_1^2 - \ww_2^{-1} \ff_2^2
\]
as a new free variable because we are free to choose
$\alpha$ in the IPM instance.
Differentiating again (via the product rule) then gives the Hessian matrix:
\[
\frac{-4}{r^2}
\left(
\begin{array}{cc}
\ww_1^{-2} \ff_1^2 & \ww_1^{-1} \ww_2^{-2} \ff_1 \ff_2\\
\ww_1^{-1} \ww_2^{-2} \ff_1 \ff_2 & \ww_2^{-2} \ff_2^2
\end{array}
\right)
+ \frac{-2}{r} 
\left(
\begin{array}{cc}
\ww_1^{-1} & 0\\
0 & \ww_2^{-1}
\end{array}
\right).
\]
As we are free to choose $\ww_1$ and $\ww_2$, the second
matrix can be any diagonal matrix.
Then the freedom in choosing $\ff_1$ and $\ff_2$ then means
the first matrix can be any rank $1$ object.
Therefore, this Hessian is equivalent to the block Hessian
for two commodity flows generated in 
Equation~\eqref{eq:MC2HessianBlock} (up to a change of sign,
since we are dealing with a maximization problem here), and its inverse as given
in Equation~\eqref{eq:MC2HessianBlockInv} provides the characterization
from Definition~\ref{def:TV}.
We can also check more directly that the $1 + 2$ edge can be
represented as one of these Hessian inverse blocks.

\begin{claim}
For each edge $(i,j)$ in $E_{1+2}$, there exist an edge-vertex incidence matrix $\inc$, a diagonal matrix $\WW$ and a vector $\rr$ such that 
$\WW \pgeq \rr\rr^\trp$ and
\[
\LL^{1+2}_{ij} \otimes \left( \begin{array}{cc}
1 & -1 \\
-1 & 1
\end{array} \right)
= \inc^{\trp} \left( \WW - \rr\rr^{\trp} \right) \inc.
\]
\end{claim}

\begin{proof}
For simplicity, we remove all zero rows and columns of $\LL_{ij}^{1+2}$ so that the 2-commodity matrix for edge $(i,j)$ only has dimension $4 \times 4$.
We pick 
\[
\inc = \left( \begin{array}{cc}
1 & 0 \\
-1 & 0 \\
0 & 1 \\
0 & -1
\end{array} \right), 
\WW = \left( \begin{array}{cc}
2 & 0 \\
0 & 2
\end{array} \right)
\mbox{ and }
\rr = \left( \begin{array}{c}
1 \\
1
\end{array} \right).
\]
We check the PSD condition,
\[
\WW - \rr\rr^{\trp}
= \left( \begin{array}{cc}
1 & -1 \\
-1 & 1
\end{array} \right) \pgeq {\bf 0}.
\]
Then, we check that this decomposition equals to the 2-commodity matrix
\begin{multline*}
\inc^{\trp} \left( \WW - \rr\rr^{\trp} \right) \inc
 = \left( \begin{array}{cc}
1 & 0 \\
-1 & 0 \\
0 & 1 \\
0 & -1
\end{array} \right)
\left( \begin{array}{cc}
1 & -1 \\
-1 & 1
\end{array} \right)
\left( \begin{array}{cccc}
1 & -1 & 0 & 0\\
0 & 0 & 1 & -1
\end{array} \right) \\
 = \left( \begin{array}{c}
1 \\ 
-1 \\
-1 \\
1
\end{array} \right)
\left( \begin{array}{cccc}
1 & -1 & -1 & 1
\end{array} \right)
 = \left( \begin{array}{cccc}
1 & -1 & -1 & 1 \\ 
-1 & 1 & 1 & -1 \\
1 & -1 & -1 & 1 \\ 
-1 & 1 & 1 & -1 
\end{array} \right)
= \LL^{1+2}_{ij} \otimes \left( \begin{array}{cc}
1 & -1 \\
-1 & 1
\end{array} \right).
\end{multline*}
\end{proof}

\begin{proof}[Proof of Lemma~\ref{lem:McTwoStrictToTV}]
Since the linear system related to the isotropic total variation minimization matrix is the same as the linear system for 2-commodity, all complexity parameters of these two linear systems are the same.
\end{proof}


\section{A Natural Decision Problem: Vector in the Image of a Matrix}
\label{sec:lsd}

When we are given a linear system $\AA \xx = \cc$, a natural question is
whether the system in fact has a solution.
This is true whenever $\cc$ lies in the image of $\AA$, or
equivalently, when $\Pi_{\AA} \cc = \cc$.

We consider a slightly relaxed version of the question of whether
$\cc$ lies in the image of $\AA$.
We define the Linear System Decision Problem, abbreviated \LSD, as follows:


\begin{definition}[Linear System Decision Problem, \LSD]
  Given a linear system $(\AA,\cc)$ and an approximation parameter
  $\eps > 0$, 
we refer to the $\LSD$ problem for the triple
$(\AA,\cc,\eps)$ as the problem of outputting
  \begin{enumerate}
  \item 
    \label{cas:GLDYES}
    $\YES$ if there exists $\xx$ s.t. $\AA \xx = \cc$.
  \item 
    \label{cas:GLDNO}
    $\NO$ if for all  $\xx$, we have $\norm{\AA \xx - \cc} > \eps
    \norm{\cc}$.
  \item Any output is acceptable if neither of the above two cases hold.
  \end{enumerate}
\end{definition}



\begin{algorithm}[htb]
\renewcommand{\algorithmicrequire}{\textbf{Input:}}
\renewcommand\algorithmicensure {\textbf{Output:}}
    \caption{\label{alg:GLDtoGLA}
      Reduce $\LSD$ to $\LSA$.}
    \begin{algorithmic}[1]
      \STATE Find $\xx$ which is a solution to $\LSA$ instance $(\AA,\cc,\eps)$.
      \IF {$\norm{\AA \xx - \cc} \leq \eps \norm{\cc}$ }
          \RETURN \YES.
      \ELSE
         \RETURN \NO.
      \ENDIF
\end{algorithmic}
\end{algorithm}

\begin{lemma}
Algorithm~\ref{alg:GLDtoGLA} solves $\LSD$ instance $(\AA,\cc,\eps)$.
\end{lemma}

\begin{proof}
  First, suppose we are in Case~\ref{cas:GLDYES}, i.e. there exists
  $\xx$ s.t. $\AA \xx = \cc$,
  or equivalently $\PPi_{\AA} \cc = \cc$.
  From this and $\xx$ being an {$\eps$-\asoln} to $(\AA,\cc)$, it
  follows that
  \begin{align*}
\norm{\AA \xx - \cc}
 = 
\norm{\AA \xx - \PPi_{\AA} \cc} 
\leq 
\eps \norm{\PPi_{\AA} \cc} 
=
\eps \norm{\cc}
.
  \end{align*}
So, the algorithm will return $\YES$.

Second, suppose we are in Case~\ref{cas:GLDNO}, then for all $\xx$
we have $\norm{\AA \xx - \cc} > \eps \norm{\cc}$,
so the algorithm must return $\NO$.
\end{proof}


\subsection{Why Not Solve the Exact Decision Problem?}
\label{sec:expSmallNullProj}

Earlier in this section, we showed how to solve an approximate linear
system decision problem, $\LSD$ by a reduction to an $\LSA$ problem.
However, given matrix $\AA$ and a vector $\cc$ with polynomially
bounded integer entries and condition number, it is natural to ask whether
solving $\LSA$ to reasonably high accuracy could in fact let us solve
exactly the decision problem of whether there exists an $\xx$
s.t. $\AA \xx = \cc$.
Note that this is equivalent to $\left( \II - \PPi_{\AA} \right) \cc =
\zero$.
Lemma~\ref{lem:ProjLengthLower} implies that when $\AA$ and $\cc$ have
integer valued entries, either $\PPi_{\AA} \cc = \zero$ or
$\norm{\PPi_{\AA} \cc} \geq 1/\sigma_{\max}(\AA)$.
If a similar separation result was true for
$\norm{\left( \II - \PPi_{\AA} \right) \cc}$, 
we could use this to solve the exact
decision problem.
However, as we will show below, there exists a matrix and vector pair
$\AA \in \mathbb{Z}^{(n+1) \times n}$ and a vector $\cc \in
\mathbb{Z}^{n+1}$ whose entries are polynomially
bounded (in fact bounded between -1 and 2), and with polynomially
bounded $\sigma_{\min}(\AA)$ and $\sigma_{\max}(\AA)$ s.t. 
$\norm{\left( \II - \PPi_{\AA} \right) \cc}$ is exponentially small in
$n$.
This suggests that solving the exact decision problem of whether
$\left( \II - \PPi_{\AA} \right) \cc = \zero$ would require solving
$\LSA$ to exponentially high precision, or adopting some radically
different approach. 

We now outline the construction of $\AA$ and $\cc$.
Each column $i$ of $\AA$ has two nonzero entries: $\AA_{i,i} = 2$ and $\AA_{i,i+1} = -1$.
The rank of $\AA$ is $n$.
Let $\rr$ be the only unit vector of $\nulls(\AA^{\trp})$.
Then, 
\[
\II - \PPi_{\AA} = \rr \rr^{\trp}.
\]
We compute $\rr$. Since $\AA^{\trp} \rr = {\bf 0}$, we have
\[
\rr_{i+1} = 2\rr_{i}, \forall 1 \le i \le n.
\]
Since $\norm{\rr}_2 = 1$, we have $\rr_1 = \frac{1}{2^{n+2}- 1}$.
Let $\cc$ be a vector whose first entry is 1 and all other entries are 0.
Then,
\[
\norm{\left( \II - \PPi_{\AA} \right) \cc}_2
= \abs{\rr^{\trp} \cc}_2
= \frac{1}{2^{n+2} - 1},
\]
which is exponentially small.

We show that the nonzero singular values of $\AA$ are bounded by constants.
Note that $\AA^{\trp} \AA$ is a tridiagonal matrix such that, 
\begin{enumerate}
\item all entries on the main diagonal are 5, 
\item all entries on the first diagonal below the main diagonal are -2, and 
\item all entries on the first diagonal below the main diagonal are -2.
\end{enumerate}
$\AA^{\trp}\AA$ has full rank.
$\AA^{\trp} \AA$ can be written as $\LL + \DD$, where $\LL$ is the graph Laplacian matrix for an undirected path with identical edge weight 2, and $\DD$ is a diagonal matrix whose first and last diagonal entries are 3 and all other diagonal entries are 1.
Thus,
\[
\LL + \II \pleq \AA^{\trp} \AA \pleq \LL + 3 \II.
\]
Since $\lambda_{\max}(\LL) = 4$, we have for all eigenvalues of
$\lambda(\AA^{\trp} \AA)$ that 
\[
1 \le \lambda\left( \AA^{\trp} \AA \right) \le 11,
\]
that is,
\[
1 \le \sigma_{\min} (\AA) \le \sigma_{\max} (\AA)  \le \sqrt{11}.
\]



\FloatBarrier

\section*{Acknowledgements}
We thank Richard Peng and Daniel Spielman for helpful comments and
discussions and we thank Richard Peng for suggesting the example in
Section~\ref{sec:expSmallNullProj} which demonstrates that the
projection of a vector onto the null space of a matrix can be
exponentially small despite well-conditionedness assumptions on the
matrix-vector pair.

\appendix 

\section{Linear Algebra Background}
\label{sec:linalg}

\begin{proof}[Proof of Lemma~\ref{fac:voltageErrorEquiv}]
First observe
  \begin{align*}
\norm{\AA^{\trp} \AA \xx - \AA^{\trp} \cc}_{(\AA^{\trp}\AA)^{\pinv}}^{2}
&  = 
(\AA^{\trp} \AA \xx - \AA^{\trp} \cc)^{\trp} (\AA^{\trp}\AA)^{\pinv} (\AA^{\trp} \AA \xx - \AA^{\trp} \cc)
\\
&  = 
(\AA \xx - \cc)^{\trp} \proj_{\AA} (\AA \xx - \cc)
\\
&  = 
(\proj_{\AA}\AA \xx - \proj_{\AA}\cc)^{\trp} (\proj_{\AA}\AA \xx - \proj_{\AA}\cc)
\\
&  = 
\norm{ \AA \xx - \proj_{\AA} \cc }^{2}_2 
  .
  \end{align*}
Taking $\xx = \zero$ it follows that  
$\norm{ \proj_{\AA} \cc} =  
  \norm{\AA^{\trp} \cc}_{(\AA^{\trp} \AA)^{\pinv}}$, and combining our
  observations immediately gives the second part of the Fact.

A similar argument gives claims for $\norm{ \xx - \xx^* }_{\AA^{\trp}\AA}$.
\begin{align*}
\norm{ \xx - \xx^* }_{\AA^{\trp}\AA}^2
& = (\xx - \xx^*)^{\trp} \AA^{\trp} \AA (\xx - \xx^*) \\
& = \xx^{\trp} \AA^{\trp} \AA \xx - 2\xx^{\trp} \AA^{\trp} \AA \xx^*
+ (\xx^*)^{\trp} \AA^{\trp} \AA \xx^* \\
& = \xx^{\trp} \AA^{\trp} \AA \xx - 2\xx^{\trp} \AA^{\trp} \proj_{\AA} \cc
+ \cc^{\trp} \proj_{\AA}^{\trp} \proj_{\AA} \cc \\
& = \norm{ \AA \xx - \proj_{\AA} \cc }^{2}_2 .
\end{align*}
The third equality uses the fact that $\AA\xx^* = \proj_{\AA} \cc$.
This completes the proof.
\end{proof}
\section{Using Complexity Parameters}

\label{sec:parameters}

In this section, we bound the norms, eigenvalues of matrix $\AA \in \mathbb{R}^{n \times m}$ and the norms of vector $\cc \in \mathbb{R}^n$, using the complexity parameters defined in Definition~\ref{def:complexity_parameters}:
%
\begin{enumerate}
\item $s: \nnz(\AA)$, 
\item $U: \max \left( \norm{\AA}_{\max}, \norm{\cc}_{\max}, \frac{1}{\anzmin(\AA)}, \frac{1}{\anzmin(\cc)} \right)$,
\item $K: \kappa(\AA)$.
\end{enumerate}

Note $n, m \le s$.

\begin{claim}
\label{clm:para_infty_norm}
$\norm{\AA}_{\infty} \le sU$.
\end{claim}

\begin{claim}
\label{clm:para_l2_norm}
$\norm{\cc}_2 \le \sqrt{s} U$.
\end{claim}
\begin{proof}
This follows the relations of norms,
\[
\norm{\cc}_2 \le \sqrt{n} \norm{\cc}_{\infty}
\le \sqrt{s} U.
\]
\end{proof}

\begin{claim}
\label{clm:para_eig}
$\frac{1}{sU^2} \le \lambda_{\max}(\AA^{\trp} \AA) \le sU^2$ and $\lambda_{\min}(\AA^{\trp} \AA) \ge \frac{1}{sK^2U^2}$.
\end{claim}
\begin{proof}
By the definition of Frobenius norm, 
\[
\norm{\AA}_F^2 = \sum_{i,j} \AA_{ij}^2
= \sum_{i = 1}^{\min\{m,n\}} \lambda_i (\AA^{\trp} \AA).
\]
Thus,
\[
\lambda_{\max} (\AA^{\trp}\AA) \le \norm{\AA}_F^2
\le sU^2,
\]
and 
\[
\lambda_{\max} (\AA^{\trp} \AA) \ge \frac{\norm{\AA}_F^2}{\min\{m,n\}}
\ge \frac{1}{sU^2}.
\]

By the definition of $\kappa(\AA^{\trp} \AA)$,
\[
\lambda_{\min} (\AA^{\trp} \AA) = \frac{\lambda_{\max} (\AA^{\trp} \AA)}{\kappa (\AA^{\trp} \AA)}
\ge \frac{1}{sK^2U^2}.
\]
This completes the proof.
\end{proof}


\bibliographystyle{alpha}
\bibliography{refs}

\newcommand{\etalchar}[1]{$^{#1}$}
\begin{thebibliography}{ANKKS{\etalchar{+}}12}

\bibitem[ABFM14]{alexeev2014phase}
Boris Alexeev, Afonso~S Bandeira, Matthew Fickus, and Dustin~G Mixon.
\newblock Phase retrieval with polarization.
\newblock {\em SIAM Journal on Imaging Sciences}, 7(1):35--66, 2014.

\bibitem[ANKKS{\etalchar{+}}12]{arie2012global}
Mica Arie-Nachimson, Shahar~Z Kovalsky, Ira Kemelmacher-Shlizerman, Amit
  Singer, and Ronen Basri.
\newblock Global motion estimation from point matches.
\newblock In {\em 3D Imaging, Modeling, Processing, Visualization and
  Transmission (3DIMPVT), 2012 Second International Conference on}, pages
  81--88. IEEE, 2012.

\bibitem[BCPT05]{BomanCPT05}
Erik~G Boman, Doron Chen, Ojas Parekh, and Sivan Toledo.
\newblock On factor width and symmetric h-matrices.
\newblock {\em Linear algebra and its applications}, 405:239--248, 2005.

\bibitem[BV04]{BoydV04:book}
S.~Boyd and L.~Vandenberghe.
\newblock {\em Convex Optimization}.
\newblock Camebridge University Press, 2004.
\newblock Available at https://web.stanford.edu/\textasciitilde
  boyd/cvxbook/bv\_cvxbook.pdf.

\bibitem[CFM{\etalchar{+}}14]{CohenFMNPW14}
Michael~B. Cohen, Brittany~Terese Fasy, Gary~L. Miller, Amir Nayyeri, Richard
  Peng, and Noel Walkington.
\newblock Solving 1-laplacians in nearly linear time: Collapsing and expanding
  a topological ball.
\newblock In {\em Proceedings of the Twenty-Fifth Annual {ACM-SIAM} Symposium
  on Discrete Algorithms, {SODA} 2014, Portland, Oregon, USA, January 5-7,
  2014}, pages 204--216, 2014.

\bibitem[CKM{\etalchar{+}}14]{CKMPPRX14}
Michael~B. Cohen, Rasmus Kyng, Gary~L. Miller, Jakub~W. Pachocki, Richard Peng,
  Anup~B. Rao, and Shen~Chen Xu.
\newblock Solving sdd linear systems in nearly mlog1/2n time.
\newblock In {\em Proceedings of the Forty-sixth Annual ACM Symposium on Theory
  of Computing}, STOC '14, pages 343--352, New York, NY, USA, 2014. ACM.

\bibitem[CLM{\etalchar{+}}16]{CohenLMPS16}
Michael~B Cohen, Yin~Tat Lee, Gary Miller, Jakub Pachocki, and Aaron Sidford.
\newblock Geometric median in nearly linear time.
\newblock In {\em Proceedings of the 48th Annual ACM SIGACT Symposium on Theory
  of Computing}, pages 9--21. ACM, 2016.
\newblock Available at: https://arxiv.org/abs/1606.05225.

\bibitem[CMMP13a]{ChinMMP13}
Hui~Han Chin, Aleksander M{{a}}dry, Gary~L. Miller, and Richard Peng.
\newblock Runtime guarantees for regression problems.
\newblock In {\em Proceedings of the 4th conference on Innovations in
  Theoretical Computer Science}, ITCS '13, pages 269--282, New York, NY, USA,
  2013. ACM.

\bibitem[CMMP13b]{CMMP13}
Hui~Han Chin, Aleksander Madry, Gary~L Miller, and Richard Peng.
\newblock Runtime guarantees for regression problems.
\newblock In {\em Proceedings of the 4th conference on Innovations in
  Theoretical Computer Science}, pages 269--282. ACM, 2013.

\bibitem[CMSV17]{CMSA17}
Michael~B. Cohen, Aleksander Madry, Piotr Sankowski, and Adrian Vladu.
\newblock Negative-weight shortest paths and unit capacity minimum cost flow in
  $\tilde{O}(m^{10/7} \log w)$ time: (extended abstract).
\newblock In {\em Proceedings of the Twenty-Eighth Annual ACM-SIAM Symposium on
  Discrete Algorithms}, SODA '17, pages 752--771, Philadelphia, PA, USA, 2017.
  Society for Industrial and Applied Mathematics.

\bibitem[CS05]{ChanS05:book}
Tony Chan and Jianhong Shen.
\newblock {\em Image Processing and Analysis: Variational, {PDE}, Wavelet, and
  Stochastic Methods}.
\newblock Society for Industrial and Applied Mathematics, Philadelphia, PA,
  USA, 2005.

\bibitem[DS07]{DS07}
Samuel~I. Daitch and Daniel~A. Spielman.
\newblock Support-graph preconditioners for 2-dimensional trusses.
\newblock {\em CoRR}, abs/cs/0703119, 2007.

\bibitem[DS08]{DaitchS08}
Samuel~I. Daitch and Daniel~A. Spielman.
\newblock Faster approximate lossy generalized flow via interior point
  algorithms.
\newblock In {\em Proceedings of the 40th annual ACM symposium on Theory of
  computing}, STOC '08, pages 451--460, New York, NY, USA, 2008. ACM.
\newblock Available at http://arxiv.org/abs/0803.0988.

\bibitem[Fle00]{Fleischer00}
Lisa~K. Fleischer.
\newblock Approximating fractional multicommodity flow independent of the
  number of commodities.
\newblock {\em SIAM Journal on Discrete Mathematics}, 13:505--520, 2000.

\bibitem[GK98]{GargK96}
Naveen Garg and Jochen K\"{o}nemann.
\newblock Faster and simpler algorithms for multicommodity flow and other
  fractional packing problems.
\newblock In {\em In Proceedings of the 39th Annual Symposium on Foundations of
  Computer Science}, pages 300--309, 1998.

\bibitem[GY04]{GoldfarbY04}
Donald Goldfarb and Wotao Yin.
\newblock Second-order cone programming methods for total variation-based image
  restoration.
\newblock {\em SIAM J. Sci. Comput}, 27:622--645, 2004.

\bibitem[GY05]{GY05}
Donald Goldfarb and Wotao Yin.
\newblock Second-order cone programming methods for total variation-based image
  restoration.
\newblock {\em SIAM Journal on Scientific Computing}, 27(2):622--645, 2005.

\bibitem[KLOS14]{KelnerLOS14}
Jonathan~A. Kelner, Yin~Tat Lee, Lorenzo Orecchia, and Aaron Sidford.
\newblock An almost-linear-time algorithm for approximate max flow in
  undirected graphs, and its multicommodity generalizations.
\newblock In {\em Proceedings of the Twenty-fifth Annual ACM-SIAM Symposium on
  Discrete Algorithms}, SODA '14, pages 217--226, Philadelphia, PA, USA, 2014.
  Society for Industrial and Applied Mathematics.

\bibitem[KLP{\etalchar{+}}16]{KPSS16}
Rasmus Kyng, Yin~Tat Lee, Richard Peng, Sushant Sachdeva, and Daniel~A.
  Spielman.
\newblock Sparsified cholesky and multigrid solvers for connection laplacians.
\newblock In {\em Proceedings of the Forty-eighth Annual ACM Symposium on
  Theory of Computing}, STOC '16, pages 842--850, New York, NY, USA, 2016. ACM.

\bibitem[KMP12]{KelnerMP12}
Jonathan~A. Kelner, Gary~L. Miller, and Richard Peng.
\newblock Faster approximate multicommodity flow using quadratically coupled
  flows.
\newblock In {\em Proceedings of the 44th symposium on Theory of Computing},
  STOC '12, pages 1--18, New York, NY, USA, 2012. ACM.
\newblock Available at http://arxiv.org/abs/1202.3367.

\bibitem[KMT11]{KoutisMT11}
Ioannis Koutis, Gary~L Miller, and David Tolliver.
\newblock Combinatorial preconditioners and multilevel solvers for problems in
  computer vision and image processing.
\newblock {\em Computer Vision and Image Understanding}, 115(12):1638--1646,
  2011.

\bibitem[LMP{\etalchar{+}}91]{LeightonMPSTT91}
Tom Leighton, Fillia Makedon, Serge Plotkin, Clifford Stein, Eva Tardos, and
  Spyros Tragoudas.
\newblock Fast approximation algorithms for multicommodity flow problems.
\newblock In {\em JOURNAL OF COMPUTER AND SYSTEM SCIENCES}, pages 487--496,
  1991.

\bibitem[LS14]{LeeS14}
Yin~Tat Lee and Aaron Sidford.
\newblock Path finding methods for linear programming: Solving linear programs
  in $\tilde{O}(\sqrt{rank})$ iterations and faster algorithms for maximum
  flow.
\newblock In {\em Foundations of Computer Science (FOCS), 2014 IEEE 55th Annual
  Symposium on}, pages 424--433. IEEE, 2014.
\newblock Available at http://arxiv.org/abs/1312.6677 and
  http://arxiv.org/abs/1312.6713.

\bibitem[LS15]{LeeS15}
Yin~Tat Lee and Aaron Sidford.
\newblock Efficient inverse maintenance and faster algorithms for linear
  programming.
\newblock In {\em Foundations of Computer Science (FOCS), 2015 IEEE 56th Annual
  Symposium on}, pages 230--249. IEEE, 2015.
\newblock Available at: https://arxiv.org/abs/1503.01752.

\bibitem[Mad10]{Madry10a}
Aleksander Madry.
\newblock Faster approximation schemes for fractional multicommodity flow
  problems via dynamic graph algorithms.
\newblock In {\em STOC '10: Proceedings of the 42nd ACM symposium on Theory of
  computing}, pages 121--130, New York, NY, USA, 2010. ACM.

\bibitem[Mad13]{Mad13}
Aleksander Madry.
\newblock Navigating central path with electrical flows: From flows to
  matchings, and back.
\newblock In {\em Proceedings of the 2013 IEEE 54th Annual Symposium on
  Foundations of Computer Science}, FOCS '13, pages 253--262, Washington, DC,
  USA, 2013. IEEE Computer Society.

\bibitem[Mad16]{Mad16}
Aleksander Madry.
\newblock Computing maximum flow with augmenting electrical flows.
\newblock In {\em Proceedings of the 2016 IEEE 57th Annual Symposium on
  Foundations of Computer Science}, FOCS '16, Washington, DC, USA, 2016. IEEE
  Computer Society.

\bibitem[MTW14]{marchesini2014alternating}
Stefano Marchesini, Yu-Chao Tu, and Hau-tieng Wu.
\newblock Alternating projection, ptychographic imaging and phase
  synchronization.
\newblock {\em arXiv preprint arXiv:1402.0550}, 2014.

\bibitem[Nem]{Nemirovski04}
Arkadi Nemirovski.
\newblock Interior point polynomial time methods in convex programming.

\bibitem[OSB15]{stable2015}
Onur Özyeşil, Amit Singer, and Ronen Basri.
\newblock Stable camera motion estimation using convex programming.
\newblock {\em SIAM Journal on Imaging Sciences}, 8(2):1220--1262, 2015.

\bibitem[Pen16]{Peng16}
Richard Peng.
\newblock Approximate undirected maximum flows in $o(mpoly\log(n))$ time.
\newblock In {\em Proceedings of the Twenty-seventh Annual ACM-SIAM Symposium
  on Discrete Algorithms}, SODA '16, pages 1862--1867, Philadelphia, PA, USA,
  2016. Society for Industrial and Applied Mathematics.

\bibitem[Saa03]{Saad03:book}
Y.~Saad.
\newblock {\em Iterative Methods for Sparse Linear Systems}.
\newblock Society for Industrial and Applied Mathematics, Philadelphia, PA,
  USA, 2nd edition, 2003.
\newblock Available at http://www-users.cs.umn.edu/\textasciitilde
  saad/toc.pdf.

\bibitem[She13]{Sherman13}
Jonah Sherman.
\newblock Nearly maximum flows in nearly linear time.
\newblock In {\em Proceedings of the 2013 IEEE 54th Annual Symposium on
  Foundations of Computer Science}, FOCS '13, pages 263--269, Washington, DC,
  USA, 2013. IEEE Computer Society.

\bibitem[Spi07]{Spielman07:survey}
Daniel~A Spielman.
\newblock Spectral graph theory and its applications.
\newblock In {\em Foundations of Computer Science, 2007. FOCS'07. 48th Annual
  IEEE Symposium on}, pages 29--38. IEEE, 2007.

\bibitem[Spi16]{Spi16}
Daniel~A. Spielman.
\newblock Nsf award 1562041: Generalized algebraic graph theory: Algorithms and
  analysis.
\newblock {\em ALGORITHMIC FOUNDATIONS}, 2016.

\bibitem[SS11]{singer2011three}
Amit Singer and Yoel Shkolnisky.
\newblock Three-dimensional structure determination from common lines in
  cryo-em by eigenvectors and semidefinite programming.
\newblock {\em SIAM journal on imaging sciences}, 4(2):543--572, 2011.

\bibitem[SS12]{shkolnisky2012viewing}
Yoel Shkolnisky and Amit Singer.
\newblock Viewing direction estimation in cryo-em using synchronization.
\newblock {\em SIAM journal on imaging sciences}, 5(3):1088--1110, 2012.

\bibitem[ST08]{ShklarskiT08}
Gil Shklarski and Sivan Toledo.
\newblock Rigidity in finite-element matrices: Sufficient conditions for the
  rigidity of structures and substructures.
\newblock {\em {SIAM} J. Matrix Analysis Applications}, 30(1):7--40, 2008.

\bibitem[ST14]{SpielmanTengSolver:journal}
D.~Spielman and S.~Teng.
\newblock Nearly linear time algorithms for preconditioning and solving
  symmetric, diagonally dominant linear systems.
\newblock {\em SIAM Journal on Matrix Analysis and Applications},
  35(3):835--885, 2014.
\newblock Available at http://arxiv.org/abs/cs/0607105.

\bibitem[Str69]{Str69}
Volker Strassen.
\newblock Gaussian elimination is not optimal.
\newblock {\em Numer. Math.}, 13(4):354--356, August 1969.

\bibitem[SZS{\etalchar{+}}08]{SzeliskiZSVKATR08}
Richard Szeliski, Ramin Zabih, Daniel Scharstein, Olga Veksler, Vladimir
  Kolmogorov, Aseem Agarwala, Marshall Tappen, and Carsten Rother.
\newblock A comparative study of energy minimization methods for markov random
  fields with smoothness-based priors.
\newblock {\em IEEE Transactions on Pattern Analysis and Machine Intelligence},
  30:1068--1080, 2008.

\bibitem[TBI97]{trefethen1997numerical}
Lloyd~N Trefethen and David Bau~III.
\newblock {\em Numerical linear algebra}, volume~50.
\newblock Siam, 1997.

\bibitem[Ten10]{Teng10:survey}
Shang-Hua Teng.
\newblock {The Laplacian Paradigm: Emerging Algorithms for Massive Graphs}.
\newblock In {\em {Theory and Applications of Models of Computation}}, pages
  2--14, 2010.

\bibitem[Vai89]{Vaidya89}
P.~M. Vaidya.
\newblock Speeding-up linear programming using fast matrix multiplication.
\newblock In {\em Proceedings of the 30th Annual Symposium on Foundations of
  Computer Science}, pages 332--337, Washington, DC, USA, 1989. IEEE Computer
  Society.

\bibitem[Wil12]{Wil12}
Virginia~Vassilevska Williams.
\newblock Multiplying matrices faster than coppersmith-winograd.
\newblock In {\em Proceedings of the Forty-fourth Annual ACM Symposium on
  Theory of Computing}, STOC '12, pages 887--898, New York, NY, USA, 2012. ACM.

\bibitem[Woo14]{Woodruff14}
David~P Woodruff.
\newblock Sketching as a tool for numerical linear algebra.
\newblock {\em Theoretical Computer Science}, 10(1-2):1--157, 2014.
\newblock Available at http://arxiv.org/abs/1411.4357.

\bibitem[WR07]{WohlbergRodriguez07}
B.~Wohlberg and P.~Rodriguez.
\newblock An iteratively reweighted norm algorithm for minimization of total
  variation functionals.
\newblock {\em Signal Processing Letters, IEEE}, 14(12):948 --951, dec. 2007.

\bibitem[Wri97]{Wright97}
Stephen~J Wright.
\newblock {\em Primal-dual interior-point methods}.
\newblock SIAM, 1997.

\bibitem[WW10]{WilW10}
Virginia~Vassilevska Williams and Ryan Williams.
\newblock Subcubic equivalences between path, matrix and triangle problems.
\newblock In {\em Proceedings of the 2010 IEEE 51st Annual Symposium on
  Foundations of Computer Science}, FOCS '10, pages 645--654, Washington, DC,
  USA, 2010. IEEE Computer Society.

\bibitem[WYYZ08]{WangYYZ08}
Yilun Wang, Junfeng Yang, Wotao Yin, and Yin Zhang.
\newblock A new alternating minimization algorithm for total variation image
  reconstruction.
\newblock {\em SIAM Journal on Imaging Sciences}, 1(3):248--272, 2008.

\bibitem[Ye11]{Ye97:book}
Yinyu Ye.
\newblock {\em Interior point algorithms: theory and analysis}, volume~44.
\newblock John Wiley \& Sons, 2011.
\newblock Available at: http://web.stanford.edu/\textasciitilde yyye/main.ps.

\bibitem[ZS14]{zhao2014rotationally}
Zhizhen Zhao and Amit Singer.
\newblock Rotationally invariant image representation for viewing direction
  classification in cryo-em.
\newblock {\em Journal of structural biology}, 186(1):153--166, 2014.

\end{thebibliography}

\end{document}